\documentclass[sigconf, nonacm]{acmart}

\usepackage{pvldb}

\usepackage{graphicx}
\usepackage{subcaption}
\usepackage{bm}
\usepackage{amsmath}
\usepackage{amsthm}
\usepackage{mathtools}
\usepackage{algorithmicx}
\usepackage{algorithm}
\usepackage[noend]{algpseudocode}
\algnewcommand\Input{\item[\textbf{Input:}]}
\algnewcommand\Output{\item[\textbf{Output:}]}
\usepackage[capitalize,noabbrev]{cleveref}
\usepackage{booktabs}
\usepackage{makecell}

\newtheorem{definition}{Definition}

\newtheorem{lemma}{Lemma}
\newtheorem{theorem}{Theorem}
\newtheorem{corollary}{Corollary}

\DeclareMathOperator*{\argmin}{argmin}
\DeclareMathOperator*{\argmax}{argmax}

\usepackage[most]{tcolorbox}

\newtcolorbox{definition_box}{
    colback=white,
    colframe=black,
    boxrule=0.5pt,
    arc=0pt,
    breakable,
    left=1.5mm, right=1.5mm, top=1.0mm, bottom=1.0mm,
    before skip=4pt, after skip=4pt
}

\newtcolorbox{theorem_box}{
    colback=black!9!white,
    colframe=black!9!white,
    boxrule=0pt,
    arc=0pt,
    breakable,
    left=1.5mm, right=1.5mm, top=1.0mm, bottom=1.0mm,
    before skip=4pt, after skip=4pt
}

\usepackage{enumitem}

\renewcommand\vldbavailabilityurl{https://github.com/atsukisato/pgm-attack}

\begin{document}
\title{Poisoning Attacks on the PGM-index}

\author{Atsuki Sato}
\orcid{0009-0001-5366-4842}
\email{a\_sato@hal.t.u-tokyo.ac.jp}
\affiliation{%
    \institution{Graduate School of Information Science and Technology, The University of Tokyo}
    \city{Tokyo}
    \country{Japan}
}

\author{Martin Aum\"uller}
\email{maau@itu.dk}
\affiliation{%
    \institution{Algorithms Group, IT University of Copenhagen}
    \city{Copenhagen}
    \country{Denmark}
}

\author{Yusuke Matsui}
\email{matsui@hal.t.u-tokyo.ac.jp}
\affiliation{%
    \institution{Graduate School of Information Science and Technology, The University of Tokyo}
    \city{Tokyo}
    \country{Japan}
}

\begin{abstract}
The PGM-index (Ferragina and Vinciguerra, VLDB'20) is one of the most practical learned indexes, owing to its theoretical elegance and consistently strong empirical performance.
It is built on optimal piecewise linear approximations (PLAs) that minimize the number of segments.
In this paper, we ask how sensitive this optimal PLA itself is to poisoning attacks.
We propose \textsc{PGM-attack}, an efficient poisoning attack that sequentially inserts adversarial keys to inflate the resulting number of segments, and we develop a method for deriving theoretical upper bounds on the number of segments attainable under arbitrary insertions.
Our experiments show that poisoning only 10\% of the keys allows \textsc{PGM-attack} to increase the segment count by up to $120\times$.
On every evaluated instance, our instance-dependent upper bound is at most $1.92\times$ the segment count attained by \textsc{PGM-attack}, certifying that \textsc{PGM-attack} achieves at least 52\% of the optimum.
This increase in the number of segments enlarges the PGM-index by up to $120\times$.
Moreover, the attack also transfers to other learned indexes, substantially inflating the index size of PLA-based ones in particular.
Our results reveal that, despite the optimality of its PLAs, the PGM-index has an intrinsic vulnerability rooted in its optimization objective, motivating robustness-aware objective design for future learned indexes.
\end{abstract}

\maketitle


\begin{figure}[t]
    \centering
    \includegraphics[width=\columnwidth]{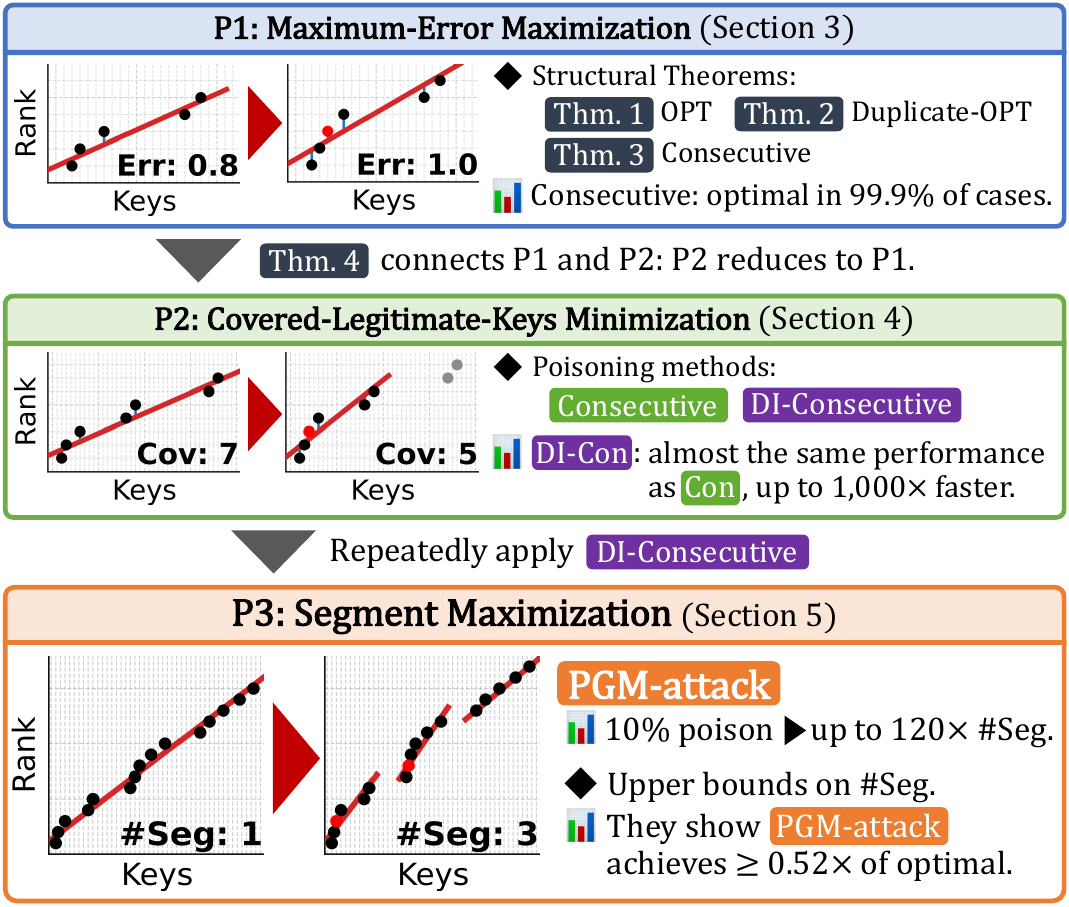}
    \caption{Our bottom-up approach progresses from maximum-error maximization (\textbf{P1}), through covered-key minimization (\textbf{P2}), to segment maximization (\textbf{P3}). The \textbf{P3} panel illustrates how our \textsc{PGM-attack} can effectively degrade a PLA: before poisoning, a single segment suffices to approximate the ranks of all keys within error $\varepsilon = 1$, whereas inserting only two poison keys increases the required number of segments to three.}
    \label{fig:overview}
\end{figure}

\section{Introduction}
\label{sec:introduction}

In recent years, learned indexes~\cite{kraska2018case_li} have attracted considerable attention as space-efficient, high-performance alternatives to traditional index structures such as B+-trees.
The PGM-index~\cite{ferragina2020pgm} is a leading example.
It recursively constructs \textit{optimal piecewise linear approximations (PLAs)} of the cumulative distribution function (CDF), where an optimal PLA is one that uses the minimum number of segments.
This optimality allows the PGM-index to provide worst-case theoretical guarantees while achieving a favorable trade-off between memory usage and query time.
These properties make the PGM-index an appealing choice for large-scale settings where both practical efficiency and worst-case guarantees matter, and it has begun to be adopted in high-performance database systems and search engines~\cite{manticore_manual_introduction,infiniflow_medium_infinity_010}.

As learned indexes become increasingly practical, understanding their robustness becomes correspondingly important.
By quantifying how much an index's memory efficiency and query performance can deteriorate after the insertion of a small amount of data, we can assess its stability and deploy it with greater confidence.
For traditional indexes, such effects are often straightforward to characterize: for example, a B+-tree requires $\Theta(n)$ space and supports queries in $\Theta(\log n)$ time, making the impact of a small increase in the number of keys readily predictable.
For learned indexes, however, the impact is far less clear because their performance can depend sensitively on the data distribution.
Motivated by this concern, recent studies have investigated attacks against several learned indexes, including poisoning attacks against mean-squared-error (MSE) minimizing linear models and Recursive Model Indexes (RMIs)~\cite{kraska2018case_li,kornaropoulos2022price,sato2026foundations}, as well as algorithmic-complexity attacks against ALEX~\cite{ding2020alex,yang2023algorithmic,schuster2025learned}.

We note that, although the PGM-index enjoys several theoretical guarantees, its robustness against poisoning (or, more generally, against data insertions) remains largely uncharacterized, which is an important concern in practical deployments.
For example, the PGM-index has a worst-case guarantee of provably superior performance to conventional B+-trees~\cite{ferragina2020pgm}, but this guarantee does not imply robustness against poisoning:
the strength of the PGM-index lies in the fact that it typically performs far better than B+-trees, and if a small number of poison keys can drive its performance toward this worst case, this strength is lost.
Moreover, in terms of sensitivity to attacks, i.e., the discrepancy between the performance before and after an attack, the PGM-index may even be more sensitive than a conventional B+-tree.
Such sensitivity poses a serious problem when deploying the PGM-index in real systems:
database administrators are often drawn to the PGM-index for its high performance, yet they would have to worry that its benefits might be largely eliminated by only a small number of key insertions.

In this paper, we propose \textsc{PGM-attack}, a poisoning attack against PLA-based learned indexes, including the PGM-index.
\textsc{PGM-attack} substantially increases the number of segments in the optimal PLA, which forms the bottom-level PLA of the PGM-index.
The key-rank plots at the bottom of \Cref{fig:overview} show an illustrative example; inserting only two poison keys increases the number of segments from one to three.
Because these bottom-level segments dominate the model size of the PGM-index, increasing their number reduces compression efficiency and can degrade query performance through more cache misses and larger upper-level structures.

Analyzing poisoning attacks against PLAs is challenging for several reasons.
First, this problem is difficult to analyze theoretically because even the basic building block of PLA construction (maximum-error-minimizing linear regression) does not admit a simple closed-form solution.
This is clearly distinct from the MSE-minimizing regression considered in prior work~\cite{kornaropoulos2022price,sato2026foundations}, which admits a closed-form solution.
Second, inserting a poison key shifts the ranks of all subsequent keys and can change the boundaries and sizes of later segments, causing cascading changes throughout the segmentation.
Prior attacks against RMIs~\cite{kornaropoulos2022price} do not need to account for such global effects because they assume that each linear regression model covers a fixed number of keys.
Finally, simple strategies such as inserting a consecutive run of keys, as used in algorithmic-complexity attacks against ALEX~\cite{yang2023algorithmic,schuster2025learned}, are ineffective at increasing the total number of PLA segments.
Such keys can typically be absorbed into a single PLA segment and therefore have little effect on the segment count.

To address these challenges, we adopt a bottom-up approach that proceeds from simpler subproblems to our final objective, as illustrated in \Cref{fig:overview}.
We first analyze \textbf{P1}, the problem of maximizing the maximum error of linear regression.
The resulting insights allow us to solve \textbf{P2}, which minimizes the number of legitimate keys covered by a single segment.
Finally, by repeatedly solving \textbf{P2}, we address \textbf{P3}, which maximizes the total number of segments in the PLA, thereby constructing \textsc{PGM-attack}.
We also propose algorithms for deriving upper bounds on the number of segments attainable under arbitrary key insertions.
Our experiments show that poisoning only 10\% of the keys can increase the number of segments by up to $120\times$.
These results reveal that, despite the optimality of its PLAs, the PGM-index has an intrinsic vulnerability rooted in its optimization objective.
This motivates the design of future learned indexes with more robustness-aware objectives.

\textbf{Contributions.}
Our contributions are summarized as follows:
\begin{itemize}[leftmargin=1.5em]
    \item We formulate three new poisoning problems on the CDF: \textbf{P1}, maximizing the maximum error in linear regression; \textbf{P2}, minimizing the number of covered legitimate keys in linear regression; and \textbf{P3}, maximizing the number of segments in the PLA.
    \item For \textbf{P1}, we characterize optimal poison sets (\cref{sec:poisoning_max_error}) and show experimentally that appropriately selected consecutive-integer poison keys are almost always optimal (\cref{sec:experiment_poisoning_max_error}).
    \item For \textbf{P2}, we prove a reduction to \textbf{P1} and develop a near-linear-time algorithm (\cref{sec:poisoning_num_coverable_keys}). We experimentally show that this algorithm generates near-optimal consecutive-integer poison keys (\cref{sec:experiment_poisoning_num_coverable_keys}).
    \item For \textbf{P3}, we propose \textsc{PGM-attack}, which repeatedly generates candidate poison sets for \textbf{P2} and sequentially selects among them (\cref{sec:method_for_obtaining_poison_solutions_m_opt}). Injecting 10\% poison keys increases the segment count and PGM-index size by up to $120\times$. The attack also transfers to other learned indexes, particularly inflating the size of PLA-based ones (\cref{sec:experiment_poisoning_segment_number,sec:experiment_poisoning_index_performance}).
    \item For \textbf{P3}, we also propose algorithms for upper-bounding the number of segments attainable under arbitrary key insertions (\cref{sec:m_opt_instance_agnostic_upper_bound,sec:m_opt_instance_dependent_upper_bound}). On all evaluated instances, the resulting upper bound is at most $1.92\times$ the segment count obtained by \textsc{PGM-attack}, certifying that the attack achieves at least 52\% of the optimum on these instances (\cref{sec:experiment_poisoning_segment_number}).
\end{itemize}

\section{Preliminaries}
\label{sec:preliminaries}

We consider the \emph{indexing problem}, which requires building a data structure that stores a multiset $\mathcal{X}$ and supports operations such as \emph{member} (given an element $x$, is $x \in \mathcal{X}$?), \emph{predecessor} (given an element $x$, output the largest $y \leq x$ such that $y \in \mathcal{X}$), and \emph{insert} (add $x$ to $\mathcal{X}$).
A simple but efficient way to implement such a data structure is to store the elements of $\mathcal{X}$ in a sorted array $A$.
Then, all operations boil down to efficiently finding the \emph{rank} of an element $x$, i.e., the number of elements in $\mathcal{X}$ that are smaller than $x$~\cite{ferragina2020pgm}.

The idea of a \emph{learned index} as popularized by Kraska et al.~\cite{kraska2018case_li} is that the index employs a \emph{model} that learns the rank function based on the data.
In other words, the model learns the cumulative distribution function (CDF) of the given dataset.
Given an element $x$, the model predicts the rank of $x$ and then uses a second verification step, for example using binary search, to compute the exact answer.

\subsection{PLA Algorithm}
Piecewise linear approximation (PLA) is a central building block of learned indexes: it underlies not only the PGM-index~\cite{ferragina2020pgm} but also a broad range of designs~\cite{galakatos2019fiting,kipf2020radixspline}, including recent state-of-the-art learned indexes~\cite{liu2024learned,leying2026line}.
A PLA algorithm approximates the rank function by a sequence of linear segments: given a \emph{maximum allowed error parameter} $\varepsilon$, it repeatedly extends the current segment as far as possible while keeping the maximum error at most $\varepsilon$, and then starts a new segment.
The PLA construction problem (formally defined in \cref{sec:poisoning_m_opt}) thus builds on the segment extension problem (\cref{sec:poisoning_num_coverable_keys}), which in turn builds on the problem of minimizing the maximum error of a single segment (\cref{sec:poisoning_max_error}).

PLA algorithms come in optimal variants that minimize the number of segments~\cite{o1981line} and faster variants that construct near-optimal approximations~\cite{elmeleegy2009online}.
Since the PGM-index adopts an optimal PLA algorithm, we focus on this setting.
Accordingly, although the PGM-index serves as our primary application, our study targets the PLA algorithm itself rather than any PGM-specific implementation.

\subsection{Importance of the PLA Segment Count}
Although the PGM-index has a recursive multi-level structure, its space and query-time complexity are dominated by the bottom-level PLA over the CDF~\cite{ferragina2020pgm}: for a dataset of $n$ keys with optimal segment count $m_\mathrm{opt}$, the index requires $O(m_\mathrm{opt})$ space and $O(\log m_\mathrm{opt} + \log \varepsilon)$ query time.
Attacks that increase $m_\mathrm{opt}$ therefore directly degrade the index's practical performance.
Our theoretical analysis focuses on this fundamental PLA component, and our experiments (\Cref{sec:experiment}) validate its end-to-end impact on the PGM-index and other learned indexes.

\subsection{Threat Model}
Following prior work~\cite{kornaropoulos2022price,sato2026foundations}, we adopt a white-box setting in which the attacker knows all legitimate keys and can inject poison keys before the index is constructed; for the problems studied in \cref{sec:poisoning_num_coverable_keys,sec:poisoning_m_opt}, the attacker also knows the hyperparameter $\varepsilon$.
This is the standard starting point in attack studies~\cite{kornaropoulos2022price,sato2026foundations}: it isolates the inherent sensitivity of the PLA construction objective from deployment-specific restrictions and provides a basis for weaker-knowledge attacks.
For example, a black-box attacker could first estimate the data distribution or $\varepsilon$ from limited observations and then apply a white-box attack to the resulting surrogate~\cite{yang2023algorithmic}.
We also follow prior work~\cite{kornaropoulos2022price,sato2026foundations} by restricting poison keys to lie strictly between the minimum and maximum legitimate keys.
This restriction conservatively limits the attacker's capabilities by excluding poison keys outside the legitimate data range, which could be easily detected and removed as outliers.
We show that even under this restriction, the attacker can substantially increase the number of segments.

\section{Maximum-Error Maximization}
\label{sec:poisoning_max_error}

\begin{figure*}[t]
    \centering
    \begin{subfigure}[b]{0.24\textwidth}
        \centering
        \includegraphics[width=\textwidth]{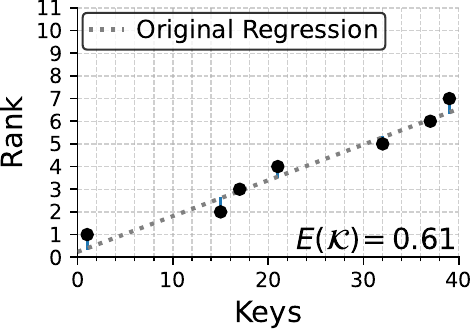}
        \caption{Regression Before Poisoning.}
        \label{fig:maxerror_poisoning_a}
    \end{subfigure}
    \begin{subfigure}[b]{0.24\textwidth}
        \centering
        \includegraphics[width=\textwidth]{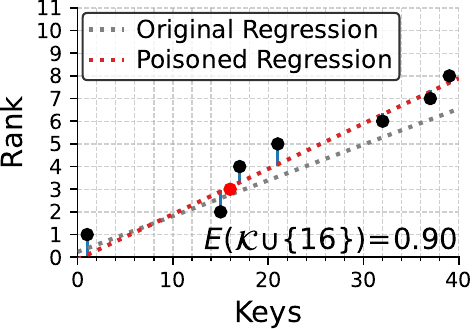}
        \caption{Optimal 1-point Poisoning.}
        \label{fig:maxerror_poisoning_b}
    \end{subfigure}
    \begin{subfigure}[b]{0.24\textwidth}
        \centering
        \includegraphics[width=\textwidth]{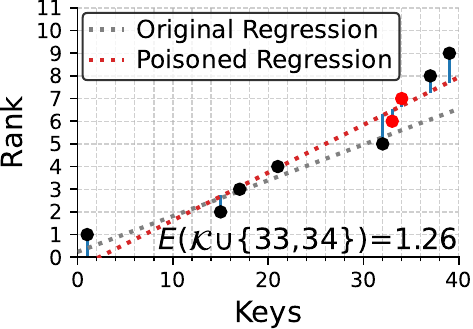}
        \caption{Optimal 2-point Poisoning.}
        \label{fig:maxerror_poisoning_c}
    \end{subfigure}
    \begin{subfigure}[b]{0.24\textwidth}
        \centering
        \includegraphics[width=\textwidth]{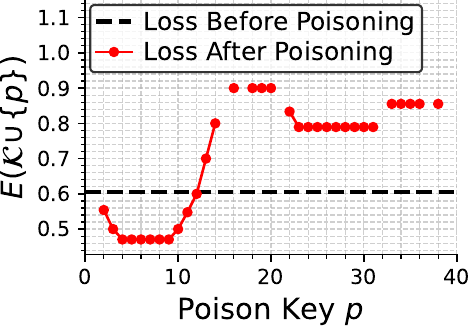}
        \caption{Loss under 1-point poisoning.}
        \label{fig:maxerror_poisoning_d}
    \end{subfigure}
    \caption{Poisoning attack on $\mathcal{K} = \{1, 15, 17, 21, 32, 37, 39\}$ under the maximum-error objective.
    By inserting poison points, the rank of each key greater than a poison increases.
    The linear regression model is then fitted on the poisoned dataset, leading to a change in the maximum error.
    Note that the optimal 2-point attack $\{33, 34\}$ does not contain the optimal 1-point attack.}
    \label{fig:maxerror_poisoning}
\end{figure*}

\subsection{Problem Setting}
We refer to the keys present in the original (pre-attack) training data as legitimate keys and denote the set of legitimate keys by $\mathcal{K}$, with $n \coloneq |\mathcal{K}|$.
We use $\mathcal{X}$ to denote a general key set that may include both legitimate and poison keys, and let $N \coloneq |\mathcal{X}|$.
For a positive integer $m$, we write $[m] \coloneq \{1,2,\dots,m\}$.

We begin by defining linear regression on CDFs.
\begin{definition_box}
\begin{definition}[\textbf{Linear Regression on CDFs (Maximum Error)}]
\label{def:linear_regression_on_cdfs_max_error}
Let $\mathcal{X} = \{x_1, x_2, \dots, x_N\}$ be a multiset of natural numbers such that $x_1 \leq x_2 \leq \dots \leq x_N$.
For each $i \in [N]$, define the rank of $x_i$ as $r_i \coloneq i$.
The maximum error of linear regression on CDFs for $\mathcal{X}$ is defined by
\begin{equation}
    E(\mathcal{X}) \coloneq \min_{w,b} \max_{i \in [N]} |wx_i+b-r_i|.
\end{equation}
\end{definition}
\end{definition_box}
Our definition differs from prior work~\cite{kornaropoulos2022price, sato2026foundations} only in the choice of loss function:
instead of minimizing the MSE~\cite[Eq. 1--2]{sato2026foundations},
\begin{equation}
\min_{w,b} \frac{1}{N}\sum_{i\in[N]}(wx_i + b - r_i)^2,
\end{equation}
we minimize the maximum error, which enables a more direct connection to PLA algorithms.
Computing $E(\mathcal{X})$ can be framed as a three-dimensional linear programming problem~\cite{boyd2004convex} and solved in $\mathcal{O}(N)$ time via Megiddo's algorithm~\cite{megiddo1984linear}; unlike MSE, no simple closed-form solution is known.

We now define the corresponding poisoning problem (again identical to the MSE-based formulation~\cite{kornaropoulos2022price,sato2026foundations} except for the objective).
\begin{definition_box}
\begin{definition}[\textbf{Maximum-Error Maximization Problem}]
\label{def:poisoning_linear_regression_on_cdfs_max_error}
Let $\mathcal{K} = \{k_1, k_2, \dots, k_n\} \subset \mathbb{N}$ be $n$ distinct legitimate keys such that
$k_1 < k_2 < \dots < k_n$, and let $\lambda \in \mathbb{N}$.
The \emph{maximum-error maximization problem} is to choose up to $\lambda$ integers from $\{k_1, k_1+1, \dots, k_n\} \setminus \mathcal{K}$ so as to maximize the maximum error:
\begin{equation}
\argmax_{\mathcal{P} ~ \mathrm{s.t.} ~ |\mathcal{P}| \leq \lambda,\;
\mathcal{P} \subseteq \{k_1, k_1+1, \dots, k_n\} \setminus \mathcal{K}}
E(\mathcal{K} \cup \mathcal{P}).
\end{equation}
\end{definition}
\end{definition_box}
\textbf{Connection to the PGM-index.}
This problem is fundamental to attacking the PGM-index.
The PLA algorithm extends each segment as far as possible while keeping the maximum error within $\varepsilon$.
Therefore, increasing $m_\mathrm{opt}$ amounts to forcing the maximum error to exceed $\varepsilon$ earlier, thereby shortening segments.
We revisit this connection in \cref{sec:poisoning_num_coverable_keys,sec:poisoning_m_opt}.

\textbf{Difficulty of the Problem.}
This problem involves several intrinsic challenges.
\textbf{Rank shifting:}
As in the MSE setting, inserting a poison key increases the rank of every subsequent key and poison by one.
This global interaction makes the optimal choice of poison keys highly non-trivial.
\textbf{Restricted integer choices:}
Integers already used as legitimate keys or previously selected poison keys cannot be reused.
This discrete feasibility constraint complicates analytical reasoning.
\textbf{No closed-form solution:}
Unlike ordinary least squares under MSE, which admits a closed-form solution, the maximum-error regression problem has no known closed-form characterization.
This significantly complicates theoretical analysis.

\textbf{Illustrative Example.}
\cref{fig:maxerror_poisoning} illustrates an example of a maximum-error poisoning attack on $\mathcal{K} = \{1, 15, 17, 21, 32, 37, 39\}$.
As in the MSE setting, inserting a poison key shifts the ranks of all subsequent keys and poisons by one.
For example, comparing \cref{fig:maxerror_poisoning_a,fig:maxerror_poisoning_b}, the insertion of a poison does not affect the ranks of keys less than the poison, while it increases the rank of each key greater than the poison by one.
\cref{fig:maxerror_poisoning_d} plots $E(\mathcal{K} \cup \{p\})$ for each candidate poison $p$.
Integers already present in $\mathcal{K}$ are excluded from the candidate set.
In each interval, the poison that maximizes the maximum error is adjacent to a legitimate key.
This structural property is formally established later in \cref{cor:single_point_attack}.

\textbf{Duplicate-Allowed Setting.}
We also introduce a relaxed setting in which duplicate values are allowed.
In this variant, poison keys may coincide with legitimate keys, and the same integer may be selected multiple times as poison.
The optimal value in this relaxed setting provides an upper bound for the original (duplicate-forbidden) problem.

\subsection{The Structure of Optimal Attacks}
We prove that there exists an optimal solution to the maximum-error poisoning problem (\cref{def:poisoning_linear_regression_on_cdfs_max_error}) in which every poison key is adjacent to a legitimate key, either directly or through a chain of adjacent poison keys.
Formally, we have the following theorem.
\begin{theorem_box}
\begin{theorem}
\label{thm:maxerror_attack_structure}
There exists an optimal solution $\mathcal{P}^\ast$ to the maximum-error poisoning problem
(\cref{def:poisoning_linear_regression_on_cdfs_max_error}) satisfying
\begin{equation}
    \forall p \in \mathcal{P}^\ast,~ \exists k \in \mathcal{K}, ~ \{i \in \mathbb{N} \mid \min(p,k) < i < \max(p,k)\} \subset \mathcal{P}^\ast.
\end{equation}
\end{theorem}
\end{theorem_box}

To prove \cref{thm:maxerror_attack_structure}, we first establish the following lemma, which shows that there exists at least one direction (either left or right) in which shifting a contiguous block of keys does not decrease the maximum error.
\begin{theorem_box}
\begin{lemma}
\label{lem:block_moving}
Let $\mathcal{X} = \{x_1, x_2, \dots, x_N\}$ be a multiset with $x_1 \leq x_2 \leq \dots \leq x_N$, and let $1 < l \leq r < N$.
Define
\begin{equation}
    \mathcal{A} \coloneq \{x_i\}_{i=1}^{l-1} \uplus \{x_i\}_{i=r+1}^{N},\quad
    \mathcal{B}(\delta) \coloneq \{x_i + \delta\}_{i=l}^{r},
\end{equation}
for $\delta \in [x_{l-1} - x_l,\, x_{r+1} - x_r]$, and let $\mathcal{B} \coloneq \mathcal{B}(0)$.
Then, for any $\delta_{-} \in [x_{l-1} - x_l, 0]$ and $\delta_{+} \in [0, x_{r+1} - x_r]$,
\begin{equation}
\label{eq:block_moving_condition}
    E(\mathcal{A}\uplus\mathcal{B}) \leq \max(E(\mathcal{A}\uplus\mathcal{B}(\delta_{-})), E(\mathcal{A}\uplus\mathcal{B}(\delta_{+}))),
\end{equation}
where $\uplus$ denotes multiset addition.
\end{lemma}
\end{theorem_box}
\begin{proof}[Proof Sketch of \cref{lem:block_moving}]
Fix $\varepsilon \geq 0$ and consider the feasible sets of line parameters.
By reinterpreting the block shift in a suitably transformed parameter space, we can show that
\begin{equation}
\label{eq:block_moving_condition_epsilon}
    E(\mathcal{A}\uplus\mathcal{B}(\delta_{-})) \leq \varepsilon ~\land~
    E(\mathcal{A}\uplus\mathcal{B}(\delta_{+})) \leq \varepsilon
    \;\Rightarrow\;
    E(\mathcal{A}\uplus\mathcal{B}) \leq \varepsilon,
\end{equation}
and $\varepsilon = \max\left(E(\mathcal{A}\uplus\mathcal{B}(\delta_{-})),\, E(\mathcal{A}\uplus\mathcal{B}(\delta_{+}))\right)$ yields the claim.
\end{proof}

\begin{figure}[t]
    \centering
    \includegraphics[width=\columnwidth]{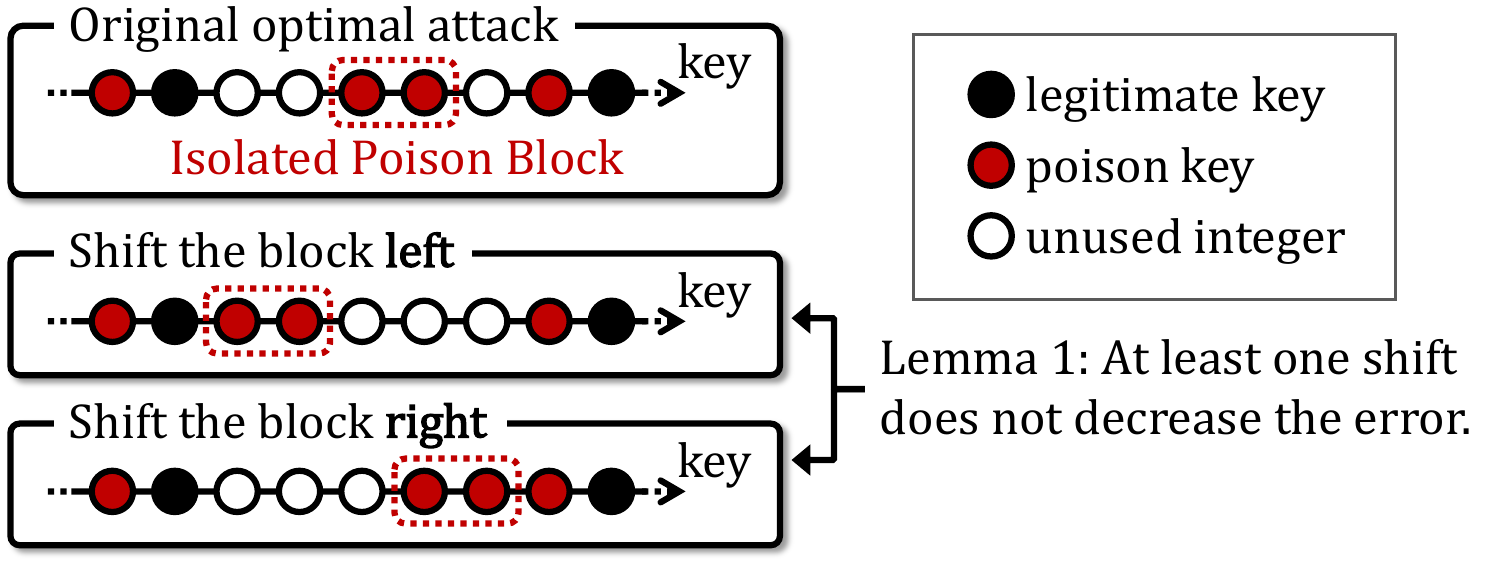}
    \caption{Illustration of the proof of \cref{thm:maxerror_attack_structure}.
    By \cref{lem:block_moving}, shifting an isolated poison block reduces the number of isolated blocks without decreasing the maximum error.}
    \label{fig:proof_maxerror_theorem}
\end{figure}

Using \cref{lem:block_moving}, we can prove \cref{thm:maxerror_attack_structure}.
\begin{proof}[Proof Sketch of \cref{thm:maxerror_attack_structure}]
Consider an optimal solution with the fewest isolated poison blocks.
By \cref{lem:block_moving}, any isolated block can be shifted to a legitimate key or another poison block without decreasing the error (see \cref{fig:proof_maxerror_theorem}).
This contradicts the minimality of the optimal solution, stating that such an optimal solution has no isolated poison blocks.
\end{proof}

\cref{thm:maxerror_attack_structure} implies that, to find an optimal solution, it suffices to consider poison sets in which every poison key is adjacent to a legitimate key, either directly or via a chain of adjacent poison keys.
The number of such candidates is at most $\binom{2n-3+\lambda}{\lambda}$; without this structural result, the number of possible poison sets is roughly $\binom{k_n-k_1}{\lambda}$.
Thus, \cref{thm:maxerror_attack_structure} reduces the search space from depending on the key-universe size to depending on $n$, substantially reducing the cost of enumeration.

\textbf{Corollary: Single-Point Attack.}
From \cref{thm:maxerror_attack_structure}, we obtain the following corollary for the single-point attack:

\begin{theorem_box}
\begin{corollary}
\label{cor:single_point_attack}
For the single-point attack ($\lambda = 1$), there exists an optimal solution where the poison is directly adjacent to a legitimate key.
Formally, there exists an optimal $\mathcal{P}^\ast$ satisfying
\begin{equation}
    p^\ast \in \{k + 1 \mid k \in \mathcal{K}\} \cup \{k - 1 \mid k \in \mathcal{K}\}.
\end{equation}
\end{corollary}
\end{theorem_box}

By this corollary, we can find an optimal single-point attack by checking all integers adjacent to legitimate keys.
Since there are $\mathcal{O}(n)$ such candidates and each evaluation takes $\mathcal{O}(n)$ time, we can compute the optimal single-point poison in $\mathcal{O}(n^2)$ time.

Note, however, that greedily repeating the optimal single-point attack does not necessarily yield a globally optimal solution.
For example, in \cref{fig:maxerror_poisoning}, the optimal two-point attack is $\mathcal{P}=\{33,34\}$ (see \cref{fig:maxerror_poisoning_c}).
Neither $\mathcal{P}=\{33\}$ nor $\mathcal{P}=\{34\}$ is an optimal single-point attack, so the greedy repetition of optimal single-point attacks fails to recover the optimal two-point attack.
Note also that our theorems assert existence of an optimal adjacent solution, not that every optimum is adjacent; e.g., in \cref{fig:maxerror_poisoning}, $\{19\}$ is also optimal for the single-point attack.

\textbf{Duplicate-Allowed Setting.}
We now present a structural theorem for the duplicate-allowed setting.
\begin{theorem_box}
\begin{theorem}[Structure of an Optimal Attack in the Duplicate-Allowed Setting]
\label{thm:duplicate_allowed_structure}
In the duplicate-allowed setting, every optimal solution $\mathcal{P}^\ast$ to the maximum-error poisoning problem uses the full budget, i.e., $|\mathcal{P}^\ast|=\lambda$.
Moreover, there exists an optimal solution $\mathcal{P}^\ast$ such that all poison keys take the same value, and this value is a legitimate key in $\mathcal{K}$.
\end{theorem}
\end{theorem_box}
\begin{proof}[Proof Sketch of \cref{thm:duplicate_allowed_structure}]
Adding a poison key strictly increases the maximum error, so every optimal solution uses the full budget.
By \cref{lem:block_moving}, all poison keys can be moved onto legitimate keys.
Finally, poison mass at distinct legitimate keys can be merged at one of them without decreasing the error, yielding an optimal solution in which all poison keys equal a single legitimate key.
\end{proof}
This theorem suggests that concentrating poisons at a single location is a promising strategy.
Motivated by this theorem, in the next section we define a class of poison sets in which poisons are concentrated within a single contiguous region, and we propose an efficient algorithm for finding optimal solutions within this class.

This result provides a clear contrast to prior observations for the MSE setting.
Prior work conjectured that, in the duplicate-allowed MSE setting, an optimal attack can always be realized by concentrating poisons on at most three legitimate keys: the two extremes ($k_1$ and $k_n$) and one internal key~\cite[Conj.~1]{sato2026foundations}.
By contrast, our theorem shows that under the maximum-error objective, there always exists an optimal solution that concentrates all poisons on a single legitimate key.
This highlights a structural difference between the MSE and maximum-error objectives.

\subsection{Methods for Obtaining Poison Solutions}

Based on the theoretical results above, we propose four poisoning methods: three for the setting without duplicate keys (\textsc{Optimal}, \textsc{Greedy}, and \textsc{Consecutive}) and one for the setting that allows duplicates (\textsc{Duplicate-Opt}).
The main takeaway is that the \textsc{Consecutive} method offers the best overall trade-off: it is faster than the \textsc{Greedy} method and, as we show later in \cref{sec:experiment_poisoning_max_error}, typically attains the global optimum in the vast majority of cases.

\textbf{Optimal Method.}
By \cref{thm:maxerror_attack_structure}, we can restrict the candidate optimal solutions to
$\mathcal{O}\left(\binom{2n-3+\lambda}{\lambda}\right)$ possibilities.
Since evaluating each candidate takes $\mathcal{O}(n+\lambda)$ time, the \textsc{Optimal} method yields an exact optimal solution in $\mathcal{O}\left(\binom{2n-3+\lambda}{\lambda}(n+\lambda)\right)$ time.

\textbf{Greedy Method.}
The \textsc{Greedy} method repeatedly applies the optimal single-point attack.
At each iteration, we compute the optimal single poison for the current set $\mathcal{K}\cup\mathcal{P}$ and add it to $\mathcal{P}$.

The total time complexity is $\mathcal{O}(\lambda (n+\lambda)^2)$.
In each iteration, there are $\mathcal{O}(n+\lambda)$ candidates for the optimal single-point attack, and evaluating each candidate takes $\mathcal{O}(n+\lambda)$ time.
This process is repeated $\mathcal{O}(\lambda)$ times.

\textbf{Consecutive Method.}
Motivated by \cref{thm:duplicate_allowed_structure}, we consider solutions restricted to \emph{consecutive poisons}, defined as follows.
\begin{definition_box}
\begin{definition}
\label{def:consecutive_poisons}
A poison set $\mathcal{P}$ is called \emph{consecutive} if there exist integers $l,r$ with $k_1 < l \le r < k_n$ such that
\begin{equation}
    \mathcal{P} = \{l,l+1,\dots,r\}\setminus\mathcal{K}
    \quad \land \quad
    |\mathcal{P}|=\lambda.
\end{equation}
\end{definition}
\end{definition_box}

The \textsc{Consecutive} method searches for an optimal solution under this restriction.
The following theorem allows us to further reduce the candidate space.

\begin{theorem_box}
\begin{theorem}
\label{thm:consecutive_structure}
There exists an optimal consecutive solution whose left or right boundary is adjacent to a legitimate key.
Formally, there exists an optimal consecutive poison set $\mathcal{P}_{\mathrm{con}}^\ast$ such that
\begin{equation}
    \min(\mathcal{P}_{\mathrm{con}}^\ast)-1 \in \mathcal{K}
    \;\lor\;
    \max(\mathcal{P}_{\mathrm{con}}^\ast)+1 \in \mathcal{K}.
\end{equation}
\end{theorem}
\end{theorem_box}

\begin{proof}[Proof of \cref{thm:consecutive_structure}]
By \cref{lem:block_moving}, any consecutive poison block not adjacent to a legitimate key can be shifted without decreasing the maximum error until it becomes adjacent to one, yielding the claim.
\end{proof}

Using \cref{thm:consecutive_structure}, \textsc{Consecutive} runs in $\mathcal{O}(n(n+\lambda))$ time;
there are $\mathcal{O}(n)$ candidate intervals to examine, and evaluating each candidate takes $\mathcal{O}(n+\lambda)$ time.

\textbf{Duplicate-Opt Method.}
The \textsc{Duplicate-Opt} method computes the optimal solution in the duplicate-allowed setting.
By \cref{thm:duplicate_allowed_structure}, it suffices to try $n$ candidates, one for each $k\in\mathcal{K}$, where all $\lambda$ poisons are placed at $k$.
Evaluating each candidate takes $\mathcal{O}(n+\lambda)$ time.
Hence, the total time complexity is $\mathcal{O}(n(n+\lambda))$, and this yields the optimal solution in the duplicate-allowed setting.

\section{Covered-Legitimate-Keys Minimization}
\label{sec:poisoning_num_coverable_keys}

This section bridges the attack studied in the previous section, namely maximizing maximum error for a fixed set of covered keys, to the problem that is directly relevant to the PGM-index.

\subsection{Problem Setting}

We consider the following setting.
Given a maximum allowed error parameter $\varepsilon$, a linear regression model is built so that it covers as many keys as possible while keeping the maximum error at most $\varepsilon$.
\begin{definition_box}
\begin{definition}[\textbf{Linear Regression on CDFs (Number of Covered Keys)}]
\label{def:linear_regression_on_cdfs_number_of_covered_keys}
Let $\mathcal{X} = \{x_1, x_2, \dots, x_N\}$ be a multiset of natural numbers with $x_1 \leq x_2 \leq \dots \leq x_N$.
For each $i \in [N]$, define the rank of $x_i$ as $r_i \coloneq i$.
Given $\varepsilon \geq 0$, the maximum number of covered keys of linear regression on CDFs for $\mathcal{X}$ under the maximum allowed error parameter $\varepsilon$ is defined by
\begin{equation}
\label{eq:linear_regression_on_cdfs_number_of_covered_keys}
    S_{\varepsilon}(\mathcal{X})
    \coloneq
    \max_{w,b}
    \max\left\{
        i \in [N]
        \;\middle|\;
        \max_{j \in [i]} |wx_j + b - r_j| \le \varepsilon
    \right\}.
\end{equation}
\end{definition}
\end{definition_box}
This problem can be solved in $\mathcal{O}(S_{\varepsilon}(\mathcal{X}))$ time using the algorithm of~\cite{o1981line}.
The PGM-index construction repeats this procedure to obtain a PLA with as few segments as possible.

We next define the poisoning variant of CDF linear regression under the number-of-covered-keys objective.
\begin{definition_box}
\begin{definition}[\textbf{Covered-Legitimate-Keys Minimization Problem}]
\label{def:poisoning_linear_regression_on_cdfs_number_of_covered_keys}
Let $\mathcal{K} = \{k_1, k_2, \dots, k_n\} \subset \mathbb{N}$ be $n$ distinct legitimate keys such that
$k_1 < k_2 < \dots < k_n$, let $\varepsilon \ge 0$, and let $\lambda \in \mathbb{N}$.
For a poison set $\mathcal{P} \subseteq \{k_1, k_1+1, \dots, k_n\} \setminus \mathcal{K}$ with $|\mathcal{P}| \le \lambda$, let $\mathcal{X} \coloneq \mathcal{K} \cup \mathcal{P} = \{x_1, x_2, \dots, x_{N}\}$ be the sorted key set.
The \emph{number of covered legitimate keys} is defined as
\begin{equation}
\label{eq:poisoning_linear_regression_on_cdfs_number_of_covered_keys}
    C_{\varepsilon}(\mathcal{K}, \mathcal{P})
    \coloneq
    \left| \mathcal{K} \cap \{x_1, x_2, \dots, x_{S_{\varepsilon}(\mathcal{X})}\} \right|.
\end{equation}
The \emph{covered-legitimate-keys minimization problem} is to find a poison set minimizing the number of covered legitimate keys:
\begin{equation}
    \argmin_{\mathcal{P} ~ \mathrm{s.t.} ~ |\mathcal{P}| \leq \lambda,\;
    \mathcal{P} \subseteq \{k_1, k_1+1, \dots, k_n\} \setminus \mathcal{K}}
    C_{\varepsilon}(\mathcal{K}, \mathcal{P}).
\end{equation}
\end{definition}
\end{definition_box}
The key-rank plots in the \textbf{P2} block of \cref{fig:overview} illustrate this problem: before poisoning, a segment covers 7 legitimate keys, but after inserting a single poison key, it covers only 5.

Since, on the model-construction side, one cannot distinguish legitimate keys from poisoned keys, it fits a model that maximizes $S_{\varepsilon}(\mathcal{K} \cup \mathcal{P})$.
The attacker, however, chooses $\mathcal{P}$ to minimize $C_{\varepsilon}(\mathcal{K}, \mathcal{P})$.
This objective captures the attacker's goal of minimizing the \emph{effective coverage size} of a segment, i.e., the number of \emph{legitimate} keys covered.
Repeatedly solving this problem increases the number of segments in the resulting PGM-index (detailed in \cref{sec:poisoning_m_opt}).

\subsection{Reducing Covered-Keys Minimization to Maximum-Error Maximization}
\label{sec:bridging_maxerror_and_coverable_keys}

We now relate the maximum-error maximization problem (\cref{def:poisoning_linear_regression_on_cdfs_max_error}) to the covered-legitimate-keys minimization problem (\cref{def:poisoning_linear_regression_on_cdfs_number_of_covered_keys}) through the following theorem.
\begin{theorem_box}
\begin{theorem}
\label{thm:problem_b_c_relationship}
Let $\mathcal{P}^\ast_{\mathrm{cov}}$ be an optimal solution to the covered-legitimate-keys minimization problem (\cref{def:poisoning_linear_regression_on_cdfs_number_of_covered_keys}),
and let $C^\ast$ denote the corresponding number of covered legitimate keys.
If $C^\ast < n$, let $\mathcal{P}^\ast_{\mathrm{err}}$ be an optimal solution to 
the maximum-error maximization problem (\cref{def:poisoning_linear_regression_on_cdfs_max_error}) for the key set $\mathcal{K}_{\leq C^\ast + 1} \coloneq \{k_1, k_2, \dots, k_{C^\ast + 1}\}$.
Then, $\mathcal{P}^\ast_{\mathrm{err}}$ is also an optimal solution to the covered-legitimate-keys minimization problem.
\end{theorem}
\end{theorem_box}

\begin{proof}[Proof of \cref{thm:problem_b_c_relationship}]
From the definition of $\mathcal{P}^\ast_{\mathrm{cov}}$ and $\mathcal{P}^\ast_{\mathrm{err}}$, we have
$E(\mathcal{K}_{\leq C^\ast + 1} \cup \mathcal{P}^\ast_{\mathrm{err}}) \geq E(\mathcal{K}_{\leq C^\ast + 1} \cup \mathcal{P}^\ast_{\mathrm{cov}}) > \varepsilon$.
Therefore, we have $C_{\varepsilon}(\mathcal{K},\mathcal{P}^\ast_{\mathrm{err}}) \leq C^\ast$; thus, $\mathcal{P}^\ast_{\mathrm{err}}$ is optimal for \cref{def:poisoning_linear_regression_on_cdfs_number_of_covered_keys}.
\end{proof}

This theorem shows that, to find an optimal solution for \cref{def:poisoning_linear_regression_on_cdfs_number_of_covered_keys}, it suffices to search for optimal poison sets for \cref{def:poisoning_linear_regression_on_cdfs_max_error} over prefixes of the legitimate keys.
Thus, we can leverage the insights from \cref{sec:poisoning_max_error} to design efficient poisoning methods.

\begin{figure}[t]
    \centering
    \begin{minipage}{0.48\columnwidth}
        \centering
        \includegraphics[width=\linewidth]{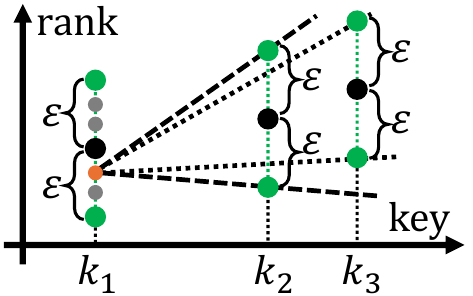}
        \caption{%
            Discrete-intercept linear regression.
            For each candidate intercept, extend the segment while feasible and select the longest one.
        }
        \label{fig:swing_extension}
    \end{minipage}
    \hfill
    \begin{minipage}{0.48\columnwidth}
        \centering
        \includegraphics[width=\linewidth]{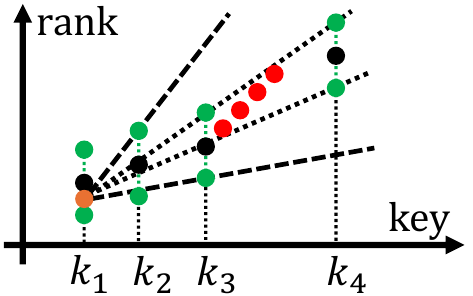}
        \caption{%
            \textsc{DI-Consecutive}.
            With precomputed feasible slope intervals, each poison set candidate can be evaluated in $\mathcal{O}(|\mathcal{I}|\log c)$ time.
        }
        \label{fig:swing_consec_algo}
    \end{minipage}
\end{figure}

\subsection{Methods for Obtaining Poison Solutions}
\label{sec:methods_for_obtaining_poison_solutions_num_coverable_keys}

Based on \cref{thm:problem_b_c_relationship}, we propose three methods for constructing poison sets.
Let $c \coloneq C_{\varepsilon}(\mathcal{K}, \emptyset)$ denote the number of covered legitimate keys without poisoning.
The main takeaway is that, while \textsc{Consecutive} is effective, it takes quadratic time in $c$; \textsc{DI-Consecutive} reduces this to $\mathcal{O}(c\log c)$ while achieving comparable performance in practice (as shown in \cref{sec:experiment_poisoning_num_coverable_keys}).

\textbf{Optimal Method.}
By \cref{thm:maxerror_attack_structure,thm:problem_b_c_relationship}, an optimal solution can be obtained by enumerating $\mathcal{O}\left(\binom{2c-3+\lambda}{\lambda}\right)$ candidates, as in \cref{sec:poisoning_max_error}.
Since each candidate is evaluated in $\mathcal{O}(c+\lambda)$ time, the total running time is
$\mathcal{O}\left(\binom{2c-3+\lambda}{\lambda}(c+\lambda)\right)$.

\textbf{Consecutive Method.}
Motivated by \textsc{Consecutive} proposed in \cref{sec:poisoning_max_error}, we enumerate all consecutive poison sets and select the one minimizing $C_{\varepsilon}(\mathcal{K}, \mathcal{P})$ .
By \cref{thm:consecutive_structure,thm:problem_b_c_relationship}, it suffices to consider $\mathcal{O}(c)$ candidates whose left or right boundary is adjacent to a legitimate key.
Evaluating each candidate takes $\mathcal{O}(c+\lambda)$ time, giving a total running time of $\mathcal{O}\left(c(c+\lambda)\right)$.

\textbf{Discrete-Intercept Consecutive Method (\textsc{DI-Consecutive}).}
Because $c$ can be large (indeed, prior work~\cite{ferragina2020why,liu2024learned} has shown that $c$ scales as $\Theta(\varepsilon^2)$), the quadratic-time \textsc{Consecutive} method can be impractical, especially when it must be applied repeatedly.
We therefore propose \textsc{DI-Consecutive}, which approximates $C_{\varepsilon}(\mathcal{K}, \mathcal{P})$ using an extension of SwingFilter~\cite{elmeleegy2009online}.
Whereas SwingFilter fixes the intercept at $k_1$ to $0$, our extension considers a predefined set $\mathcal{I}$ of candidate intercepts (\cref{fig:swing_extension}).
After precomputation, sparse tables evaluate each consecutive candidate in $\mathcal{O}(|\mathcal{I}|\log c)$ time (\cref{fig:swing_consec_algo}), making the total running time $\mathcal{O}(|\mathcal{I}|c\log c)$.
As we show later in \cref{sec:experiment_poisoning_num_coverable_keys}, approximately 40 intercept candidates suffice to closely match \textsc{Consecutive} in our experiments.
Further details are provided in Appendix~\ref{app:detail_discrete_intercept_consec}.

\section{Poisoning to Maximize $m_\mathrm{opt}$}
\label{sec:poisoning_m_opt}

In this section, we propose a poisoning attack that increases the number of segments in an optimal PLA over a CDF, which is used as the bottom-level PLA of the PGM-index.
We call our attack \emph{\textsc{PGM-attack}}.
\textsc{PGM-attack} repeatedly applies the attack from the previous section, which minimizes the number of covered keys, thereby increasing the number of segments in the optimal PLA.
Note that \textsc{PGM-attack} targets not the PGM-index specifically, but rather the PLA construction algorithm that underlies a wide range of learned indexes~\cite{galakatos2019fiting,kipf2020radixspline,liu2024learned,leying2026line}.

\subsection{Problem Setting}

We first formalize the optimal PLA on CDFs and define $m_{\mathrm{opt}}(\mathcal{X}, \varepsilon)$, i.e., the minimum number of segments.
\begin{definition_box}
\begin{definition}[\textbf{PLA on CDFs}]
\label{def:pla_construction}
Let $\mathcal{X} = \{x_1, x_2, \dots, x_N\}$ be a multiset of natural numbers such that $x_1 \le x_2 \le \dots \le x_N$.
For each $i \in [N]$, define the rank by $r_i \coloneq i$, and let $\varepsilon \ge 0$.

A \emph{segment} is a triple $(s,e,f)$ where $1 \le s \le e \le N$ and $f$ is a linear function such that $\max_{i \in \{s, \dots, e\}} |f(x_i) - r_i| \le \varepsilon$.

A \emph{PLA with $m$ segments} for $\mathcal{X}$ is a sequence of segments
$\bigl((s_1,e_1,f_1), \dots, (s_m,e_m,f_m)\bigr)$ such that $s_1 = 1$, $e_m = N$, and $s_{j+1} = e_j + 1$ for all $j \in [m-1]$.

The objective is to minimize the number of segments $m$.
The minimum number of segments is denoted by $m_{\mathrm{opt}}(\mathcal{X}, \varepsilon)$.
\end{definition}
\end{definition_box}

This problem can be solved in $\mathcal{O}(N)$ time using the algorithm of \cite{o1981line,ferragina2020pgm}.
We next define the poisoning problem of maximizing $m_{\mathrm{opt}}$ for the PGM-index.
\begin{definition_box}
\begin{definition}[\textbf{Segment Maximization Problem}]
\label{def:maximizing_m_opt}
Let $\mathcal{K} = \{k_1, k_2, \dots, k_{n}\} \subset \mathbb{N}$ be $n$ distinct legitimate keys such that $k_1 < k_2 < \dots < k_{n}$, let $\varepsilon \ge 0$, and let $\lambda \in \mathbb{N}$.

The \emph{segment maximization problem} is to find a poison set maximizing the number of segments:
\begin{equation}
    \argmax_{\mathcal{P} ~ \mathrm{s.t.} ~ |\mathcal{P}| \le \lambda,\;
    \mathcal{P} \subseteq \{k_1, k_1+1, \dots, k_n\} \setminus \mathcal{K}}
    m_{\mathrm{opt}}(\mathcal{K} \cup \mathcal{P}, \varepsilon).
\end{equation}
\end{definition}
\end{definition_box}
\textbf{Difficulty of the Problem.}
This problem is challenging for two main reasons.
\textbf{Downstream dependence:}
The placement of poison keys in the early part of the key space affects the entire sequence of subsequent segments.
Thus, the contribution of each poison key cannot be evaluated independently, which makes the problem highly non-local.
\textbf{Computational constraints:}
In typical deployments, the key set size $n$ is often on the order of $10^8$, and the poison budget $\lambda$ is often on the order of $10^6$. Moreover, depending on the choice of $\varepsilon$, the number of keys covered by a single segment can also be large, e.g., on the order of $10^4$ or more.
Let $c$ denote this average number.
Naive dynamic programming or direct use of heavy subroutines such as \textsc{Consecutive} (\cref{sec:methods_for_obtaining_poison_solutions_num_coverable_keys}) incurs a total computational cost of $\Omega(n\lambda c)$ and is therefore practically infeasible.

\subsection{\texorpdfstring{\textsc{PGM-attack}}{PGM-attack}}
\label{sec:method_for_obtaining_poison_solutions_m_opt}

\begin{algorithm}[t]
\caption{\textsc{PGM-attack}}
\label{alg:pgm_poisoning}
\begin{algorithmic}[1]
\Input Legitimate key set $\mathcal{K}$, total poison budget $\lambda$
\Output Poison key set $\mathcal{P}$
\State $n \gets |\mathcal{K}|$
\State $\theta_0 \gets \mathrm{getTheta0}(\mathcal{K})$ \Comment{initial value of $\theta$}
\State $s \gets 1$ \Comment{segment start (index in $\mathcal{K}$)}
\State $\mathcal{P} \gets \emptyset$
\While{$s \leq n$}
 \State $\theta \gets \mathrm{getTheta}(\theta_0, (s-1)/n, |\mathcal{P}|/\lambda)$ \Comment{\cref{eq:dynamic_theta}}
 \State $C^\ast \gets C_\varepsilon(\mathcal{K}[s,s+1,\dots], \emptyset)$ \Comment{\cref{eq:poisoning_linear_regression_on_cdfs_number_of_covered_keys}}
 \State $\lambda^\ast \gets 0, \; \mathcal{P}^\ast_\mathrm{seg} \gets \emptyset$ 
 \State $\mathcal{K}_\mathrm{seg} \gets \mathcal{K}[s,s+1, \dots,s+C^\ast-1]$
 \For{$\lambda_\mathrm{cand} \in \Lambda$}
 \If{$|\mathcal{P}| + \lambda_\mathrm{cand} > \lambda$}
    \State \textbf{continue}
 \EndIf
 \State $\mathcal{P}_\mathrm{cand} \gets \textsc{DI-Consecutive}(\mathcal{K}_\mathrm{seg}, \lambda_\mathrm{cand})$ \Comment{\cref{sec:methods_for_obtaining_poison_solutions_num_coverable_keys}}
 \State $C_\mathrm{cand} \gets C_\varepsilon(\mathcal{K}_\mathrm{seg}, \mathcal{P}_\mathrm{cand})$
 \If{$C_\mathrm{cand} + \theta \lambda_\mathrm{cand} < C^\ast + \theta \lambda^\ast$} 
 \State $\mathcal{P}^\ast_\mathrm{seg} \gets \mathcal{P}_\mathrm{cand}, \;\; C^\ast \gets C_\mathrm{cand}, \;\; \lambda^\ast \gets \lambda_\mathrm{cand}$
 \EndIf
 \EndFor
 \State $\mathcal{P} \gets \mathcal{P} \cup \mathcal{P}^\ast_\mathrm{seg}, \;\; s \gets s + C^\ast$
\EndWhile
\State \Return $\mathcal{P}$
\end{algorithmic}
\end{algorithm}

\textbf{Algorithm overview.}
To address downstream dependencies and computational constraints, we adopt a sequential heuristic that fixes segments from left to right.
At each step, the algorithm selects poison keys for the current segment (i.e., the segment starting from the $s$-th legitimate key), and then advances $s$ to the beginning of the next segment.
Specifically, we prepare a candidate set $\Lambda$ for the number of poisons to be inserted into the current segment.
For each $\lambda \in \Lambda$, we use \textsc{DI-Consecutive} (\cref{alg:di_consecutive}) to efficiently generate a poison set.
Among these candidates, we select the one that minimizes the \textit{cost}.
The full procedure is given in \cref{alg:pgm_poisoning}.

\textbf{Cost function.}
Let $C_\mathrm{cand}$ denote the number of legitimate keys covered after inserting $\lambda_\mathrm{cand}$ poison keys.
We define the cost as $C_\mathrm{cand} + \theta \lambda_\mathrm{cand}$, where $\theta$ is a trade-off parameter justified as follows.
Suppose that each poison key reduces the number of covered legitimate keys by $\hat{\theta}$ in expectation.
Assigning $\lambda_\mathrm{cand}$ poison keys to the current segment has two effects: it leaves $C_\mathrm{cand}$ legitimate keys covered by the current segment, while consuming $\lambda_\mathrm{cand}$ poison keys that could otherwise reduce the coverage of subsequent segments by $\hat{\theta}\lambda_\mathrm{cand}$ keys in expectation.
Thus, $C_\mathrm{cand} + \hat{\theta}\lambda_\mathrm{cand}$ represents the expected number of covered legitimate keys that this candidate fails to remove from coverage.
Minimizing it therefore minimizes the expected number of legitimate keys covered per segment, equivalently increasing the expected number of segments.

\textbf{Tuning $\theta$.}
Following this intuition, we empirically initialize $\theta$ as the average reduction in covered legitimate keys per poison.
We sample 100 random segments, construct poison sets for each $\lambda \in \Lambda$, and set $\theta_0$ to the average reduction per poison.

We then adapt $\theta$ through a multiplicative update based on the imbalance between the progress over legitimate keys and the consumption of the poison budget.
Formally, we define
\begin{equation}
\label{eq:dynamic_theta}
    \theta = \theta_0 e^{\mu (r_\lambda - r_n)}, 
    \quad 
    r_n \coloneq \frac{s - 1}{n}, 
    \quad 
    r_\lambda \coloneq \frac{|\mathcal{P}|}{\lambda}.
\end{equation}
The parameter $\mu > 0$ controls the strength of the correction.
When $r_\lambda > r_n$, the poison budget is consumed faster than the algorithm progresses through the legitimate keys, so $\theta$ increases and discourages larger values of $\lambda$.
Conversely, when $r_\lambda < r_n$, $\theta$ decreases and favors larger values of $\lambda$.
Thus, $\theta$ acts as an adaptive penalty coefficient that balances the two resources throughout the process.
Based on the results of the ablation analysis detailed in Appendix~\ref{app:ablation_study_mu}, we set $\mu = 100$ by default.

\textbf{Choice of $\Lambda$.}
To reduce computation, we restrict the number of poisons assigned to each segment to a small candidate set $\Lambda$.
As shown in \cref{sec:experiment_poisoning_num_coverable_keys}, the reduction in covered legitimate keys typically saturates around $\lambda = 2\varepsilon + 1$, beyond which additional poisons provide little or no benefit.
This is because, under a mild sparsity assumption on the keys, placing $2\varepsilon + 1$ poison keys adjacent to the first legitimate key is sufficient to reduce the number of covered legitimate keys to one.
Motivated by this observation, we set $\Lambda \coloneq \left\{ \left\lfloor t(2\varepsilon+1) / 9 \right\rfloor \,\middle|\, t=0,1,\dots,9 \right\}$.

\subsection{Instance-Agnostic Upper Bound on $m_\mathrm{opt}$}
\label{sec:m_opt_instance_agnostic_upper_bound}

In this section, we derive an upper bound on $m_{\mathrm{opt}}$.
This upper bound provides quantitative insight for both attackers and defenders.
From the attacker's perspective, this upper bound indicates how much further $m_{\mathrm{opt}}$ could potentially be increased beyond the current solution.
From the defender's perspective, it can be used as a measure of robustness against data insertions, including both poisoning and benign insertions.

We first present the following theorem, which provides an upper bound on $m_{\mathrm{opt}}$ after poisoning using only minimal information about the instance $\mathcal{K}$.
\begin{theorem_box}
\begin{theorem}
\label{thm:m_opt_upper_bound}
When $\varepsilon \geq 1/2$, for any $\mathcal{K}$ and $\mathcal{P}$,
\begin{equation}
\label{eq:m_opt_upper_bound}
 m_{\mathrm{opt}}(\mathcal{K} \cup \mathcal{P}, \varepsilon) \leq m_{\mathrm{opt}}(\mathcal{K}, \varepsilon) + |\mathcal{P}|.
\end{equation}
\end{theorem}
\end{theorem_box}
\begin{proof}[Proof Sketch of \cref{thm:m_opt_upper_bound}]
Given an optimal PLA for $\mathcal{K}$, inserting a single key $p$ requires at most one additional segment: add a singleton segment if $p$ lies between two segments, or split the segment containing $p$ otherwise.
Applying this argument repeatedly proves \cref{eq:m_opt_upper_bound}.
\end{proof}

In standard PGM-index settings, $\varepsilon \geq 1$ (see \cite[Def. 1]{ferragina2020pgm}), so the assumption $\varepsilon \ge 1/2$ in \cref{thm:m_opt_upper_bound} always holds.
Interestingly, when $\varepsilon < 1/2$, \cref{thm:m_opt_upper_bound} can fail.
For example, with $\varepsilon = 0$ and $\mathcal{K}=\{1,3,5,7\}$, $m_{\mathrm{opt}}(\mathcal{K}, 0) = 1$.
However, with $\mathcal{P}=\{4\}$, we have $m_{\mathrm{opt}}(\mathcal{K} \cup \mathcal{P}, 0) = 3$, which is greater than $m_{\mathrm{opt}}(\mathcal{K}, 0) + |\mathcal{P}| = 2$.

We further show that there exists $\mathcal{K}$ for which the bound in \cref{thm:m_opt_upper_bound} is tight.
\begin{theorem_box}
\begin{theorem}
\label{thm:m_opt_upper_bound_strict}
For any $\varepsilon \geq 0$ and $\lambda \in \mathbb{N}$,
\begin{equation}
\label{eq:m_opt_upper_bound_strict}
\exists\, \mathcal{K}, \mathcal{P}
~\ ~ \text{s.t.} ~\ ~
 |\mathcal{P}| = \lambda ~\land~
 m_{\mathrm{opt}}(\mathcal{K} \cup \mathcal{P}, \varepsilon) = m_{\mathrm{opt}}(\mathcal{K}, \varepsilon) + \lambda.
\end{equation}
\end{theorem}
\end{theorem_box}
\begin{proof}[Proof Sketch of \cref{thm:m_opt_upper_bound_strict}]
The proof is constructive.
Let $\mathcal{K}$ consist of $\lambda+1$ sufficiently long blocks of consecutive integers separated by sufficiently large gaps.
Then each segment covers one block, so $m_{\mathrm{opt}}(\mathcal{K},\varepsilon)=\lambda+1$.
Let $\mathcal{P}$ be a size-$\lambda$ poison set with one poison at each gap midpoint.
Each inserted poison key creates one additional segment, and thus $m_{\mathrm{opt}}(\mathcal{K}\cup\mathcal{P},\varepsilon)=2\lambda+1$.
\end{proof}
\cref{thm:m_opt_upper_bound_strict} shows that the upper bound in \cref{thm:m_opt_upper_bound} is tight as an instance-agnostic bound, and hence cannot be improved without using instance-specific information.

\subsection{Instance-Dependent Upper Bound on $m_\mathrm{opt}$}
\label{sec:m_opt_instance_dependent_upper_bound}

By exploiting the structure of $\mathcal{K}$, we derive a tighter upper bound on $m_{\mathrm{opt}}$ than the instance-agnostic bound in \cref{sec:m_opt_instance_agnostic_upper_bound}.
We relax the problem (\cref{def:maximizing_m_opt}) to obtain an efficiently computable bound with theoretical guarantees.
Our approach consists of three steps:
relaxing the problem by partitioning $\mathcal{K}$ into blocks, deriving per-block upper bounds, and globally allocating the poison budget.
We briefly describe each step below; full details are provided in Appendix~\ref{app:instance_upper_bound_algorithm}.

\textbf{Relaxation.}
We weaken the PLA construction by partitioning $\mathcal{K}$ into blocks and fixing the slope within each block in advance.
We partition $\mathcal{K}$ using an $\alpha\varepsilon$-PLA with $\alpha \ge 1$ and fix each block's slope to that of the corresponding segment.
We evaluate the upper bound for $\alpha \in \{1.0,1.2,1.4,1.6,1.8,2.0\}$ and take the minimum.
We also strengthen the attacker by allowing duplicate poisons.

\textbf{Per-block upper bound.}
Fixing the slope within each block allows us to derive a closed-form lower bound on the number of poisons required for a segment to span any pair of keys.
We use these pairwise costs as edge weights in a DAG, where each path represents a segmentation and its weight lower-bounds the required poison budget.
The edge weights satisfy the Monge property~\cite{aggarwal1986geometric}, allowing efficient shortest-path computation via LARSCH~\cite{larmore1991line}.
Combined with Lagrangian relaxation~\cite{aggarwal1993finding}, this yields an upper bound on the number of segments achievable under a given poison budget.
The bound is parameterized by a Lagrange parameter $\nu>0$; we evaluate it for each $\nu \in \{1,2,3,4,5,10,20,40,80,160\}$ and take the minimum.

\textbf{Allocating poisons across blocks.}
Each per-block bound is concave in the allocated budget, so the global allocation reduces to a concave knapsack problem.
We solve it by greedily assigning each poison to the block with the largest marginal gain~\cite{federgruen1986greedy}.

\section{Experiments}
\label{sec:experiment}

This section evaluates our poisoning attacks and their impact on index performance.
We first study attacks on CDF linear regression under the maximum-error loss (\cref{sec:experiment_poisoning_max_error}) and attacks that minimize the number of covered legitimate keys (\cref{sec:experiment_poisoning_num_coverable_keys}).
We then evaluate the effect of \textsc{PGM-attack} on the number of PLA segments (\cref{sec:experiment_poisoning_segment_number}) and the performance of various indexes (\cref{sec:experiment_poisoning_index_performance}).

\subsection{Experimental Setup}
\label{sec:experiment_setup}

We implemented all methods in C++ and ran experiments on a Linux machine with an Intel Core i9-11900H CPU @ 2.50\,GHz and 64\,GB of RAM.
We compiled the code with GCC~9.4.0 and the \texttt{-O3} flag.
We used eight threads for poison generation to reduce wall-clock time.
In contrast, index construction and query-time measurements were performed using a single thread for reproducibility.

\paragraph{Datasets.}
We use standard real-world benchmarks for learned indexes together with synthetic datasets.
From SOSD~\cite{sosd-vldb}, we use \textbf{Amzn}, \textbf{Osmc}, and \textbf{Face}; from the ALEX study~\cite{ding2020alex}, we use \textbf{YCSB}, \textbf{Longitudes}, and \textbf{Longlat}.
We scale the real-valued Longitudes and Longlat keys to $[0,2^{63}]$ and round them to integers.
Amzn and Osmc contain 800M keys, while the other real-world datasets contain 200M keys each.
For the synthetic datasets \textbf{Uniform}, \textbf{Normal}, and \textbf{Lognormal}, we draw 200M keys i.i.d., scale them to $[0,2^{63}]$, and round them to integers.
All datasets are duplicate-free.

\subsection{Maximum-Error Maximization}
\label{sec:experiment_poisoning_max_error}

We evaluate the maximum-error poisoning attacks introduced in \cref{sec:poisoning_max_error}.
For each dataset in \cref{sec:experiment_setup}, we construct a legitimate key set by extracting $n$ consecutive keys, reflecting that each PLA segment covers a contiguous key range.
We choose the starting position uniformly at random and repeat the experiment with 100 seeds.
We fix $n=100$ and vary the poison budget as $\lambda \in \{2,4,6,\dots,20\}$.

\paragraph{Methods.}
We evaluate \textbf{\textsc{Greedy}}, \textbf{\textsc{Consecutive}}, and \textbf{\textsc{Duplicate-Opt}} from \cref{sec:poisoning_max_error}.
We compare them with two baselines:
\textbf{\textsc{Random}}, which uniformly samples $\lambda$ keys from the allowed integers, and
\textbf{\textsc{Random-Adjacent}}, which uniformly samples from integers adjacent to legitimate keys, as motivated by \cref{thm:maxerror_attack_structure}.

\begin{figure}[t]
\centering
\includegraphics[width=\columnwidth]{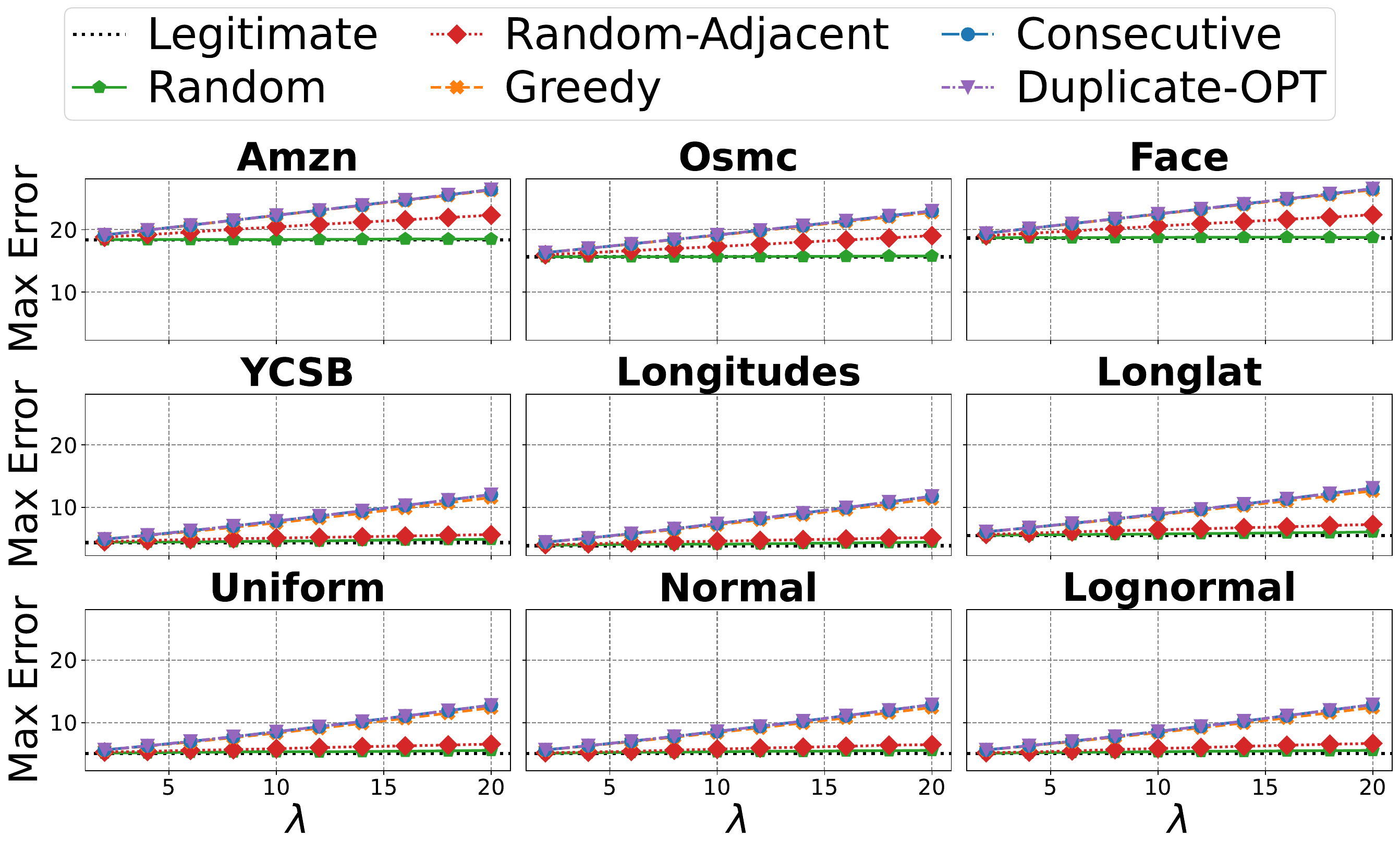}
\caption{Mean maximum error across 100 seeds. While \textsc{Random} barely changes the maximum error, our methods increase it substantially.}
\label{fig:lambda_to_maxerror_mean}
\end{figure}

\paragraph{Results.}
\cref{fig:lambda_to_maxerror_mean} shows the maximum error averaged over 100 seeds.
While \textsc{Random} has little effect, \textsc{Greedy}, \textsc{Consecutive}, and \textsc{Duplicate-Opt} substantially increase the maximum error.
Their curves nearly overlap, and the increase is approximately linear in $\lambda$, with a slope of about $0.4$ across all datasets.
\textsc{Random-Adjacent} sometimes outperforms \textsc{Random}, but remains consistently less effective than these three methods.

\textsc{Consecutive} is optimal in almost all cases.
Let $E_{\mathrm{C}}$ and $E_{\mathrm{D}}$ denote the maximum errors achieved by \textsc{Consecutive} and \textsc{Duplicate-Opt}, respectively.
Among all $9{,}000$ cases (9 datasets, 10 values of $\lambda$, and 100 seeds), we observe $E_{\mathrm{D}}=E_{\mathrm{C}}$ in $8{,}994$ cases.
In the remaining six cases, the gap satisfies $E_{\mathrm{D}}-E_{\mathrm{C}}\leq 0.87$.
Because $E_{\mathrm{D}}$ upper-bounds the optimum of the original problem, these results show that \textsc{Consecutive} attains the optimum in the vast majority of cases and remains close to it in all remaining cases.

We next compare \textsc{Consecutive} with \textsc{Greedy}.
Let $E_{\mathrm{G}}$ denote the maximum error achieved by \textsc{Greedy}.
We have $E_{\mathrm{C}}\geq E_{\mathrm{G}}$ in all $9{,}000$ cases, with strict inequality in $3{,}105$ cases and a maximum gap of $4.4$.
Moreover, \textsc{Consecutive} is substantially faster: it runs in $\mathcal{O}(n(n+\lambda))$ time, whereas \textsc{Greedy} requires $\mathcal{O}(\lambda(n+\lambda)^2)$ time.
Thus, \textsc{Consecutive} is both more effective and efficient than \textsc{Greedy}.

\paragraph{Key takeaway.}
\textsc{Consecutive} efficiently finds a globally optimal maximum-error attack in almost all cases.

\subsection{Covered-Legitimate-Keys Minimization}
\label{sec:experiment_poisoning_num_coverable_keys}

\begin{figure}[t]
\centering
\includegraphics[width=\columnwidth]{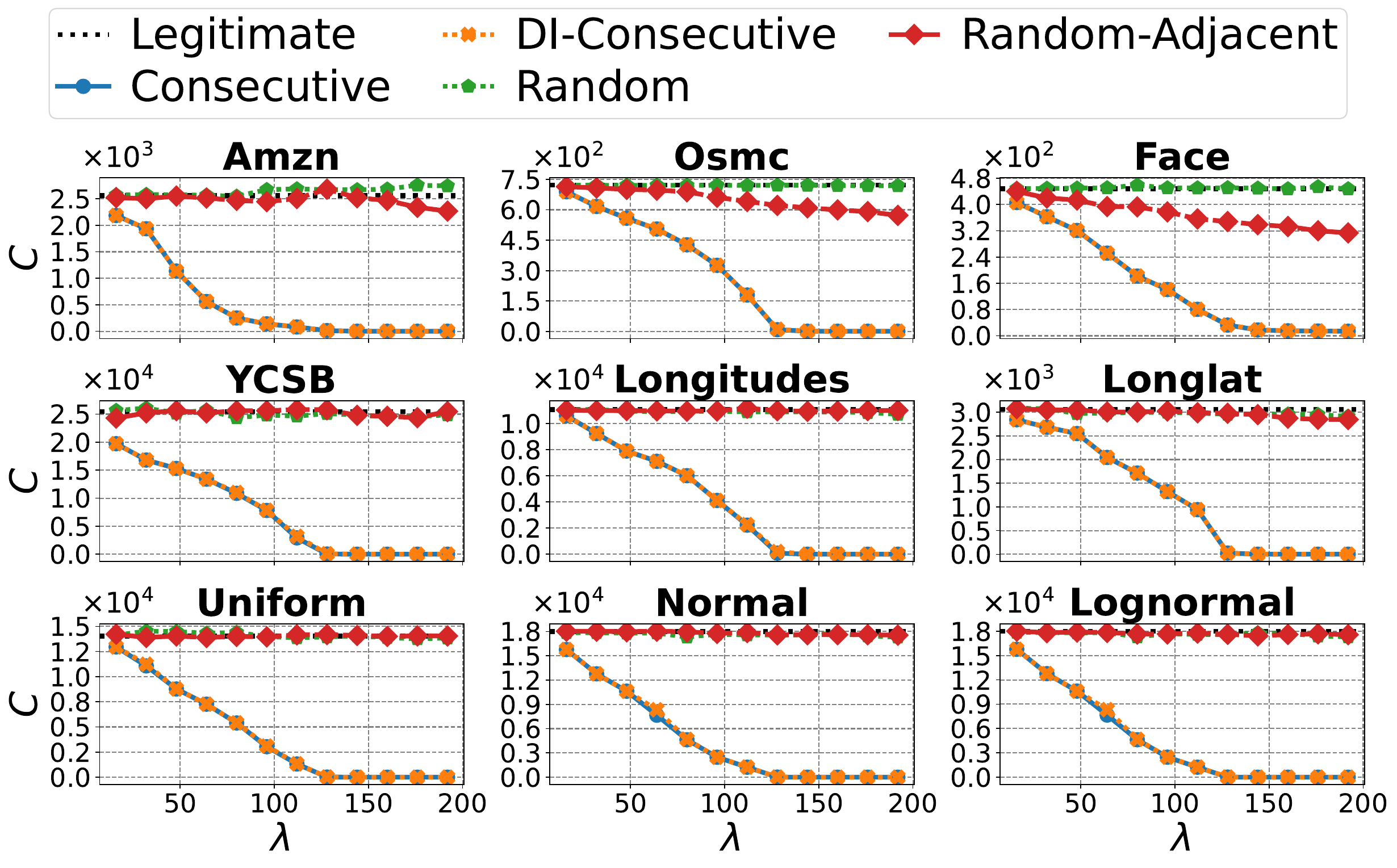}
\caption{Mean number of covered legitimate keys $C$ versus $\lambda$ for $\varepsilon=64$.
\textsc{Consecutive} and \textsc{DI-Consecutive} achieve lower $C$ than the random baselines.}
\label{fig:lambda_to_covered_keys_mean_eps64}
\end{figure}

\begin{figure}[t]
\centering
\includegraphics[width=\columnwidth]{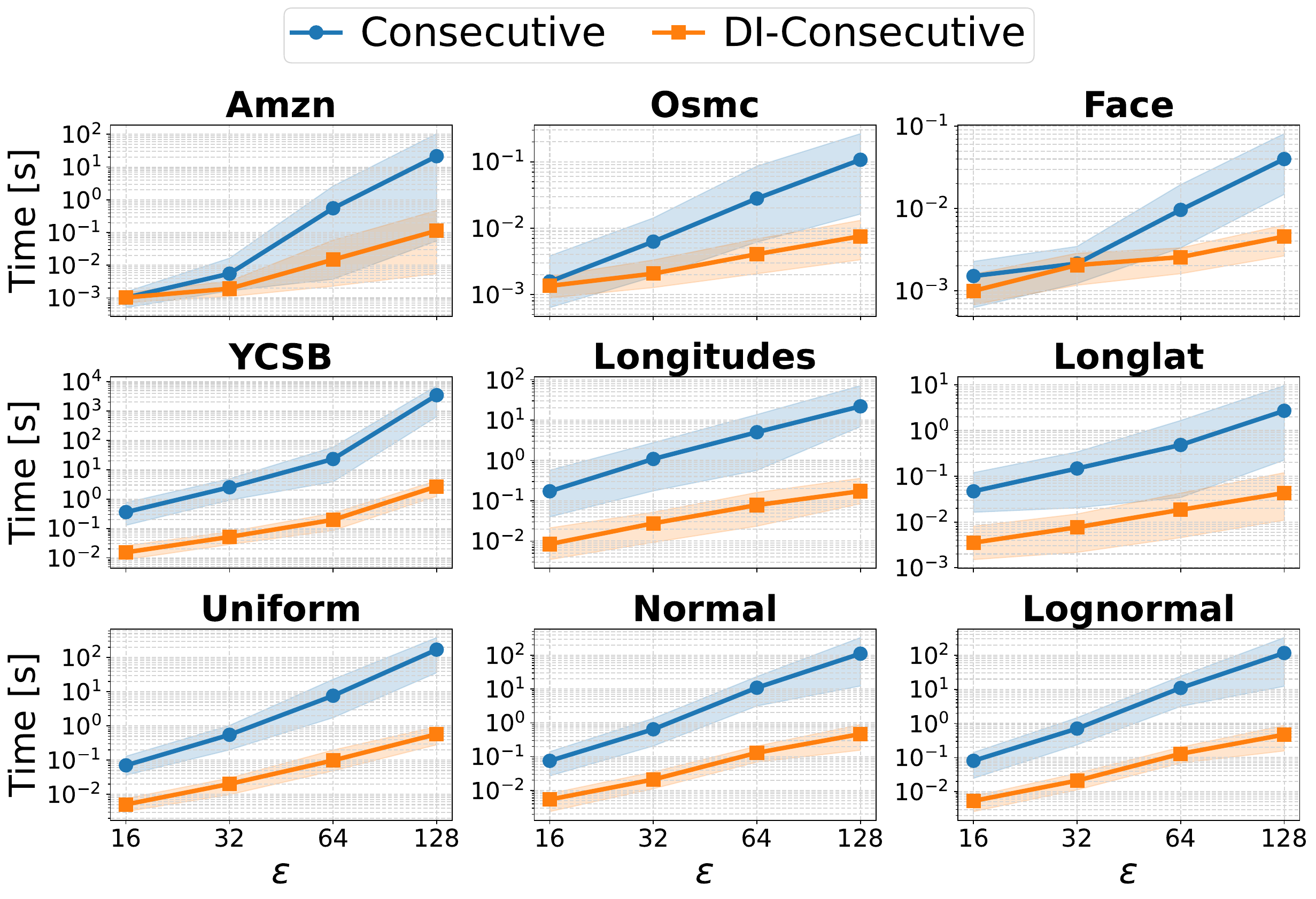}
\caption{Poison generation time.
\textsc{DI-Consecutive} is over $1{,}000\times$ faster than \textsc{Consecutive}.}
\label{fig:epsilon_generation_time}
\end{figure}

We evaluate the covered-legitimate-keys minimization attack introduced in \cref{sec:poisoning_num_coverable_keys}.
For each dataset and seed, we construct the legitimate key set $\mathcal{K}$ (see \cref{def:poisoning_linear_regression_on_cdfs_number_of_covered_keys}) by selecting a random position and taking the suffix starting from that position.
Thus, the number of keys covered before poisoning varies across datasets and seeds.
We set $\varepsilon=64$, vary $\lambda\in\{16i\mid i=1,\dots,12\}$, and use 20 random seeds.
For \textsc{DI-Consecutive}, we use
$\mathcal{I}=\{t\varepsilon/20\mid t=-20,-19,\dots,20\}$.

\paragraph{Methods.}
We evaluate \textbf{\textsc{Consecutive}} and \textbf{\textsc{DI-Consecutive}} from \cref{sec:poisoning_num_coverable_keys}.
We also include \textbf{\textsc{Random}} and \textbf{\textsc{Random-Adjacent}}, defined as in \cref{sec:experiment_poisoning_max_error}, as baselines.

\paragraph{Results.}
\Cref{fig:lambda_to_covered_keys_mean_eps64} shows the mean number of covered legitimate keys before and after poisoning.
\textsc{Random} has little effect, and \textsc{Random-Adjacent} achieves only a modest reduction.
By contrast, \textsc{Consecutive} and \textsc{DI-Consecutive} substantially reduce the covered-key count and achieve nearly identical results.

When $\lambda \leq 2\varepsilon$, the number of covered keys decreases almost linearly with $\lambda$.
In contrast, when $\lambda > 2\varepsilon$, the number of covered keys converges to nearly one and does not decrease further.
This observation motivates our choice of $\Lambda$ in \cref{sec:method_for_obtaining_poison_solutions_m_opt}, where we use 10 evenly spaced integers from $0$ to $2\varepsilon + 1$.

\Cref{fig:epsilon_generation_time} compares the poison generation times of \textsc{Consecutive} and \textsc{DI-Consecutive} with $\lambda=\varepsilon$.
The points show the means over seeds, and the shaded regions indicate the 5th--95th percentiles.
\textsc{DI-Consecutive} is consistently much faster than \textsc{Consecutive}.
The speedup becomes particularly pronounced when $\varepsilon$ is large or when the number of legitimate keys before poisoning is large, reaching over $1{,}000\times$ in the largest cases.

\paragraph{Key takeaway.}
\textsc{DI-Consecutive} achieves nearly the same reduction in covered legitimate keys as \textsc{Consecutive}, while being over three orders of magnitude faster in the largest cases.

\subsection{PLA Segment Maximization}
\label{sec:experiment_poisoning_segment_number}

\begin{table}[t]
    \centering
    \caption{$m_{\mathrm{opt}}$ before/after poisoning, and instance-dependent upper bounds ($\varepsilon=128$).
    \textit{Original} reports $m_{\mathrm{opt}}$ before poisoning.
    \textsc{PGM-attack} increases $m_{\mathrm{opt}}$ by up to $120\times$ under $10\%$ poisoning.
    The instance-specific upper bound is at most $1.92\times$ the $m_{\mathrm{opt}}$ achieved by \textsc{PGM-attack}.}
    \label{tab:pla_segment_num_eps128}
    \setlength{\tabcolsep}{2pt}
    \begingroup
    \fontsize{7.5pt}{9.0pt}\selectfont
    \begin{tabular}{@{}l r rr rr@{}}
        \toprule
            & & \multicolumn{2}{c}{$\lambda = 0.01\,n$} & \multicolumn{2}{c}{$\lambda = 0.1\,n$} \\
        \cmidrule(lr){3-4} \cmidrule(lr){5-6}
        Dataset & Original & \textsc{PGM-attack} & Instance UB & \textsc{PGM-attack} & Instance UB \\
        \midrule
        Amzn & 268K & 335K {\footnotesize (1.25$\times$)} & 619K {\footnotesize (2.31$\times$)} & 670K {\footnotesize (2.50$\times$)} & 1.23M {\footnotesize (4.60$\times$)} \\
        Osmc & 668K & 723K {\footnotesize (1.08$\times$)} & 1.23M {\footnotesize (1.85$\times$)} & 1M {\footnotesize (1.51$\times$)} & 1.93M {\footnotesize (2.88$\times$)} \\
        Face & 256K & 283K {\footnotesize (1.10$\times$)} & 444K {\footnotesize (1.73$\times$)} & 374K {\footnotesize (1.46$\times$)} & 668K {\footnotesize (2.60$\times$)} \\
        YCSB & 663 & 10.3K {\footnotesize (15.5$\times$)} & 17.6K {\footnotesize (26.6$\times$)} & 79.8K {\footnotesize (120$\times$)} & 130K {\footnotesize (196$\times$)} \\
        Longitudes & 13.7K & 22.3K {\footnotesize (1.63$\times$)} & 41.5K {\footnotesize (3.02$\times$)} & 92.4K {\footnotesize (6.73$\times$)} & 154K {\footnotesize (11.2$\times$)} \\
        Longlat & 58.1K & 67.7K {\footnotesize (1.16$\times$)} & 122K {\footnotesize (2.11$\times$)} & 138K {\footnotesize (2.37$\times$)} & 245K {\footnotesize (4.22$\times$)} \\
        Uniform & 3.42K & 12.9K {\footnotesize (3.78$\times$)} & 23.7K {\footnotesize (6.91$\times$)} & 83.7K {\footnotesize (24.5$\times$)} & 136K {\footnotesize (39.8$\times$)} \\
        Normal & 3.5K & 13K {\footnotesize (3.70$\times$)} & 23.7K {\footnotesize (6.78$\times$)} & 83.5K {\footnotesize (23.8$\times$)} & 136K {\footnotesize (38.9$\times$)} \\
        Lognormal & 3.49K & 12.9K {\footnotesize (3.71$\times$)} & 23.7K {\footnotesize (6.81$\times$)} & 83.4K {\footnotesize (23.9$\times$)} & 136K {\footnotesize (39.1$\times$)} \\
        \bottomrule
    \end{tabular}
    \endgroup
\end{table}

\begin{table}[t]
    \centering
    \caption{Results for $m_{\mathrm{opt}}$ after poisoning at different $\varepsilon$ values ($\lambda = 0.1n$). Larger $\varepsilon$ tends to yield a higher increase ratio.}
    \label{tab:pla_segment_num_eps}
    \setlength{\tabcolsep}{3pt}
    \small
    \begin{tabular}{@{}l *{4}{r}@{}}
        \toprule
        Dataset & $\varepsilon=16$ & $\varepsilon=32$ & $\varepsilon=64$ & $\varepsilon=128$ \\
        \midrule
        Amzn & 12.1M {\footnotesize (1.52$\times$)} & 4.79M {\footnotesize (1.94$\times$)} & 1.73M {\footnotesize (2.17$\times$)} & 670K {\footnotesize (2.50$\times$)} \\
        Osmc & 9.18M {\footnotesize (1.49$\times$)} & 4.3M {\footnotesize (1.49$\times$)} & 2.06M {\footnotesize (1.50$\times$)} & 1M {\footnotesize (1.51$\times$)} \\
        Face & 3.08M {\footnotesize (1.45$\times$)} & 1.53M {\footnotesize (1.45$\times$)} & 760K {\footnotesize (1.45$\times$)} & 374K {\footnotesize (1.46$\times$)} \\
        YCSB & 810K {\footnotesize (11.6$\times$)} & 361K {\footnotesize (14.4$\times$)} & 169K {\footnotesize (24.2$\times$)} & 79.8K {\footnotesize (120$\times$)} \\
        Longitudes & 831K {\footnotesize (5.05$\times$)} & 380K {\footnotesize (6.07$\times$)} & 185K {\footnotesize (6.66$\times$)} & 92.4K {\footnotesize (6.73$\times$)} \\
        Longlat & 1.1M {\footnotesize (2.44$\times$)} & 537K {\footnotesize (2.45$\times$)} & 271K {\footnotesize (2.42$\times$)} & 138K {\footnotesize (2.37$\times$)} \\
        Uniform & 930K {\footnotesize (4.57$\times$)} & 391K {\footnotesize (7.38$\times$)} & 177K {\footnotesize (13.2$\times$)} & 83.7K {\footnotesize (24.5$\times$)} \\
        Normal & 929K {\footnotesize (4.58$\times$)} & 391K {\footnotesize (7.39$\times$)} & 177K {\footnotesize (13.1$\times$)} & 83.5K {\footnotesize (23.8$\times$)} \\
        Lognormal & 929K {\footnotesize (4.58$\times$)} & 391K {\footnotesize (7.39$\times$)} & 177K {\footnotesize (13.1$\times$)} & 83.4K {\footnotesize (23.9$\times$)} \\
        \bottomrule
    \end{tabular}
\end{table}

\begin{table}[t]
\centering
\caption{Poison generation time (hours). $\varepsilon=128, \lambda = 0.1n$.}
\label{tab:generation-time}
\small
\setlength{\tabcolsep}{3.5pt}
\renewcommand{\arraystretch}{0.88}
\noindent
\begin{tabular}{@{}l @{\quad} r@{}}
    \toprule
    Dataset & Time (h) \\
    \midrule
    Amzn & 12.3 \\
    Osmc & 5.0 \\
    Face & 1.0 \\
    \bottomrule
    \end{tabular}
    \hspace{1em}
    \begin{tabular}{@{}l @{\quad} r@{}}
    \toprule
    Dataset & Time (h) \\
    \midrule
    YCSB & 14.7 \\
    Longitudes & 5.3 \\
    Longlat & 2.6 \\
    \bottomrule
    \end{tabular}
    \hspace{1em}
    \begin{tabular}{@{}l @{\quad} r@{}}
    \toprule
    Dataset & Time (h) \\
    \midrule
    Uniform & 7.3 \\
    Normal & 7.1 \\
    Lognormal & 7.2 \\
    \bottomrule
\end{tabular}
\end{table}

We evaluate how much \textbf{\textsc{PGM-attack}} increases $m_{\mathrm{opt}}$, along with the instance-dependent upper bounds (\textbf{Instance UB}) from \cref{sec:poisoning_m_opt}.
We use all keys in each dataset: Amzn and Osmc contain 800M keys; the others, 200M.

\paragraph{Poisoning Impact on $m_{\mathrm{opt}}$.}
\cref{tab:pla_segment_num_eps128} reports $m_{\mathrm{opt}}$ before and after \textsc{PGM-attack} poisoning for $\varepsilon = 128$, with parenthesized values indicating the factor relative to pre-poisoning $m_{\mathrm{opt}}$.
\textsc{PGM-attack} can increase $m_{\mathrm{opt}}$ by large factors.
On YCSB, \textsc{PGM-attack} yields a $15.5\times$ increase in $m_{\mathrm{opt}}$ at 1\% poisoning ($\lambda = 0.01n$) and up to $120\times$ at 10\% poisoning ($\lambda = 0.1n$).
Even on datasets other than YCSB, \textsc{PGM-attack} increases $m_{\mathrm{opt}}$ by factors ranging from $1.46\times$ to $24.5\times$ with 10\% poisoning.
In contrast, \textsc{Random} and \textsc{Random-Adjacent} yield increases of at most $1.30\times$ on the other datasets and $2.42\times$ even on YCSB (see Appendix~\ref{app:additional_experiments}).
These results show \textsc{PGM-attack} increases $m_{\mathrm{opt}}$ far more effectively than naive random insertions.

\paragraph{Upper bounds.}
\cref{tab:pla_segment_num_eps128} also reports the Instance UB.
Across all instances, the Instance UB is at most $1.92\times$ the segment count \textsc{PGM-attack} achieves, certifying that \textsc{PGM-attack} achieves at least $1/1.92 \approx 52\%$ of the optimum on every instance.
The remaining gap may reflect suboptimality of the attack, looseness of the bound, or both.
The Instance UB is more than an order of magnitude tighter than the instance-agnostic upper bound; for example, on Amzn, the instance-agnostic upper bound is $80$M, whereas the Instance UB is only $1.23$M.
These results demonstrate that exploiting instance-specific structure enables substantially tighter upper bounds.

\paragraph{Dependence on $\varepsilon$.}
\cref{tab:pla_segment_num_eps} reports the poisoning impact for varying $\varepsilon$ with $\lambda = 0.1n$; parenthesized values indicate increases relative to the pre-poisoning $m_{\mathrm{opt}}$.
Larger $\varepsilon$ generally yields a larger increase ratio, as can be explained as follows.
Let $\hat{\theta}$ denote the average per-segment reduction in legitimate keys caused by one poison, and let $m_{\mathrm{opt}}$ and $m'_{\mathrm{opt}}$ denote the segment counts before and after poisoning, respectively.
Since each segment receives on average $\lambda / m'_{\mathrm{opt}}$ poisons while covered legitimate keys drop from $n / m_{\mathrm{opt}}$ to $n / m'_{\mathrm{opt}}$, we get $\hat{\theta} \lambda / m'_{\mathrm{opt}} \approx n / m_{\mathrm{opt}} - n / m'_{\mathrm{opt}}$, i.e., $m'_{\mathrm{opt}} / m_{\mathrm{opt}} \approx 1 + \hat{\theta} \lambda / n$.
As shown in \cref{fig:lambda_to_covered_keys_mean_eps64}, covered legitimate keys decrease approximately linearly to $1$ at $\lambda \approx 2\varepsilon$, so $\hat{\theta} \approx c / (2\varepsilon)$, where $c$ is the pre-poisoning covered keys per segment.
Since prior work~\cite{ferragina2020why,liu2024learned} suggests $c$ scales superlinearly with $\varepsilon$ (e.g., $c = \Theta(\varepsilon^2)$ under certain distributional assumptions), $m'_{\mathrm{opt}} / m_{\mathrm{opt}}$ increases with $\varepsilon$.
This implies a trade-off: larger $\varepsilon$ enables the superlinear compression that makes the PGM-index attractive, but also increases poisoning sensitivity.

\paragraph{Poison Generation Time.}
\cref{tab:generation-time} reports poison-key generation time for $\varepsilon=128$, $\lambda=0.1n$.
Across datasets, generation takes $1.0$--$14.7$ hours.
Although non-negligible, this cost remains practical for an offline attacker, since the poison keys can be prepared in advance and inserted before index construction; it is therefore a one-time cost incurred per target instance.
Generation time tends to grow larger on datasets where poisoning was more effective, e.g., YCSB.
This is because \textsc{PGM-attack} processes keys left to right, and larger per-segment shortening allows the next segment to start earlier, forcing repeated recomputation over already-processed regions.

\paragraph{Key takeaways.}
\textsc{PGM-attack} markedly increases $m_{\mathrm{opt}}$, especially when $\varepsilon$ is large or the legitimate-key instance has a small $m_{\mathrm{opt}}$.
Our upper bound certifies \textsc{PGM-attack} achieves at least 52\% of the optimum on all instances.

\subsection{Impact on Index Performance}
\label{sec:experiment_poisoning_index_performance}

We evaluate the impact of poison keys generated by \textsc{PGM-attack} on various index structures.
We use all keys in each dataset.
We report only the results on real-world datasets in the main paper; the complete results are provided in Appendix~\ref{app:additional_experiments}.

\paragraph{Poisoning Methods.}
In addition to our \textbf{\textsc{PGM-attack}}, we evaluate \textbf{\textsc{Random}}, which samples poison keys uniformly at random.
This comparison isolates the effect of deliberately selecting poison keys from that of merely increasing the number of keys.

\paragraph{Indexes.}
In addition to the PGM-index, we evaluate the following index structures:
\textbf{FT} (FITing-Tree)~\cite{galakatos2019fiting}, an early representative PLA-based learned index;
\textbf{PGM++}~\cite{liu2024learned}, a state-of-the-art PLA-based learned index extending the PGM-index;
\textbf{RMI}~\cite{kristo2020case}, the first learned index and a representative non-PLA-based approach; and
\textbf{B+-tree}~\cite{bayer1972organization}, a representative classical index structure.

\paragraph{Evaluation Metrics.}
We evaluate index size and query time.
Index size is defined as the memory footprint of the index structure.
Query time is the average time required to look up a uniformly sampled key in the sorted array.
We measure query time end to end; for example, for the PGM-index, it includes both the evaluation of the hierarchical PLA models and the last-mile binary search around the final predicted position.
For each method, we execute ten runs of $10^6$ queries and report the median average query time per query.

\begin{table}[t]
    \centering
    \caption{Impact of \textsc{PGM-attack} on index size and query time for $\varepsilon = 128$ and $\lambda = 0.1n$.
    \textsc{PGM-attack} increases the PGM-index size by up to $120\times$ and query time by up to $1.22\times$.}
    \label{tab:pgm_index_size_query_time_eps128}
    \small
    \setlength{\tabcolsep}{2.5pt}
    \begin{tabular}{@{}l l rr rr@{}}
        \toprule
        Dataset & Index & \multicolumn{2}{c}{Index Size [MB]} & \multicolumn{2}{c}{Query Time [$\mu\mathrm{s}$]} \\
        \cmidrule(lr){3-4} \cmidrule(lr){5-6}
        &  & Random & \textsc{PGM-attack} & Random & \textsc{PGM-attack} \\
        \midrule
            Amzn & PGM & $4.2$ {\footnotesize (1.01$\times$)} & $10.5$ {\footnotesize (2.50$\times$)} & $0.81$ {\footnotesize (1.01$\times$)} & $0.96$ {\footnotesize (1.20$\times$)} \\
             & FT & $33.3$ {\footnotesize (1.01$\times$)} & $58.9$ {\footnotesize (1.78$\times$)} & $0.85$ {\footnotesize (1.01$\times$)} & $0.95$ {\footnotesize (1.13$\times$)} \\
             & PGM++ & $4.2$ {\footnotesize (1.01$\times$)} & $10.5$ {\footnotesize (2.50$\times$)} & $0.74$ {\footnotesize (0.98$\times$)} & $0.86$ {\footnotesize (1.15$\times$)} \\
             & RMI & $0.19$ {\footnotesize (1.00$\times$)} & $0.19$ {\footnotesize (1.00$\times$)} & $0.84$ {\footnotesize (1.00$\times$)} & $0.98$ {\footnotesize (1.18$\times$)} \\
             & B+-tree & $62.3$ {\footnotesize (1.10$\times$)} & $62.3$ {\footnotesize (1.10$\times$)} & $0.97$ {\footnotesize (1.01$\times$)} & $0.95$ {\footnotesize (0.99$\times$)} \\
            \midrule
            Osmc & PGM & $10.5$ {\footnotesize (1.01$\times$)} & $15.7$ {\footnotesize (1.51$\times$)} & $0.92$ {\footnotesize (1.01$\times$)} & $1.03$ {\footnotesize (1.12$\times$)} \\
             & FT & $61.2$ {\footnotesize (1.01$\times$)} & $75.8$ {\footnotesize (1.25$\times$)} & $0.93$ {\footnotesize (1.00$\times$)} & $0.94$ {\footnotesize (1.02$\times$)} \\
             & PGM++ & $2.6$ {\footnotesize (2.02$\times$)} & $1.4$ {\footnotesize (1.13$\times$)} & $0.97$ {\footnotesize (0.95$\times$)} & $1.02$ {\footnotesize (1.00$\times$)} \\
             & RMI & $0.19$ {\footnotesize (1.00$\times$)} & $0.19$ {\footnotesize (1.00$\times$)} & $1.45$ {\footnotesize (0.97$\times$)} & $1.51$ {\footnotesize (1.02$\times$)} \\
             & B+-tree & $62.3$ {\footnotesize (1.10$\times$)} & $62.3$ {\footnotesize (1.10$\times$)} & $0.96$ {\footnotesize (1.04$\times$)} & $0.93$ {\footnotesize (1.01$\times$)} \\
            \midrule
            Face & PGM & $4.0$ {\footnotesize (1.00$\times$)} & $5.9$ {\footnotesize (1.46$\times$)} & $0.76$ {\footnotesize (1.01$\times$)} & $0.78$ {\footnotesize (1.02$\times$)} \\
             & FT & $23.4$ {\footnotesize (1.00$\times$)} & $29.3$ {\footnotesize (1.25$\times$)} & $0.82$ {\footnotesize (0.99$\times$)} & $0.85$ {\footnotesize (1.02$\times$)} \\
             & PGM++ & $1.9$ {\footnotesize (1.00$\times$)} & $5.9$ {\footnotesize (3.12$\times$)} & $0.73$ {\footnotesize (0.96$\times$)} & $0.81$ {\footnotesize (1.06$\times$)} \\
             & RMI & $0.19$ {\footnotesize (1.00$\times$)} & $0.19$ {\footnotesize (1.00$\times$)} & $1.20$ {\footnotesize (0.89$\times$)} & $1.35$ {\footnotesize (1.00$\times$)} \\
             & B+-tree & $15.6$ {\footnotesize (1.10$\times$)} & $15.6$ {\footnotesize (1.10$\times$)} & $0.80$ {\footnotesize (1.01$\times$)} & $0.80$ {\footnotesize (1.01$\times$)} \\
            \midrule
            YCSB & PGM & $0.021$ {\footnotesize (1.98$\times$)} & $1.2$ {\footnotesize (120$\times$)} & $0.67$ {\footnotesize (1.03$\times$)} & $0.79$ {\footnotesize (1.22$\times$)} \\
             & FT & $0.33$ {\footnotesize (1.27$\times$)} & $4.7$ {\footnotesize (18.5$\times$)} & $0.65$ {\footnotesize (1.03$\times$)} & $0.72$ {\footnotesize (1.15$\times$)} \\
             & PGM++ & $0.001$ {\footnotesize (12.8$\times$)} & $1.7$ {\footnotesize (35770$\times$)} & $0.68$ {\footnotesize (1.22$\times$)} & $0.73$ {\footnotesize (1.30$\times$)} \\
             & RMI & $0.19$ {\footnotesize (1.00$\times$)} & $0.19$ {\footnotesize (1.00$\times$)} & $0.54$ {\footnotesize (1.00$\times$)} & $0.72$ {\footnotesize (1.33$\times$)} \\
             & B+-tree & $15.6$ {\footnotesize (1.10$\times$)} & $15.6$ {\footnotesize (1.10$\times$)} & $0.82$ {\footnotesize (1.03$\times$)} & $0.80$ {\footnotesize (1.00$\times$)} \\
            \midrule
            Longit & PGM & $0.22$ {\footnotesize (1.01$\times$)} & $1.4$ {\footnotesize (6.73$\times$)} & $0.76$ {\footnotesize (0.99$\times$)} & $0.85$ {\footnotesize (1.10$\times$)} \\
             & FT & $1.2$ {\footnotesize (1.03$\times$)} & $4.9$ {\footnotesize (4.03$\times$)} & $0.67$ {\footnotesize (0.99$\times$)} & $0.72$ {\footnotesize (1.06$\times$)} \\
             & PGM++ & $0.22$ {\footnotesize (1.92$\times$)} & $1.7$ {\footnotesize (14.7$\times$)} & $0.73$ {\footnotesize (0.98$\times$)} & $0.73$ {\footnotesize (0.98$\times$)} \\
             & RMI & $0.19$ {\footnotesize (1.00$\times$)} & $0.19$ {\footnotesize (1.00$\times$)} & $0.78$ {\footnotesize (1.00$\times$)} & $0.85$ {\footnotesize (1.09$\times$)} \\
             & B+-tree & $15.6$ {\footnotesize (1.10$\times$)} & $15.6$ {\footnotesize (1.10$\times$)} & $0.81$ {\footnotesize (1.02$\times$)} & $0.80$ {\footnotesize (1.01$\times$)} \\
            \midrule
            Longlat & PGM & $0.93$ {\footnotesize (1.02$\times$)} & $2.2$ {\footnotesize (2.37$\times$)} & $0.85$ {\footnotesize (1.01$\times$)} & $0.85$ {\footnotesize (1.01$\times$)} \\
             & FT & $5.1$ {\footnotesize (1.02$\times$)} & $8.7$ {\footnotesize (1.74$\times$)} & $0.72$ {\footnotesize (1.00$\times$)} & $0.73$ {\footnotesize (1.01$\times$)} \\
             & PGM++ & $0.28$ {\footnotesize (1.76$\times$)} & $0.34$ {\footnotesize (2.13$\times$)} & $0.82$ {\footnotesize (0.87$\times$)} & $0.83$ {\footnotesize (0.89$\times$)} \\
             & RMI & $0.19$ {\footnotesize (1.00$\times$)} & $0.19$ {\footnotesize (1.00$\times$)} & $1.21$ {\footnotesize (0.99$\times$)} & $1.27$ {\footnotesize (1.03$\times$)} \\
             & B+-tree & $15.6$ {\footnotesize (1.10$\times$)} & $15.6$ {\footnotesize (1.10$\times$)} & $0.81$ {\footnotesize (1.02$\times$)} & $0.80$ {\footnotesize (1.01$\times$)} \\
        \bottomrule
    \end{tabular}
\end{table}

\paragraph{Impact on the PGM-index.}
Table~\ref{tab:pgm_index_size_query_time_eps128} reports the index size and query time under $\varepsilon=128$ and a $10\%$ poison budget ($\lambda=0.1n$), together with their ratios to the legitimate case.

For the PGM-index, \textsc{PGM-attack} significantly increases the index size.
Because its memory footprint is dominated by the bottom-level PLA, whose size is approximately proportional to $m_{\mathrm{opt}}$, the index grows by up to $120\times$.
This extreme increase occurs on YCSB, whose legitimate index is exceptionally compact; on the other real-world datasets, the increase remains $1.46$--$6.73\times$.
These increases are substantially larger than those caused by random poisoning, which increases the index size by only $1.98\times$ on YCSB and $1.00$--$1.02\times$ on the other datasets.
This demonstrates that \textsc{PGM-attack} effectively selects keys that degrade the PGM-index.
Lookup time also consistently increases, by up to $1.22\times$; this increase is larger than the increase caused by random poisoning, which is at most $1.03\times$.
This increase is likely attributable to the enlarged bottom-level PLA, which may cause more cache misses.

\paragraph{Impact on Other Indexes.}
Table~\ref{tab:pgm_index_size_query_time_eps128} also shows that \textsc{PGM-attack} transfers to other PLA-based learned indexes.
The index size increases by $1.25$--$18.5\times$ for FITing-Tree and by $1.13$--$35{,}770\times$ for PGM++.
The extreme increase for PGM++ occurs on YCSB, where the legitimate index is exceptionally compact before poisoning.
Thus, by maximizing the number of PLA segments, \textsc{PGM-attack} also degrades the memory efficiency of other PLA-based indexes.

The RMI has a fixed model size, so its index size remains $0.19$\,MB under poisoning.
Its query time nevertheless increases by up to $1.33\times$, because keys that are difficult for PLAs to approximate also increase RMI prediction errors, enlarging the final search ranges.

In contrast, the B+-tree is largely unaffected by poisoning: its size increases by only about $1.1\times$, and its query time remains nearly unchanged.
This behavior is consistent with its $\Theta(n)$ space complexity and $\Theta(\log n)$ query-time complexity, which make its performance relatively insensitive to small changes in the key distribution.
Nevertheless, learned indexes can remain substantially more space-efficient in absolute terms; for example, on Amzn, the PGM-index uses $5.9\times$ less memory than the B+-tree.
These results highlight a trade-off between the efficiency of learned indexes and the predictable robustness of the B+-tree.
Our analysis helps make the behavior of learned indexes more tractable by characterizing and bounding their degradation under poisoning.

\paragraph{Key takeaways.}
\textsc{PGM-attack} substantially increases the PGM-index's memory footprint and moderately increases its query time.
It also transfers to other learned indexes, particularly PLA-based ones, whereas classical indexes remain largely unaffected.
These results demonstrate that our attack exploits the distribution sensitivity distinguishing learned indexes from traditional indexes.

\section{Related Work}
\label{sec:related_work}

\subsection{Learned Indexes}
\label{sec:related_work_learned_index}

Learned indexes enhance or replace classical index structures, such as the B+-tree, with machine-learned models to improve memory efficiency and/or query throughput~\cite{kraska2018case_li}.
The core idea is to approximate the cumulative distribution function (CDF) of the data with a learned regression model.
Since the seminal work of Kraska et al.~\cite{kraska2018case_li}, learned indexes have been extended in various directions, including multi-dimensional data, string keys, dynamic workloads, and theoretical analyses; we refer the reader to a recent survey~\cite{al2024survey} for a comprehensive overview.

A particularly important and practical line of research is based on PLA.
Its explicit error control and high compressibility make PLA well suited for learned indexes.
Representative examples include FITing-Tree~\cite{galakatos2019fiting}, which replaces B+-tree leaves with linear models, and RadixSpline~\cite{kipf2020radixspline}, which combines a PLA with a radix table for fast and memory-efficient static indexing.
More recently, faster PLA construction methods such as GreedyPLA and SwingFilter have been studied, trading slight losses in memory efficiency and query performance for shorter construction times~\cite{qin2025piecewise}.
Among PLA-based indexes, the PGM-index~\cite{ferragina2020pgm,ferragina2020why} is one of the most mature and practically relevant, attracting substantial research attention and adoption in real-world systems.
For example, Manticore Search, an open-source search engine, uses the PGM-index for secondary indexes~\cite{manticore_manual_introduction}, and Infinity, an open-source AI-native database, uses PGM-based structures for filtering tasks~\cite{infiniflow_medium_infinity_010}.

\subsection{Attacks on Learned Indexes}
\label{sec:related_work_attacks_on_learned_index}

Although learned indexes achieve strong performance on typical workloads, recent studies have shown that they can be vulnerable to adversarial manipulation.
Prior work has investigated attacks that degrade the performance of learned Bloom filters~\cite{dong2025poisoning,farwah2024exploiting,reviriego2021learned}, learned sketches~\cite{jing2022deceiving}, learned cardinality estimators~\cite{zhang2024pace,li2025algorithmic}, and learned index advisors~\cite{zhou2024trap,zheng2024robustness}.

Attacks and robustness issues concerning learned indexes in the narrower sense (i.e., learned B-tree) have also received considerable attention~\cite{kornaropoulos2022price,yang2023algorithmic,schuster2025learned,sato2026foundations}.
Kornaropoulos et al.~\cite{kornaropoulos2022price} study poisoning attacks against linear regression models trained on CDFs and extend their attack to RMI.
A related but distinct line of work studies algorithmic complexity attacks (ACAs)~\cite{yang2023algorithmic,schuster2025learned} against ALEX~\cite{ding2020alex}.
Unlike our setting, in which poison keys are injected before index construction, these attacks exploit weaknesses in ALEX's online insertion algorithms to degrade its performance.
We compare against these attacks in Appendix~\ref{app:additional_experiments} and find that our attack consistently causes greater performance degradation.

\section{Discussion}
\label{sec:discussion}

\paragraph{Designing robust learned indexes.}
Developing learned indexes that perform well on benign data while remaining resilient to poisoning attacks, or more generally to unfavorable data insertions, is an important direction for future work.
Our results show that inserting only a small number of keys can substantially reduce the compression benefit of PGM-index, which is built on optimal PLAs.
This finding suggests that worst-case performance guarantees alone are insufficient to guarantee robustness: achieving robustness requires rethinking the objective function used to construct the index.
A seemingly simple defense is to pre-insert the poisoning keys identified by our method as virtual points.
However, this would merely incur the attack-induced increase in index size in advance, sacrificing the benign-data efficiency that the defense is intended to preserve.
Possible directions include relaxing the maximum-error constraint or using more flexible models in hard-to-approximate regions, as well as separating suspicious or hard-to-approximate keys from the main index structure.

A key challenge is that effective poisoning keys are difficult to distinguish from legitimate ones.
\cref{thm:maxerror_attack_structure} guarantees the existence of an optimal poison set consisting of consecutive blocks adjacent to legitimate keys.
Therefore, classical robust regression techniques~\cite{fischler1981random,chen2013approximating}, which suppress the influence of outliers, cannot be directly applied to regression over the CDF, as also discussed in prior work~\cite{kornaropoulos2022price,sato2026foundations}.
Developing effective defenses will therefore likely require accounting for these problem-specific characteristics.

\paragraph{Improving poisoning attacks and upper bounds.}
The Instance UB is up to $1.92\times$ the segment count achieved by \textsc{PGM-attack}.
This gap may stem from suboptimality of the attack, looseness of the bound, or both.
The attack's suboptimality arises mainly from its greedy strategy, which selects poisoning keys segment by segment from left to right.
The bound's looseness has two sources: the relaxation and the per-block upper bound.
In both cases, dynamic programming could achieve better results, but a naive implementation incurs quadratic time complexity, which is a major bottleneck.

\paragraph{Transferability of \textsc{PGM-attack}.}
As shown in \cref{tab:pgm_index_size_query_time_eps128}, the poisoning keys generated by \textsc{PGM-attack} remain largely effective across a diverse range of learned indexes.
We attribute this transferability to the fact that \textsc{PGM-attack} targets not an implementation-specific property of the PGM-index, but the PLA, a fundamental approximation component, which is used in many learned indexes.
Many learned-index architectures share the same underlying principle of partitioning the key space into local regions and approximating each region with a simple model.
The poisoning keys generated by \textsc{PGM-attack} deliberately induce many local regions that are hard to approximate, thereby inflating the number of models or the size of auxiliary structures even in indexes that adopt different model classes or partitioning strategies.

\paragraph{Attacks without knowledge of the PLA parameter.}
Although our threat model is white-box and assumes that the attacker knows the error parameter $\varepsilon$, we also evaluate the attack under the more realistic scenario in which this knowledge is unavailable.
Specifically, we generate poison keys assuming $\varepsilon_{\mathrm{attack}}$ and evaluate them against a PLA built with $\varepsilon_{\mathrm{true}}$, which may differ from $\varepsilon_{\mathrm{attack}}$.
The attack is most effective when $\varepsilon_{\mathrm{attack}} = \varepsilon_{\mathrm{true}}$, but it retains substantial effectiveness even under mismatch; complete results are provided in Appendix~\ref{app:additional_experiments}.
A possible explanation is that \textsc{PGM-attack} creates local irregularities that remain difficult to approximate with linear models across different PLA parameters.
These results show that our white-box attack also functions as a \textit{gray-box} attack, where the attacker knows the legitimate key set but not the PLA parameter.
This finding provides a basis for developing attacks under more restrictive knowledge assumptions.

\paragraph{Beyond a security threat model.}
Although we formulate our study in terms of a poisoning adversary, its significance extends beyond security.
Our attack can also be viewed as a worst-case analysis of data insertions into PLA-based indexes, identifying inputs that sharply increase the segment count, index size, and query time.
The generated poisoning keys can therefore serve as stress tests for learned indexes and help expose performance failure modes that should be addressed in the design of robust data structures.

\section{Conclusion}
\label{sec:conclusion}

We presented the first systematic study of poisoning attacks against optimal PLA construction.
Specifically, we investigated three problems: maximizing the maximum error in CDF linear regression, minimizing the number of legitimate keys covered under a maximum-error constraint, and maximizing the number of segments in an optimal PLA.
For each problem, we developed both theoretical results and practical algorithms.
Experiments on real-world and synthetic datasets showed that our attacks substantially degrade the memory efficiency of the PGM-index and other learned indexes.
These results reveal that optimal segment minimization does not guarantee robustness to adversarial insertions, motivating the development of robustness-aware learned indexes.


\clearpage

\bibliographystyle{ACM-Reference-Format}
\bibliography{
  bib/attack,
  bib/classic_index,
  bib/learned_data_structures,
  bib/learned_index,
  bib/learned_sort,
  bib/others,
  bib/pla
}

\clearpage
\appendix
\section{Proofs}
\label{app:proofs}
Here, we provide the proofs we omitted in the main text.

\subsection{Proof of \cref{lem:block_moving}}
\begin{proof}[Proof of \cref{lem:block_moving}]
The proof proceeds in three main steps: (1) defining the feasible region in a transformed parameter space, (2) deriving the translation of this region when the block is shifted, and (3) utilizing the Minkowski difference and convexity to establish the final error bound.
\cref{fig:proof_maxerror_lemma} illustrates the main idea.

\begin{figure}[t]
    \centering
    \includegraphics[width=\columnwidth]{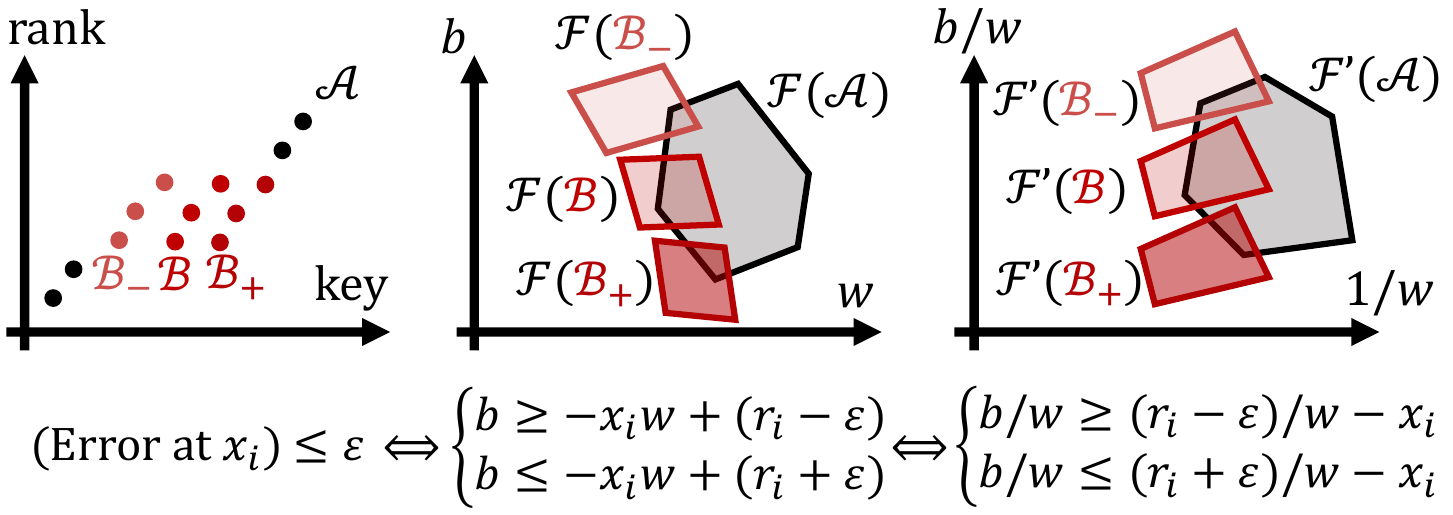}
    \caption{Illustration of the proof of \cref{lem:block_moving}.
    Let $\mathcal{B}_{-} \coloneq \mathcal{B}(\delta_{-})$, $\mathcal{B} \coloneq \mathcal{B}(0)$, and $\mathcal{B}_{+} \coloneq \mathcal{B}(\delta_{+})$.
    In the transformed $(1/w,b/w)$ space, moving the block along the key axis, i.e., from $\mathcal{B}_{-}$ to $\mathcal{B}$ to $\mathcal{B}_{+}$, corresponds to translating its feasible region along the $b/w$ direction.}
    \label{fig:proof_maxerror_lemma}
\end{figure}

\textbf{1. Feasible Region and Parameter Transformation.}
We fix the error threshold $\varepsilon \ge 0$.
For a given multiset $\mathcal{X} = \{x_1, x_2, \dots, x_N\}$ with $x_1 \leq x_2 \leq \dots \leq x_N$, we define the feasible region $\mathcal{F}(\mathcal{X})$ in the $(w, b)$ plane as:
\begin{equation}
    \mathcal{F}(\mathcal{X}) \coloneq \{ (w, b) \in \mathbb{R}^2 \mid \forall x_i \in \mathcal{X}, \, |w x_i + b - r_i| \leq \varepsilon \},
\end{equation}
where $r_i$ is the rank of $x_i$ in $\mathcal{X}$, i.e., $r_i \coloneq i$.
Since each inequality $|w x_i + b - r_i| \leq \varepsilon$ defines a convex region bounded by two parallel lines, their intersection $\mathcal{F}(\mathcal{X})$ is also a convex region.

Now, we apply the transformation $(u, v) = (1/w, b/w)$.
Since both keys $x_i$ and ranks $r_i$ are non-decreasing, the optimal slope $w$ is strictly positive.
Therefore, the transformation $(u, v) = (1/w, b/w)$ is a bijection.
In this transformed space, the inequality $|w x_i + b - r_i| \leq \varepsilon$ is equivalent to:
\begin{equation}
    (r_i - \varepsilon)u - v \leq x_i \leq (r_i + \varepsilon)u - v.
\end{equation}
Let $\mathcal{F}'(\mathcal{X})$ denote the feasible region in the $(u, v)$ space.
Being defined by a set of linear inequalities, $\mathcal{F}'(\mathcal{X})$ is also a convex polygon.

\textbf{2. Shift in the Transformed Space.} 
The condition $\delta \in [x_{l-1} - x_l, x_{r+1} - x_r]$ ensures that the rank $r_i$ for each key $x_i$ remains constant.
For the shifted block $\mathcal{B}(\delta) = \{ x_i + \delta \}_{i=l}^r$, the feasibility condition becomes:
\begin{align}
    &~ (r_i - \varepsilon)u - v \leq x_i + \delta \leq (r_i + \varepsilon)u - v \\
    \iff&~ (r_i - \varepsilon)u - (v + \delta) \leq x_i \leq (r_i + \varepsilon)u - (v + \delta).
\end{align}
This matches the condition for the original block $\mathcal{B}(0)$ evaluated at the point $(u, v + \delta)$. Therefore, the feasible region for the shifted block is simply the region for the original block shifted by $-\delta$ along the $v$-axis:
\begin{equation}
    \mathcal{F}'(\mathcal{B}(\delta)) = \mathcal{F}'(\mathcal{B}(0)) + \bm{d}(\delta)
\end{equation}
where $\bm{d}(\delta) \coloneq (0, -\delta)$.

\textbf{3. Overlap and Convexity.} 
Let $\mathcal{G} \coloneq \mathcal{F}'(\mathcal{A})$ and $\mathcal{H} \coloneq \mathcal{F}'(\mathcal{B}(0))$. Both $\mathcal{G}$ and $\mathcal{H}$ are convex sets. The combined dataset $\mathcal{A} \uplus \mathcal{B}(\delta)$ has an error of at most $\varepsilon$ if and only if $\mathcal{G} \cap (\mathcal{H} + \bm{d}(\delta)) \neq \emptyset$.

Using the Minkowski difference $\mathcal{D} \coloneq \mathcal{G} - \mathcal{H} = \{ \bm{g} - \bm{h} \mid \bm{g} \in \mathcal{G}, \bm{h} \in \mathcal{H} \}$, the intersection condition $\mathcal{G} \cap (\mathcal{H} + \bm{d}(\delta)) \neq \emptyset$ is equivalent to $\bm{d}(\delta) \in \mathcal{D}$. Because $\mathcal{G}$ and $\mathcal{H}$ are convex, their Minkowski difference $\mathcal{D}$ is also convex.

Now, let $\varepsilon = \max(E(\mathcal{A} \uplus \mathcal{B}(\delta_{-})), E(\mathcal{A} \uplus \mathcal{B}(\delta_{+})))$. By assumption, both extreme shifts achieve an error of at most $\varepsilon$, which implies:
\begin{equation}
    \bm{d}(\delta_{-}) = (0, -\delta_{-}) \in \mathcal{D} \quad \text{and} \quad \bm{d}(\delta_{+}) = (0, -\delta_{+}) \in \mathcal{D}.
\end{equation}
Since $\delta_{-} \leq 0 \leq \delta_{+}$, the shift $\delta = 0$ is a convex combination of $\delta_{-}$ and $\delta_{+}$. Consequently, the zero vector $\bm{d}(0) = (0, 0)$ is a convex combination of $\bm{d}(\delta_{-})$ and $\bm{d}(\delta_{+})$. By the convexity of $\mathcal{D}$, it follows that $\bm{d}(0) \in \mathcal{D}$.

This implies $\mathcal{G} \cap \mathcal{H} \neq \emptyset$, meaning there exists a valid parameter pair $(u, v)$ for the unshifted dataset. Thus, $E(\mathcal{A} \uplus \mathcal{B}(0)) \leq \varepsilon$, which completes the proof.
\end{proof}

\subsection{Proof of \cref{thm:maxerror_attack_structure}}
\begin{proof}[Proof of \cref{thm:maxerror_attack_structure}]
We prove the claim by contradiction.
Assume that no optimal solution satisfies the stated property.
That is, in every optimal solution, there exists at least one poison key that does not belong to any block adjacent to a legitimate key.

Among all optimal solutions, let $\mathcal{P}^\ast$ be one minimizing the number of \textit{isolated poison blocks}, where an isolated poison block is a maximal consecutive-integer subset of poisons whose endpoints are not adjacent to any legitimate key.
Formally, an \emph{isolated poison block} is a set of poison keys $\{a, a+1, \dots, b\} \subseteq \mathcal{P}^\ast$ such that $a-1 \notin \mathcal{K} \cup \mathcal{P}^\ast$ and $b+1 \notin \mathcal{K} \cup \mathcal{P}^\ast$.

By assumption, $\mathcal{P}^\ast$ contains at least one such block.
Let $\mathcal{B}$ be one of them, and let $\mathcal{A} \coloneq (\mathcal{K} \cup \mathcal{P}^\ast) \setminus \mathcal{B}$.
By \cref{lem:block_moving}, shifting $\mathcal{B}$ either left or right does not decrease the maximum-error value, while reducing the number of isolated poison blocks (see \cref{fig:proof_maxerror_theorem}).
This contradicts the choice of $\mathcal{P}^\ast$, completing the proof.
\end{proof}

\subsection{Proof of \cref{thm:duplicate_allowed_structure}}
Next, we provide the proof of \cref{thm:duplicate_allowed_structure}.
\begin{proof}[Proof of \cref{thm:duplicate_allowed_structure}]
The proof proceeds in three main steps: (1) showing that any optimal solution $\mathcal{P}^\ast$ uses the full budget $|\mathcal{P}^\ast|=\lambda$, (2) showing that there exists an optimal $\mathcal{P}^\ast$ with $\mathrm{Supp}(\mathcal{P}^\ast)\subseteq\mathcal{K}$, and (3) showing that there exists an optimal $\mathcal{P}^\ast$ supported on a single integer $k\in\mathcal{K}$ (all poisons at the same value).
\begin{figure}[t]
    \centering
    \includegraphics[width=\columnwidth]{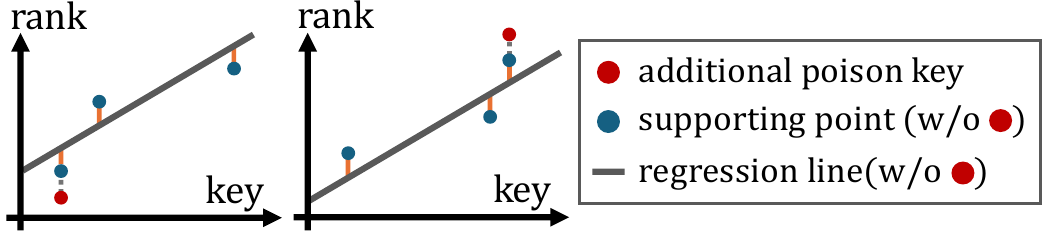}
    \caption{Step 1 of the proof of \cref{thm:duplicate_allowed_structure}: showing $|\mathcal{P}^\ast|=\lambda$.
    Regardless of how the three supporting points are arranged, we can always add a suitable poison to strictly increase the maximum error.}
    \label{fig:proof_maxerror_dup_theorem}
\end{figure}

\begin{figure}[t]
    \centering
    \includegraphics[width=\columnwidth]{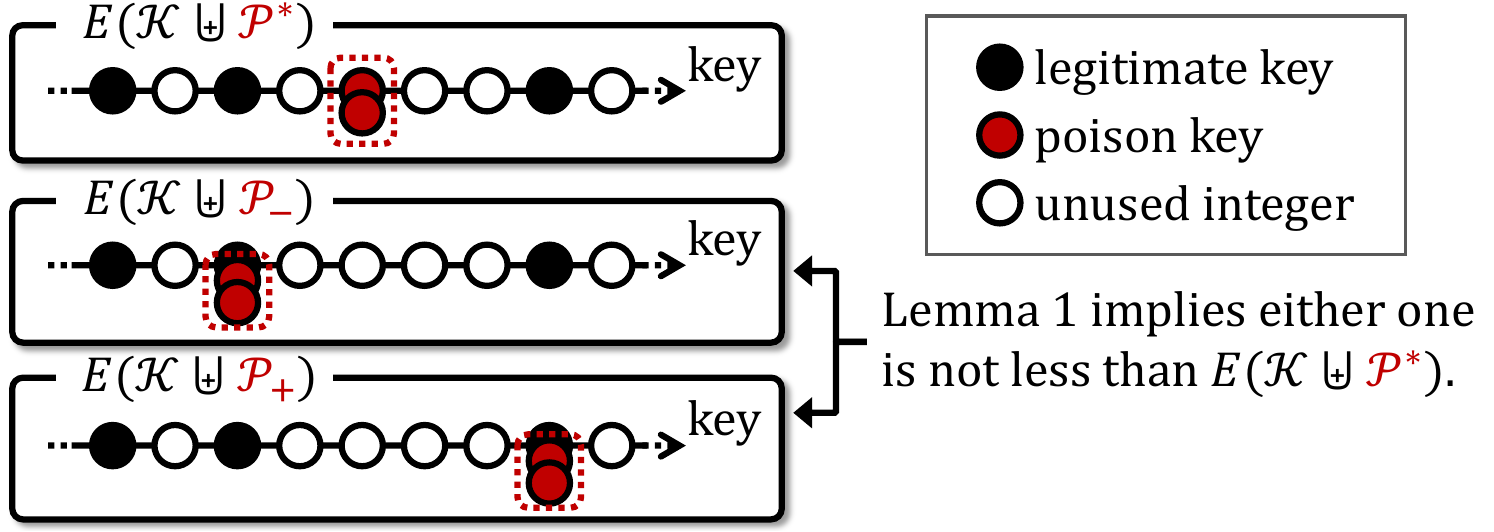}
    \caption{Step 2 of the proof of \cref{thm:duplicate_allowed_structure}: showing $\mathrm{Supp}(\mathcal{P}^\ast)\subseteq\mathcal{K}$.
    By shifting a block in $\mathrm{Supp}(\mathcal{P}^\ast)\setminus\mathcal{K}$ left or right, we can reduce $|\mathrm{Supp}(\mathcal{P}^\ast)\setminus\mathcal{K}|$ without decreasing the objective.}
    \label{fig:proof_maxerror_dup_theorem_2}
\end{figure}

\textbf{1. Full Budget at Optimality.}
It suffices to show that for any multiset $\mathcal{P}$ with $|\mathcal{P}|<\lambda$, there exists a key value $p$ such that adding one copy of $p$ to $\mathcal{P}$ strictly increases the maximum-error value $E(\mathcal{K}\uplus\mathcal{P})$.
Such a $p$ can be identified by examining the supporting points (see \cref{fig:proof_maxerror_dup_theorem}).

Let $\mathcal{P}$ be a multiset with $|\mathcal{P}|<\lambda$, let $\varepsilon = E(\mathcal{K}\uplus\mathcal{P})$, and consider an optimal maximum-error regression line for $\mathcal{K}\uplus\mathcal{P}$.
It is known that there exist three points $x_1,x_2,x_3 \in \mathcal{K}\uplus\mathcal{P}$ whose errors attain $\pm\varepsilon$ with alternating signs.
If the signs are $-\varepsilon,+\varepsilon,-\varepsilon$ (left panel of \cref{fig:proof_maxerror_dup_theorem}), choose $p=x_1$.
If the signs are $+\varepsilon,-\varepsilon,+\varepsilon$ (right panel), choose $p=x_3$.
In either case, no single line can regress the four points (i.e., $x_1$, $x_2$, $x_3$, and $p$) with maximum error $\varepsilon$, hence the optimum maximum error must strictly increase.
Therefore, whenever $|\mathcal{P}|<\lambda$, we can strictly increase $E(\mathcal{K}\uplus\mathcal{P})$ by adding one poison, which implies $|\mathcal{P}^\ast|=\lambda$.

\textbf{2. Support on Legitimate Keys.}
We prove the existence of an optimal solution supported only on legitimate keys by a contradiction argument similar to that of \cref{thm:maxerror_attack_structure}.
Assume that every optimal poison set contains at least one value that is not a legitimate key.
Among optimal solutions, choose $\mathcal{P}^\ast$ that minimizes $|\mathrm{Supp}(\mathcal{P}^\ast)\setminus\mathcal{K}|$.
Then there exists a contiguous block of support values in $\mathrm{Supp}(\mathcal{P}^\ast)\setminus\mathcal{K}$.
By shifting this block left or right (within the range where ranks do not change), we can apply the block-moving argument to obtain another optimal solution whose maximum error does not decrease, while $|\mathrm{Supp}(\mathcal{P}^\ast)\setminus\mathcal{K}|$ decreases (see \cref{fig:proof_maxerror_dup_theorem_2}).
This contradicts the minimality of $|\mathrm{Supp}(\mathcal{P}^\ast)\setminus\mathcal{K}|$.
Hence, there exists an optimal solution with $\mathrm{Supp}(\mathcal{P}^\ast)\subseteq\mathcal{K}$.

\textbf{3. Concentration on a Single Key.}
Finally, we show that any optimal poison set whose support is contained in $\mathcal{K}$ can be transformed into another optimal poison set supported on a single key value.

For $\bm{d}\in\mathbb{Z}_{\ge 0}^n$, let $\mathcal{Q}_{\mathcal{K}}(\bm{d})$ denote the multiset that contains $d_i$ copies of $k_i$ for each $i$.
Consider any $\bm{d}$ with $d_i>0$ and $d_j>0$ for some $i<j$.
Define $\bm{d}_{-}$ by moving all copies at $k_j$ to $k_i$, and define $\bm{d}_{+}$ by moving all copies at $k_i$ to $k_j$, that is,
\begin{align}
    \bm{d}   &= [d_1, \dots, d_{i-1}, & d_i & , \dots, & d_j & , d_{j+1}, \dots, d_n], \\
    \bm{d}_{-} &= [d_1, \dots, d_{i-1}, & d_i+d_j & , \dots, & 0 & , d_{j+1}, \dots, d_n], \\
    \bm{d}_{+} &= [d_1, \dots, d_{i-1}, & 0 & , \dots, & d_i+d_j & , d_{j+1}, \dots, d_n].
\end{align}
We claim that for any $\varepsilon\ge 0$,
\begin{align}
    &~ E(\mathcal{K}\uplus\mathcal{Q}_{\mathcal{K}}(\bm{d}_{-})) \leq \varepsilon
    ~\land~
    E(\mathcal{K}\uplus\mathcal{Q}_{\mathcal{K}}(\bm{d}_{+})) \leq \varepsilon \\
    \label{eq:duplicate_allowed_structure_condition}
    \implies~&~
    E(\mathcal{K}\uplus\mathcal{Q}_{\mathcal{K}}(\bm{d})) \leq \varepsilon.
\end{align}

\begin{figure}[t]
    \centering
    \includegraphics[width=0.65\columnwidth]{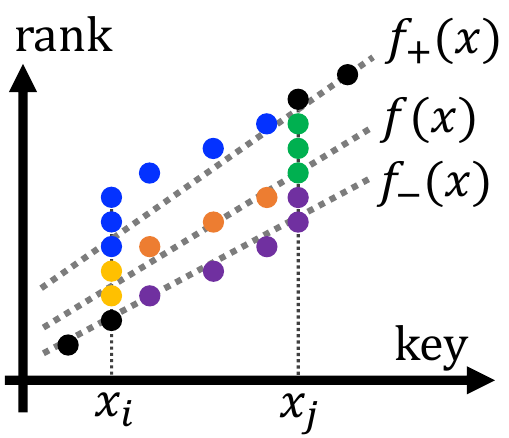}
    \caption{Step 3 of the proof of \cref{thm:duplicate_allowed_structure}. The figure illustrates the case where $d_i = 2$ and $d_j = 3$. Property 1 shows that the error of $f$ with respect to the black, yellow, and green points is at most $\varepsilon$. Property 2 shows that the error for the orange points is at most $\varepsilon$.}
    \label{fig:proof_maxerror_dup_theorem_3}
\end{figure}

To see this, let $f_{-}(x)$ and $f_{+}(x)$ be regression lines achieving maximum error at most $\varepsilon$ for
$\mathcal{K}\uplus\mathcal{Q}_{\mathcal{K}}(\bm{d}_{-})$ and $\mathcal{K}\uplus\mathcal{Q}_{\mathcal{K}}(\bm{d}_{+})$, respectively.
Define the convex combination
\begin{equation}
    f(x) \coloneq \frac{d_i \cdot f_{+}(x) + d_j \cdot f_{-}(x)}{d_i + d_j}.
\end{equation}
Then, we can show that $f(x)$ attains maximum error at most $\varepsilon$ on $\mathcal{K}\uplus\mathcal{Q}_{\mathcal{K}}(\bm{d})$, which establishes \cref{eq:duplicate_allowed_structure_condition}.
To see this, it is sufficient to show the following two properties (each of which can be proved using elementary methods):
\begin{itemize}[leftmargin=1.5em]
    \item \textbf{Property 1:} If a point $(x, y)$ has an error of at most $\varepsilon$ under both $f_{-}$ and $f_{+}$, then its error under $f$ is also at most $\varepsilon$.
    \item \textbf{Property 2:} If the points $(x, y)$ and $(x, y + d_i + d_j)$ have errors of at most $\varepsilon$ under $f_{-}$ and $f_{+}$, respectively, then the error of the point $(x, y + d_i)$ under $f$ is at most $\varepsilon$.    
\end{itemize}
We provide \cref{fig:proof_maxerror_dup_theorem_3} to aid understanding. 
$\mathcal{K}\uplus\mathcal{Q}_{\mathcal{K}}(\bm{d}_{-})$ corresponds to the black, purple, and green point sets; 
$\mathcal{K}\uplus\mathcal{Q}_{\mathcal{K}}(\bm{d}_{+})$ corresponds to the black, yellow, and blue point sets; 
and $\mathcal{K}\uplus\mathcal{Q}_{\mathcal{K}}(\bm{d})$ corresponds to the black, yellow, orange, and green point sets. 
By Property 1, it follows that the error of the black, yellow, and green point sets under $f$ is at most $\varepsilon$. 
By Property 2, it follows that the error of the orange point set under $f$ is at most $\varepsilon$. 
Thus, via Properties 1 and 2, it is shown that the error of $\mathcal{K}\uplus\mathcal{Q}_{\mathcal{K}}(\bm{d})$ under $f$ is at most $\varepsilon$.

Now, setting $\varepsilon = \max\bigl(E(\mathcal{K}\uplus\mathcal{Q}_{\mathcal{K}}(\bm{d}_{-})),\, E(\mathcal{K}\uplus\mathcal{Q}_{\mathcal{K}}(\bm{d}_{+}))\bigr)$ in \cref{eq:duplicate_allowed_structure_condition} yields immediately
\begin{equation}
\label{eq:duplicate_merge_max_bound}
    E(\mathcal{K}\uplus\mathcal{Q}_{\mathcal{K}}(\bm{d})) \leq \max\bigl(E(\mathcal{K}\uplus\mathcal{Q}_{\mathcal{K}}(\bm{d}_{-})),\, E(\mathcal{K}\uplus\mathcal{Q}_{\mathcal{K}}(\bm{d}_{+}))\bigr).
\end{equation}
Therefore, whenever poisons are distributed across multiple key values, we can merge their mass to one endpoint without decreasing the maximum error.
By repeatedly applying this merging operation, we obtain an optimal solution whose poisons are supported on a single integer $k\in\mathcal{K}$.
\end{proof}

\subsection{Proof of \cref{thm:m_opt_upper_bound}}

\begin{figure}[t]
    \centering
    \includegraphics[width=\columnwidth]{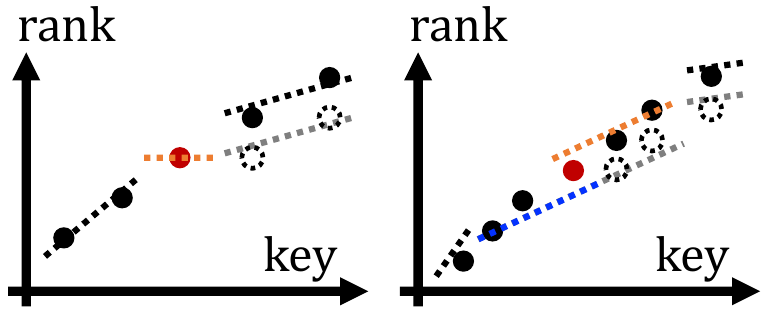}
    \caption{Proof of \cref{thm:m_opt_upper_bound}.
    (1) Case 1 (left): when $p$ lies between two adjacent segments, we add a new segment covering only $p$.
    (2) Case 2 (right): when $p$ lies inside an existing segment, we split the segment containing $p$ into two.}
    \label{fig:proof_upper_bound_agnostic}
\end{figure}

\begin{proof}[Proof of \cref{thm:m_opt_upper_bound}]
We first prove the case where one poison is added to $\mathcal{K}$; that is, we show that adding a single point $p$ to $\mathcal{K}$ satisfies $m_{\mathrm{opt}}(\mathcal{K}\cup\{p\},\varepsilon) \leq m_{\mathrm{opt}}(\mathcal{K},\varepsilon)+1$.
Now, let $\mathcal{S}=\bigl((s_1,e_1,f_1), (s_2,e_2,f_2), \dots, (s_m,e_m,f_m)\bigr)$ be a PLA for $\mathcal{K}$ achieving $m \coloneq m_{\mathrm{opt}}(\mathcal{K},\varepsilon)$.
The proof proceeds by considering the following two cases based on the position of $p$: \textbf{Case 1:} $p$ lies between two adjacent segments in $\mathcal{S}$, and \textbf{Case 2:} $p$ lies inside an existing segment in $\mathcal{S}$.

\textbf{Case 1:}
More precisely, we consider the case where there exists some $j \in [m-1]$ such that $x_{e_j} < x_p < x_{s_{j+1}}$. In this case, we can construct a PLA for $\mathcal{K} \cup \{p\}$ with $m + 1$ segments as follows (refer to the left side of \cref{fig:proof_upper_bound_agnostic}):
We retain the first $j$ segments (those preceding $p$) as they are, increase the function values of the $(j+1)$-th through $m$-th segments by one, and insert a new segment dedicated to covering $p$.
It is clear that the approximation error remains at most $\varepsilon$ for all keys in $\mathcal{K}$ and the poisoned point $p$.

\textbf{Case 2:}
More precisely, we consider the case where there exists some $j \in [m]$ such that $x_{s_j} < x_p < x_{e_j}$. In this case, a PLA with $m + 1$ segments can be constructed as follows (refer to the right side of \cref{fig:proof_upper_bound_agnostic}):
We keep the first $j-1$ segments unchanged and increase the function values of the $(j+1)$-th through $m$-th segments by one.
We then split the $j$-th segment into two parts at $x_p$, and increase the function value of the second part by one.
Under the assumption that $\varepsilon \geq 1/2$, the point $p$ can always be covered by at least one of these two segments within an error of $\varepsilon$. Thus, the approximation error for all keys and the poisoned point is guaranteed to be at most $\varepsilon$.

From \textbf{Case 1} and \textbf{Case 2}, it follows that $m_{\mathrm{opt}}(\mathcal{K}\cup\{p\},\varepsilon) \leq m_{\mathrm{opt}}(\mathcal{K},\varepsilon) + 1$. We now complete the proof for the general case using induction on $|\mathcal{P}|$.
Assume that $m_{\mathrm{opt}}(\mathcal{K}\cup\mathcal{P},\varepsilon) \leq m_{\mathrm{opt}}(\mathcal{K},\varepsilon) + k$ holds for $|\mathcal{P}| = k$.
For $|\mathcal{P}| = k+1$, let $\mathcal{P}=\mathcal{P}'\cup\{p\}$ where $|\mathcal{P}'|=k$. Then, we have:
\begin{align}
    m_{\mathrm{opt}}(\mathcal{K}\cup\mathcal{P},\varepsilon) 
    &= m_{\mathrm{opt}}((\mathcal{K}\cup\mathcal{P}')\cup\{p\},\varepsilon) \\
    &\leq m_{\mathrm{opt}}(\mathcal{K}\cup\mathcal{P}',\varepsilon) + 1 \\
    &\leq m_{\mathrm{opt}}(\mathcal{K},\varepsilon) + k + 1,
\end{align}
where the first inequality follows from the result for a single-point insertion, and the second inequality follows from the inductive hypothesis. This concludes the proof for any $\mathcal{P}$.
\end{proof}

\subsection{Proof of \cref{thm:m_opt_upper_bound_strict}}

\begin{figure}[t]
    \centering
    \includegraphics[width=\columnwidth]{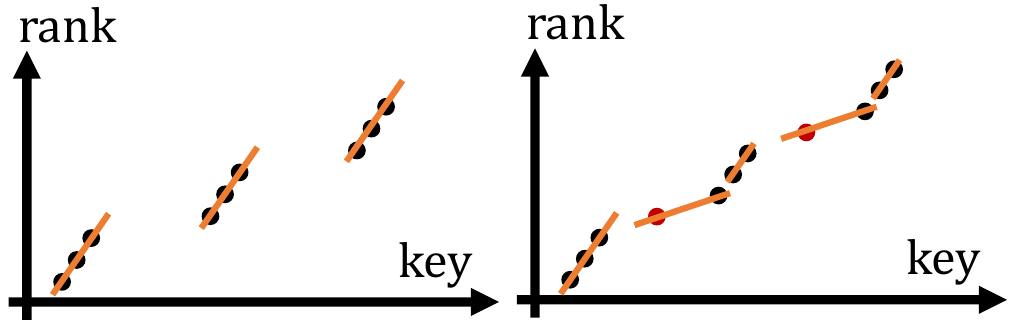}
    \caption{An example of the tight upper bound (\cref{thm:m_opt_upper_bound_strict}), where $\lambda = 2$ and $\varepsilon = 0$.
    Before poisoning (left), $m_{\mathrm{opt}}(\mathcal{K},\varepsilon) = 3$ and after poisoning (right), $m_{\mathrm{opt}}(\mathcal{K}\cup\mathcal{P},\varepsilon) = 5$.}
    \label{fig:proof_upper_bound_strict}
\end{figure}

\begin{proof}[Proof of \cref{thm:m_opt_upper_bound_strict}]
The proof is constructive. Let $b \coloneq 4\varepsilon + 3$ and $B \coloneq (2\varepsilon + 1)(2\varepsilon + 2) + 2$. We define a \textit{block} as a set of $b$ consecutive integers. We construct $\mathcal{K}$ by placing these blocks with a gap of $2B$ between each pair of adjacent blocks. The poisoning set $\mathcal{P}$ is then placed at the center of these gaps (that is, $|\mathcal{P}| = \lambda$).
Precisely, we have
\begin{align}
    \mathcal{K} &= \{2iB + j \mid i \in \{0, 1, \dots, \lambda\}, j \in \{0, 1, \dots, b-1\}\}, \\
    \mathcal{P} &= \{(2i + 1)B + (2 \varepsilon + 1) \mid i \in \{0, 1, \dots, \lambda-1\}\}.
\end{align}

In this configuration, we have $m_{\mathrm{opt}}(\mathcal{K},\varepsilon) = \lambda + 1$. This is because an optimal PLA covers each block with exactly one segment, and no single segment can cover more than one block. This property is guaranteed by the fact that $b$ and $B$ are chosen to be sufficiently large.

Furthermore, it can be shown that $m_{\mathrm{opt}}(\mathcal{K}\cup\mathcal{P},\varepsilon) = 2\lambda + 1$. This can be verified by greedily determining the segments from left to right:
\begin{itemize}[leftmargin=1.5em]
    \item The $1$-st segment covers the first block of $\mathcal{K}$.
    \item For $i \in \{1, 2, \dots, \lambda\}$, the $2i$-th segment covers the $i$-th poisoned point in $\mathcal{P}$ and the first $2\varepsilon + 1$ elements of the $(i+1)$-th block.
    \item For $i \in \{1, 2, \dots, \lambda\}$, the $(2i + 1)$-th segment covers the remaining $2\varepsilon + 2$ elements of the $(i+1)$-th block.
\end{itemize}

Thus, in this specific case, the equality in \cref{eq:m_opt_upper_bound_strict} is satisfied, completing the proof.
\end{proof}

\section{\texorpdfstring{Details on \textsc{DI-Consecutive}}{Details on DI-Consecutive}}
\label{app:detail_discrete_intercept_consec}

In this section, we provide the technical details omitted from the main text regarding \textsc{DI-Consecutive} (Discrete-Intercept Consecutive method) proposed in \cref{sec:methods_for_obtaining_poison_solutions_num_coverable_keys}.
The overall procedure is summarized in \cref{alg:di_consecutive}.

\begin{algorithm}[t]
    \caption{\textsc{DI-Consecutive}}
    \label{alg:di_consecutive}
    \begin{algorithmic}[1]
    \Input Legitimate key set $\mathcal{K}$, poison budget $\lambda$
    \Output Poison set $\mathcal{P}^\ast$
    \State Precompute slope intervals $\mathcal{I}$ and build sparse tables
    \State $C^\ast \gets C_\varepsilon(\mathcal{K}, \emptyset)$, $\mathcal{P}^\ast \gets \emptyset$
    \For{each candidate consecutive poison set $\mathcal{P}$ with $|\mathcal{P}| = \lambda$}
        \State $C_\mathcal{P} \gets C_\varepsilon(\mathcal{K}, \mathcal{P})$ \Comment{time $\mathcal{O}(|\mathcal{I}|\log c)$ via sparse tables}
        \If{$C_\mathcal{P} < C^\ast$}
            \State $C^\ast \gets C_\mathcal{P}$, $\mathcal{P}^\ast \gets \mathcal{P}$
        \EndIf
    \EndFor
    \State \Return $\mathcal{P}^\ast$
    \end{algorithmic}
\end{algorithm}

\textbf{Main Idea.}
The core objective of \textsc{DI-Consecutive} is to optimize the evaluation of each candidate of consecutive poisons, i.e., to compute the number of covered legitimate keys.
We reduce the computational complexity from the conventional $\mathcal{O}(c)$, where $c$ is the number of covered legitimate keys without poisoning, to $\mathcal{O}(|\mathcal{I}| \log c)$, where $|\mathcal{I}|$ is the number of candidate intercepts.
In typical settings where $c \geq 10^4$ and $|\mathcal{I}| \approx 41$, this improvement is substantial.

\Cref{fig:swing_detail} illustrates the intuition behind our approach.
The figure shows the evaluation of a candidate consecutive poisons starting from $k_3+1$.
\textsc{DI-Consecutive} independently computes the range of slopes that can cover:
(1) legitimate keys preceding the poisons,
(2) legitimate keys following the poisons, and
(3) the consecutive integer sequence containing the poisons.
The candidate poisons and keys are coverable if and only if these three slope ranges have a non-empty intersection.

\begin{figure}[t]
    \centering
    \includegraphics[width=\columnwidth]{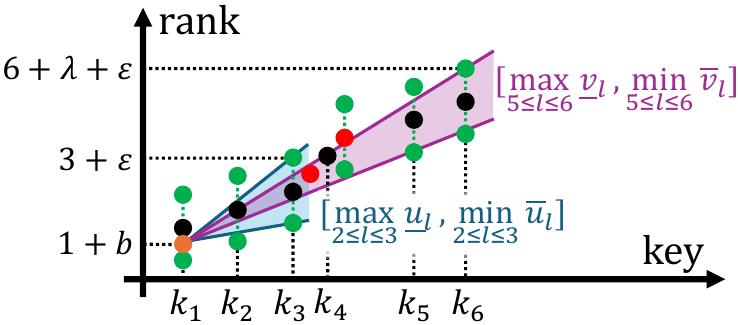}
    \caption{
        An example of \textsc{DI-Consecutive} evaluating a candidate poison sequence (starting from $k_3+1$ with $\lambda=2$). 
        For a specific intercept $b$, slope ranges for pre-poison ($k_2, k_3$) and post-poison ($k_5, k_6$) keys are retrieved in $\mathcal{O}(1)$ via sparse tables.
        The slope range for the poisons is also computed in $\mathcal{O}(1)$ by finding the range covering $k_3 + 1$ and the largest poison.
    }
    \label{fig:swing_detail}
\end{figure}

\subsection{Precomputation}
For each candidate intercept $b \in \mathcal{I}$, we perform a precomputation step with a time complexity of $\mathcal{O}(c \log c)$.
Consequently, the total precomputation complexity is $\mathcal{O}(|\mathcal{I}| c \log c)$.
Below, we describe the precomputation step for a particular $b$.

First, we compute $\underline{u}_i, \overline{u}_i, \underline{v}_i, \overline{v}_i$ for $i \in \{2, 3, \dots, c\}$ as follows:
\begin{equation}
    \underline{u}_i \coloneq \frac{(i - \varepsilon) - (1 + b)}{k_i - k_1}, \quad \overline{u}_i \coloneq \frac{(i + \varepsilon) - (1 + b)}{k_i - k_1},
\end{equation}
\begin{equation}
    \underline{v}_i \coloneq \frac{(i + \lambda - \varepsilon) - (1 + b)}{k_i - k_1}, \quad \overline{v}_i \coloneq \frac{(i + \lambda + \varepsilon) - (1 + b)}{k_i - k_1}.
\end{equation}
The interval $[\underline{u}_i, \overline{u}_i]$ represents the slope range required to cover $k_i$ when it precedes the poisons.
For example, in \cref{fig:swing_detail}, the slope of the line connecting $(k_1, 1+b)$ and $(k_3, 3+\varepsilon)$ corresponds to $\overline{u}_3$.
Similarly, $[\underline{v}_i, \overline{v}_i]$ represents the slope range required to cover the consecutive integer sequence containing the poisons when $k_i$ follows the poisons.
For instance, in \cref{fig:swing_detail}, the slope of the line connecting $(k_1, 1+b)$ and $(k_6, 6+\lambda+\varepsilon)$ corresponds to $\overline{v}_6$.
These values can be computed in $\mathcal{O}(c)$ time.

Furthermore, we construct sparse tables over these four arrays to support $\mathcal{O}(1)$ Range Maximum Queries for $\bm{\underline{u}}$ and $\bm{\underline{v}}$, and $\mathcal{O}(1)$ Range Minimum Queries for $\bm{\overline{u}}$ and $\bm{\overline{v}}$.
This construction requires $\mathcal{O}(c \log c)$ time.

\subsection{Evaluation}
The evaluation phase computes, for a given candidate consecutive poison set $\mathcal{P}$, the maximum number of legitimate keys that can be covered.

To achieve a time complexity of $\mathcal{O}(|\mathcal{I}| \log c)$ for each candidate consecutive poison set, we design the evaluation so that, for each intercept $b \in \mathcal{I}$, the number of coverable legitimate keys can be computed in $\mathcal{O}(\log c)$ time.
This $\mathcal{O}(\log c)$ computation is realized via binary search combined with an $\mathcal{O}(1)$ feasibility check, which determines whether, for the given poison set, intercept, and keys, there exists a slope under which they are coverable.

Below, we first describe the \textbf{Maximal Coverage Search}, which computes in $\mathcal{O}(\log c)$ time the number of legitimate keys coverable for a fixed poison set and intercept.
We then explain the \textbf{Coverability Verification}, which determines in $\mathcal{O}(1)$ time whether the given poison set, intercept, and keys are feasible.

\textbf{Maximal Coverage Search.}
First, for a given candidate consecutive poison set $\mathcal{P}$, we find the smallest legitimate key greater than the largest poison in the set, and denote its index by $j$.
For instance, in \cref{fig:swing_detail}, $j = 5$.
A straightforward method would inspect the integers not occupied by legitimate keys in increasing order.
In the worst case, this takes $\mathcal{O}(c+\lambda)$ time.
To reduce this cost, we precompute the number of unused integers between each pair of consecutive legitimate keys, as well as their prefix sums.
Using this information, we can perform a binary search to determine where the largest poison falls relative to the gaps between legitimate keys, thereby identifying the first legitimate key after the poisons in $\mathcal{O}(\log c)$ time.

Next, for each intercept $b \in \mathcal{I}$, we perform a binary search to find the maximum $r \in \{j, \dots, c\}$ such that $\{k_1, \dots, k_r\} \cup \mathcal{P}$ are coverable.
At each step, we invoke the \textbf{Coverability Verification} (described below) to determine whether $\{k_1, \dots, k_r\} \cup \mathcal{P}$ are coverable.
This yields the maximum feasible $r$ for the given intercept in $\mathcal{O}(\log c)$ time.

After evaluating all $b \in \mathcal{I}$, we select the intercept that yields the largest $r$, which corresponds to the maximum number of covered legitimate keys for the given poison set.

\textbf{Coverability Verification.}
The binary search requires verifying whether the set $\{k_1, \dots, k_r\} \cup \mathcal{P}$ is coverable for a given intercept $b$. We determine the existence of a valid slope in $\mathcal{O}(1)$ time. Such a slope exists if and only if the following three ranges have a non-empty intersection:
(1) the pre-poison range,
(2) the post-poison range, and
(3) the poison region range.

(1) The pre-poison range is the slope range covering legitimate keys preceding the poisons (i.e., from $k_2$ to $k_i$). This is computed as
\begin{equation}
\left[ \max_{2 \leq l \leq i} \underline{u}_l,\; \min_{2 \leq l \leq i} \overline{u}_l \right].
\end{equation}
We can obtain this in $\mathcal{O}(1)$ time using a sparse table constructed during the precomputation step.

(2) The post-poison range is the slope range covering legitimate keys following the poisons (i.e., from $k_j$ to $k_r$). This is computed as
\begin{equation}
\left[ \max_{j \leq l \leq r} \underline{v}_l,\; \min_{j \leq l \leq r} \overline{v}_l \right].
\end{equation}
We can obtain this in $\mathcal{O}(1)$ time using the precomputed sparse table.

(3) The poison region range is the slope range covering the consecutive integers including the poisons (i.e., the consecutive integers from $k_i + 1$ to the largest poison). This slope range is obtained by intersecting the slope range that covers $k_i + 1$ with the slope range that covers the largest poison. This is because, by linearity, if the error at both endpoints is at most $\varepsilon$, then the error at every point between them is also guaranteed to be at most $\varepsilon$.

\section{Details on the Instance-Dependent Upper Bound Algorithm}
\label{app:instance_upper_bound_algorithm}

\begin{figure}[t]
    \centering
    \includegraphics[width=\columnwidth]{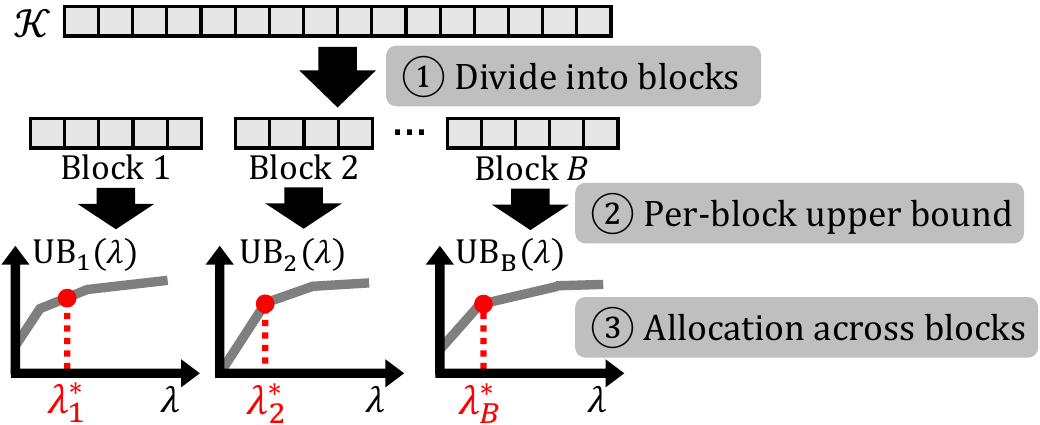}
    \caption{Overview of the instance-dependent upper-bound algorithm: partition $\mathcal{K}$ into blocks with fixed slopes, compute per-block upper bounds, and allocate the poison budget across blocks to maximize their sum.}
    \label{fig:upper_bound_algorithm_overview}
\end{figure}

In this section, we provide the details omitted in the main text regarding the instance-dependent upper-bound algorithm proposed in \cref{sec:m_opt_instance_dependent_upper_bound}.
\cref{fig:upper_bound_algorithm_overview} illustrates the pipeline: partitioning into blocks, upper-bounding the number of segments in each block, and then allocating the total poison budget across blocks.

\subsection{Dividing into Blocks}

We partition $\mathcal{K}$ into blocks using an $\alpha\varepsilon$-PLA with $\alpha\geq1$ and fix the slope within each block to that of the corresponding PLA segment.
This restricts the PLA construction by requiring each block to use a predetermined slope and by treating the blocks independently.
We also strengthen the attacker by allowing duplicate poisons.
Since these modifications weaken the index constructor and strengthen the attacker, the maximum number of segments in the resulting relaxed problem upper-bounds that in the original problem.

In our implementation, we evaluate the upper bound for $\alpha \in \{1.0,1.2,1.4,1.6,1.8,2.0\}$ and take the minimum.
Since each value of $\alpha$ yields a valid upper bound, taking their minimum preserves the upper-bound guarantee.

\subsection{Per-Block Upper Bound}

\begin{figure}[t]
    \centering
    \includegraphics[width=\columnwidth]{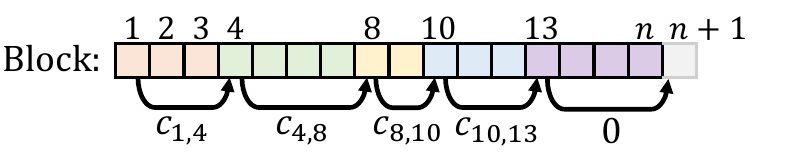}
    \caption{Example of a path in the DAG for upper-bounding the number of segments in a block.
    A path represents a segmentation; its length gives the number of segments, and its total weight lower-bounds the required number of poisons.}
    \label{fig:upper_bound_algorithm_graph}
\end{figure}

We next derive an upper bound on the number of segments in a single block.
Let $\{k_1,k_2,\dots,k_n\}$ be the keys in the block, and let $w$ be its fixed slope.

First, we prove the following lemma concerning the minimum number of poisons required to separate two keys.
\begin{theorem_box}
\begin{lemma}
\label{lem:closed_form_cij}
Let $\{k_1, k_2, \dots, k_n\}$ be the keys in the block, and let $w$ be its fixed slope.
Define $A_l \coloneq l - w k_l$ for $l \in [n]$.
Then, the attacker requires at least $\lceil c_{i,j} \rceil$ poisons to make the segment starting at $k_i$ end before $k_j$, where
\begin{equation}
\label{eq:num_poison_to_separate_keys}
    c_{i,j} \coloneq \max\left(0, 2\varepsilon - \left(\max_{i \leq l \leq j} A_l - \min_{i \leq l \leq j} A_l \right)\right).
\end{equation}
\end{lemma}
\end{theorem_box}
\begin{proof}[Proof of \cref{lem:closed_form_cij}]
Under a fixed slope $w$, let $b$ be the optimal intercept for the set of keys $\{k_i, \dots, k_j\}$.
The prediction error of the position for each key $k_l$ ($i \leq l \leq j$) is expressed as $|(w k_l + b) - l| = |b - (l - w k_l)| = |b - A_l|$.
Therefore, the optimal intercept that minimizes the maximum error in this segment is $b^\ast = \frac{1}{2}(\max_{i \leq l \leq j} A_l + \min_{i \leq l \leq j} A_l)$, and the initial maximum error $E_{\mathrm{init}}$ before poisoning is given by:
\begin{equation}
    E_{\mathrm{init}} = \frac{1}{2}\left(\max_{i \leq l \leq j} A_l - \min_{i \leq l \leq j} A_l\right)
\end{equation}

Moreover, we can show that each duplicate poison can increase the maximum error by at most $1/2$.
This holds because the linear regression can limit the increase in the maximum error to $1/2$ by increasing the intercept $b$ by $1/2$ for each added poison.

Therefore, at least
\begin{equation}
    \left\lceil \max\left(0, \frac{\varepsilon - E_{\mathrm{init}}}{1/2} \right) \right\rceil = \lceil c_{i,j} \rceil
\end{equation}
poisons are needed to make the error exceed $\varepsilon$.
\end{proof}

We now formulate the per-block problem as a path problem on a DAG.
The DAG contains nodes $1,\dots,n+1$, where node $n+1$ is a sentinel representing the end of the block.
For each $1\leq i<j\leq n$, we add an edge from $i$ to $j$ with weight $c_{i,j}$.
We also add an edge of weight zero from each node $i\in[n]$ to the sentinel node $n+1$; that is, define $c_{i,n+1} = 0$ for all $i \in [n]$.

As illustrated in \cref{fig:upper_bound_algorithm_graph}, each path from node $1$ to node $n+1$ represents a segmentation of the block.
The starting node of each edge corresponds to the starting position of a segment, so the number of edges in the path equals the number of segments.
By \cref{lem:closed_form_cij}, the weight $c_{i,j}$ lower-bounds the number of poisons required to force the segment starting at $k_i$ to end before $k_j$.
The final edge has weight zero because the last segment naturally terminates at the end of the block.
Consequently, the total weight of a path lower-bounds the number of poisons required to realize the corresponding segmentation.

Next, to reduce the computational complexity of finding shortest paths in the DAG, we prove the following lemma regarding the Monge property of $c_{i,j}$.

\begin{theorem_box}
\begin{lemma}
\label{lem:monge_property_cij}
The cost function $c_{i,j}$ satisfies the Monge property, i.e., for any $1 \leq i < i' < j < j' \leq n$,
\begin{equation}
    c_{i,j} + c_{i',j'} \leq c_{i,j'} + c_{i',j}.
\end{equation}
\end{lemma}
\end{theorem_box}
\begin{proof}[Proof of \cref{lem:monge_property_cij}]
Let $D(i, j) \coloneq \max_{i \leq l \leq j} A_l - \min_{i \leq l \leq j} A_l$.
For any $i < i' < j < j'$, consider the intervals $I_1 = [i, j]$ and $I_2 = [i', j']$.
Their union is $I_1 \cup I_2 = [i, j']$ and their intersection is $I_1 \cap I_2 = [i', j]$.
By the basic properties of the range over intervals, the sum of the ranges of the union and intersection is upper-bounded by the sum of the ranges of the original intervals:
\begin{equation}
\label{eq:range_submodular}
    D(i, j') + D(i', j) \leq D(i, j) + D(i', j').
\end{equation}
Furthermore, the range is monotonically non-decreasing with respect to interval inclusion, meaning $D(i, j') \geq \max(D(i, j), D(i', j'))$.
Since $c_{p,q}=\max(0,\,2\varepsilon-D(p,q))$, these two properties together imply $c_{i,j} + c_{i',j'} \leq c_{i,j'} + c_{i',j}$, concluding the proof.
\end{proof}

We finally derive an upper bound on the number of segments under a poison budget.
For a path $P$ from node $1$ to node $n+1$, let $|P|$ denote its number of edges and let
\begin{equation}
    C(P) \coloneq \sum_{(i,j)\in P}c_{i,j}
\end{equation}
denote its total weight.
For any parameter $\nu>0$, define
\begin{equation}
\label{eq:lagrangian_distance}
    D(\nu) \coloneq \min_P{(C(P)-\nu|P|)}
\end{equation}
Equivalently, $D(\nu)$ is the shortest-path distance from node $1$ to node $n+1$ after subtracting $\nu$ from every edge weight in the DAG.

For any path $P$ realizable with at most $\lambda$ poisons, we have
$C(P)\leq\lambda$.
By the definition of $D(\nu)$,
\begin{equation}
    D(\nu) \leq C(P)-\nu|P| \leq \lambda-\nu|P|.
\end{equation}
Therefore, for any path $P$ realizable with at most $\lambda$ poisons,
\begin{equation}
\label{eq:per_block_lagrangian_bound}
    |P| \leq \frac{\lambda-D(\nu)}{\nu}.
\end{equation}
Thus, \cref{eq:per_block_lagrangian_bound} gives an upper bound on the number of segments achievable within the block.

Subtracting the same constant $\nu$ from every edge preserves the Monge property.
Hence, $D(\nu)$ can be computed in $\mathcal{O}(n)$ time using the LARSCH algorithm~\cite{aggarwal1986geometric,larmore1991line}.
In our experiments, we evaluate the bound for $\nu \in \mathcal{N} \coloneq \{1,2,3,4,5,10,20,40,80,160\}$ and take the minimum.
For a block $b$ and an allocated budget $\lambda$, we denote the resulting bound by $\mathrm{UB}_b(\lambda)$, defined as follows:
\begin{equation}
\label{eq:per_block_upper_bound_function}
    \mathrm{UB}_b(\lambda) \coloneq \min_{\nu\in\mathcal{N}} \frac{\lambda-D_b(\nu)}{\nu},
\end{equation}
where $D_b(\nu)$ is the shortest-path distance from node $1$ to node $n+1$ after subtracting $\nu$ from every edge weight in the DAG for block $b$.

\subsection{Allocation Across Blocks}

Suppose that $\mathcal{K}$ is divided into $B$ blocks.
Let $\lambda_b$ be the number of poisons allocated to block $b$.
Since the total poison budget is $\lambda$, the allocations satisfy
\begin{equation}
    \sum_{b=1}^{B}\lambda_b \leq \lambda.
\end{equation}
Using the per-block bound in
\cref{eq:per_block_upper_bound_function}, the global upper bound is obtained by solving
\begin{equation}
\label{eq:global_poison_allocation}
    \max_{\substack{\lambda_1,\dots,\lambda_B\in\mathbb{Z}_{\geq0}\\ \sum_{b=1}^{B}\lambda_b\leq\lambda}} \sum_{b=1}^{B}\mathrm{UB}_b(\lambda_b).
\end{equation}

Each $\mathrm{UB}_b(\lambda_b)$ is the pointwise minimum of affine functions of
$\lambda_b$ and is therefore concave.
Consequently, its marginal gain
\begin{equation}
    \mathrm{UB}_b(\lambda_b+1)-\mathrm{UB}_b(\lambda_b)
\end{equation}
is non-increasing in $\lambda_b$.
The allocation problem in \cref{eq:global_poison_allocation} can therefore be solved optimally by repeatedly assigning one poison to the block with the largest current marginal gain~\cite{federgruen1986greedy}.
The sum of the resulting per-block bounds gives the instance-dependent upper bound for the selected block partition.

\section{Additional Experiments}
\label{app:additional_experiments}

\subsection{Additional Results on PLA Segment Maximization}
\label{app:additional_mopt_results}

Here, we present a comprehensive set of results on the post-poisoning $m_{\mathrm{opt}}$ and the upper bounds we derive, which were omitted from the main text.
As attack methods, in addition to \textsc{PGM-attack}, we also evaluate \textsc{Random} and \textsc{Random-Adjacent}, the same baselines used in \cref{sec:experiment_poisoning_max_error}.
As upper bounds, we evaluate the instance-agnostic upper bound (\textbf{Agnostic UB}; \cref{sec:m_opt_instance_agnostic_upper_bound}) and the instance-dependent upper bound (\textbf{Instance UB}; \cref{sec:m_opt_instance_dependent_upper_bound}).
We use $\varepsilon \in \{16,32,64,128\}$.
The numbers in parentheses indicate the factor relative to $m_{\mathrm{opt}}$ for the legitimate keys.

\subsubsection{Complete Results on Poisoning and Upper Bounds}
\label{app:complete_mopt_results}

\begin{table*}[t]
    \centering
    \caption{Results for $m_{\mathrm{opt}}$ after poisoning. Across all datasets and settings of $\varepsilon$ and $\lambda$, our \textsc{PGM-attack} causes a substantially larger increase than randomly inserting poison keys.}
    \label{tab:poisoning_impact_on_m_opt_all}
    \setlength{\tabcolsep}{2pt}
    \footnotesize
    \begin{tabular}{@{}c l r rrr rr rrr rr@{}}
        \toprule
         & & & \multicolumn{5}{c}{$\lambda = 0.01\,n$} & \multicolumn{5}{c}{$\lambda = 0.1\,n$} \\
        \cmidrule(lr){4-8} \cmidrule(lr){9-13}
         & & & \multicolumn{3}{c}{Poisoning} & \multicolumn{2}{c}{Upper Bound} & \multicolumn{3}{c}{Poisoning} & \multicolumn{2}{c}{Upper Bound} \\
        \cmidrule(lr){4-6} \cmidrule(lr){7-8} \cmidrule(lr){9-11} \cmidrule(lr){12-13}
        $\varepsilon$ & Dataset & Original & PGM-attack & Random & Random-Adj. & Agnostic UB & Instance UB & PGM-attack & Random & Random-Adj. & Agnostic UB & Instance UB \\
        \midrule
        16 & Amzn & 7.94M & 8.98M {\footnotesize (1.13$\times$)} & 7.95M {\footnotesize (1.00$\times$)} & 8.06M {\footnotesize (1.01$\times$)} & 15.9M {\footnotesize (2.01$\times$)} & 14.5M {\footnotesize (1.82$\times$)} & 12.1M {\footnotesize (1.52$\times$)} & 8.01M {\footnotesize (1.01$\times$)} & 9.12M {\footnotesize (1.15$\times$)} & 87.9M {\footnotesize (11.1$\times$)} & 21.2M {\footnotesize (2.67$\times$)} \\
         & Osmc & 6.18M & 6.76M {\footnotesize (1.10$\times$)} & 6.19M {\footnotesize (1.00$\times$)} & 6.25M {\footnotesize (1.01$\times$)} & 14.2M {\footnotesize (2.30$\times$)} & 11.2M {\footnotesize (1.81$\times$)} & 9.18M {\footnotesize (1.49$\times$)} & 6.28M {\footnotesize (1.02$\times$)} & 6.92M {\footnotesize (1.12$\times$)} & 86.2M {\footnotesize (14.0$\times$)} & 17.4M {\footnotesize (2.82$\times$)} \\
         & Face & 2.12M & 2.33M {\footnotesize (1.10$\times$)} & 2.12M {\footnotesize (1.00$\times$)} & 2.14M {\footnotesize (1.01$\times$)} & 4.12M {\footnotesize (1.94$\times$)} & 3.64M {\footnotesize (1.71$\times$)} & 3.08M {\footnotesize (1.45$\times$)} & 2.14M {\footnotesize (1.01$\times$)} & 2.34M {\footnotesize (1.10$\times$)} & 22.1M {\footnotesize (10.4$\times$)} & 5.36M {\footnotesize (2.53$\times$)} \\
         & YCSB & 70.1K & 176K {\footnotesize (2.51$\times$)} & 72.5K {\footnotesize (1.03$\times$)} & 74.1K {\footnotesize (1.06$\times$)} & 2.07M {\footnotesize (29.5$\times$)} & 352K {\footnotesize (5.02$\times$)} & 810K {\footnotesize (11.6$\times$)} & 94.8K {\footnotesize (1.35$\times$)} & 111K {\footnotesize (1.58$\times$)} & 20.1M {\footnotesize (286$\times$)} & 1.3M {\footnotesize (18.6$\times$)} \\
         & Longitudes & 164K & 262K {\footnotesize (1.59$\times$)} & 166K {\footnotesize (1.01$\times$)} & 169K {\footnotesize (1.03$\times$)} & 2.16M {\footnotesize (13.2$\times$)} & 514K {\footnotesize (3.13$\times$)} & 831K {\footnotesize (5.05$\times$)} & 180K {\footnotesize (1.10$\times$)} & 208K {\footnotesize (1.26$\times$)} & 20.2M {\footnotesize (123$\times$)} & 1.46M {\footnotesize (8.86$\times$)} \\
         & Longlat & 452K & 543K {\footnotesize (1.20$\times$)} & 456K {\footnotesize (1.01$\times$)} & 459K {\footnotesize (1.01$\times$)} & 2.45M {\footnotesize (5.42$\times$)} & 992K {\footnotesize (2.19$\times$)} & 1.1M {\footnotesize (2.44$\times$)} & 473K {\footnotesize (1.05$\times$)} & 516K {\footnotesize (1.14$\times$)} & 20.5M {\footnotesize (45.2$\times$)} & 2.01M {\footnotesize (4.45$\times$)} \\
         & Uniform & 203K & 320K {\footnotesize (1.57$\times$)} & 209K {\footnotesize (1.03$\times$)} & 209K {\footnotesize (1.03$\times$)} & 2.2M {\footnotesize (10.8$\times$)} & 615K {\footnotesize (3.02$\times$)} & 930K {\footnotesize (4.57$\times$)} & 262K {\footnotesize (1.29$\times$)} & 264K {\footnotesize (1.30$\times$)} & 20.2M {\footnotesize (99.3$\times$)} & 1.64M {\footnotesize (8.04$\times$)} \\
         & Normal & 203K & 320K {\footnotesize (1.57$\times$)} & 205K {\footnotesize (1.01$\times$)} & 209K {\footnotesize (1.03$\times$)} & 2.2M {\footnotesize (10.8$\times$)} & 614K {\footnotesize (3.02$\times$)} & 929K {\footnotesize (4.58$\times$)} & 224K {\footnotesize (1.10$\times$)} & 264K {\footnotesize (1.30$\times$)} & 20.2M {\footnotesize (99.5$\times$)} & 1.64M {\footnotesize (8.06$\times$)} \\
         & Lognormal & 203K & 320K {\footnotesize (1.57$\times$)} & 205K {\footnotesize (1.01$\times$)} & 209K {\footnotesize (1.03$\times$)} & 2.2M {\footnotesize (10.8$\times$)} & 614K {\footnotesize (3.02$\times$)} & 929K {\footnotesize (4.58$\times$)} & 224K {\footnotesize (1.10$\times$)} & 264K {\footnotesize (1.30$\times$)} & 20.2M {\footnotesize (99.5$\times$)} & 1.64M {\footnotesize (8.06$\times$)} \\
        \midrule
        32 & Amzn & 2.46M & 2.97M {\footnotesize (1.21$\times$)} & 2.47M {\footnotesize (1.00$\times$)} & 2.52M {\footnotesize (1.02$\times$)} & 10.5M {\footnotesize (4.25$\times$)} & 5.15M {\footnotesize (2.09$\times$)} & 4.79M {\footnotesize (1.94$\times$)} & 2.49M {\footnotesize (1.01$\times$)} & 3.03M {\footnotesize (1.23$\times$)} & 82.5M {\footnotesize (33.5$\times$)} & 9.19M {\footnotesize (3.73$\times$)} \\
         & Osmc & 2.89M & 3.14M {\footnotesize (1.09$\times$)} & 2.89M {\footnotesize (1.00$\times$)} & 2.92M {\footnotesize (1.01$\times$)} & 10.9M {\footnotesize (3.77$\times$)} & 5.31M {\footnotesize (1.84$\times$)} & 4.3M {\footnotesize (1.49$\times$)} & 2.92M {\footnotesize (1.01$\times$)} & 3.22M {\footnotesize (1.11$\times$)} & 82.9M {\footnotesize (28.7$\times$)} & 8.26M {\footnotesize (2.86$\times$)} \\
         & Face & 1.06M & 1.16M {\footnotesize (1.10$\times$)} & 1.06M {\footnotesize (1.00$\times$)} & 1.07M {\footnotesize (1.01$\times$)} & 3.06M {\footnotesize (2.90$\times$)} & 1.82M {\footnotesize (1.72$\times$)} & 1.53M {\footnotesize (1.45$\times$)} & 1.06M {\footnotesize (1.01$\times$)} & 1.16M {\footnotesize (1.10$\times$)} & 21.1M {\footnotesize (20.0$\times$)} & 2.69M {\footnotesize (2.55$\times$)} \\
         & YCSB & 25.1K & 63.9K {\footnotesize (2.54$\times$)} & 25.5K {\footnotesize (1.02$\times$)} & 25.9K {\footnotesize (1.03$\times$)} & 2.03M {\footnotesize (80.6$\times$)} & 113K {\footnotesize (4.51$\times$)} & 361K {\footnotesize (14.4$\times$)} & 29.3K {\footnotesize (1.17$\times$)} & 32.8K {\footnotesize (1.31$\times$)} & 20M {\footnotesize (797$\times$)} & 563K {\footnotesize (22.4$\times$)} \\
         & Longitudes & 62.6K & 102K {\footnotesize (1.63$\times$)} & 62.9K {\footnotesize (1.00$\times$)} & 63.7K {\footnotesize (1.02$\times$)} & 2.06M {\footnotesize (32.9$\times$)} & 203K {\footnotesize (3.25$\times$)} & 380K {\footnotesize (6.07$\times$)} & 66.1K {\footnotesize (1.06$\times$)} & 74.3K {\footnotesize (1.19$\times$)} & 20.1M {\footnotesize (320$\times$)} & 653K {\footnotesize (10.4$\times$)} \\
         & Longlat & 219K & 259K {\footnotesize (1.18$\times$)} & 221K {\footnotesize (1.01$\times$)} & 222K {\footnotesize (1.01$\times$)} & 2.22M {\footnotesize (10.1$\times$)} & 475K {\footnotesize (2.17$\times$)} & 537K {\footnotesize (2.45$\times$)} & 226K {\footnotesize (1.03$\times$)} & 244K {\footnotesize (1.11$\times$)} & 20.2M {\footnotesize (92.2$\times$)} & 969K {\footnotesize (4.42$\times$)} \\
         & Uniform & 53K & 103K {\footnotesize (1.94$\times$)} & 54.5K {\footnotesize (1.03$\times$)} & 54.6K {\footnotesize (1.03$\times$)} & 2.05M {\footnotesize (38.7$\times$)} & 209K {\footnotesize (3.95$\times$)} & 391K {\footnotesize (7.38$\times$)} & 68.4K {\footnotesize (1.29$\times$)} & 69.1K {\footnotesize (1.30$\times$)} & 20.1M {\footnotesize (378$\times$)} & 670K {\footnotesize (12.6$\times$)} \\
         & Normal & 52.9K & 103K {\footnotesize (1.94$\times$)} & 53.4K {\footnotesize (1.01$\times$)} & 54.5K {\footnotesize (1.03$\times$)} & 2.05M {\footnotesize (38.8$\times$)} & 209K {\footnotesize (3.95$\times$)} & 391K {\footnotesize (7.39$\times$)} & 58.2K {\footnotesize (1.10$\times$)} & 68.8K {\footnotesize (1.30$\times$)} & 20.1M {\footnotesize (379$\times$)} & 670K {\footnotesize (12.7$\times$)} \\
         & Lognormal & 52.9K & 103K {\footnotesize (1.94$\times$)} & 53.5K {\footnotesize (1.01$\times$)} & 54.5K {\footnotesize (1.03$\times$)} & 2.05M {\footnotesize (38.8$\times$)} & 209K {\footnotesize (3.95$\times$)} & 391K {\footnotesize (7.39$\times$)} & 58.2K {\footnotesize (1.10$\times$)} & 68.8K {\footnotesize (1.30$\times$)} & 20.1M {\footnotesize (379$\times$)} & 670K {\footnotesize (12.7$\times$)} \\
        \midrule
        64 & Amzn & 797K & 962K {\footnotesize (1.21$\times$)} & 798K {\footnotesize (1.00$\times$)} & 810K {\footnotesize (1.02$\times$)} & 8.8M {\footnotesize (11.0$\times$)} & 1.73M {\footnotesize (2.17$\times$)} & 1.73M {\footnotesize (2.17$\times$)} & 804K {\footnotesize (1.01$\times$)} & 929K {\footnotesize (1.16$\times$)} & 80.8M {\footnotesize (101$\times$)} & 3.24M {\footnotesize (4.06$\times$)} \\
         & Osmc & 1.38M & 1.49M {\footnotesize (1.08$\times$)} & 1.38M {\footnotesize (1.00$\times$)} & 1.39M {\footnotesize (1.01$\times$)} & 9.38M {\footnotesize (6.82$\times$)} & 2.55M {\footnotesize (1.85$\times$)} & 2.06M {\footnotesize (1.50$\times$)} & 1.39M {\footnotesize (1.01$\times$)} & 1.52M {\footnotesize (1.11$\times$)} & 81.4M {\footnotesize (59.2$\times$)} & 3.97M {\footnotesize (2.88$\times$)} \\
         & Face & 523K & 577K {\footnotesize (1.10$\times$)} & 523K {\footnotesize (1.00$\times$)} & 528K {\footnotesize (1.01$\times$)} & 2.52M {\footnotesize (4.82$\times$)} & 905K {\footnotesize (1.73$\times$)} & 760K {\footnotesize (1.45$\times$)} & 524K {\footnotesize (1.00$\times$)} & 575K {\footnotesize (1.10$\times$)} & 20.5M {\footnotesize (39.2$\times$)} & 1.35M {\footnotesize (2.57$\times$)} \\
         & YCSB & 6.95K & 25.6K {\footnotesize (3.68$\times$)} & 7.12K {\footnotesize (1.02$\times$)} & 7.36K {\footnotesize (1.06$\times$)} & 2.01M {\footnotesize (289$\times$)} & 46.6K {\footnotesize (6.71$\times$)} & 169K {\footnotesize (24.2$\times$)} & 8.73K {\footnotesize (1.26$\times$)} & 10.3K {\footnotesize (1.49$\times$)} & 20M {\footnotesize (2877$\times$)} & 272K {\footnotesize (39.1$\times$)} \\
         & Longitudes & 27.8K & 45.6K {\footnotesize (1.64$\times$)} & 27.9K {\footnotesize (1.00$\times$)} & 28.2K {\footnotesize (1.01$\times$)} & 2.03M {\footnotesize (72.9$\times$)} & 87K {\footnotesize (3.13$\times$)} & 185K {\footnotesize (6.66$\times$)} & 28.7K {\footnotesize (1.03$\times$)} & 31.7K {\footnotesize (1.14$\times$)} & 20M {\footnotesize (720$\times$)} & 312K {\footnotesize (11.2$\times$)} \\
         & Longlat & 112K & 131K {\footnotesize (1.17$\times$)} & 113K {\footnotesize (1.01$\times$)} & 113K {\footnotesize (1.01$\times$)} & 2.11M {\footnotesize (18.9$\times$)} & 238K {\footnotesize (2.12$\times$)} & 271K {\footnotesize (2.42$\times$)} & 114K {\footnotesize (1.02$\times$)} & 123K {\footnotesize (1.10$\times$)} & 20.1M {\footnotesize (180$\times$)} & 483K {\footnotesize (4.32$\times$)} \\
         & Uniform & 13.5K & 34.9K {\footnotesize (2.60$\times$)} & 13.9K {\footnotesize (1.03$\times$)} & 13.8K {\footnotesize (1.03$\times$)} & 2.01M {\footnotesize (150$\times$)} & 69K {\footnotesize (5.12$\times$)} & 177K {\footnotesize (13.2$\times$)} & 17.5K {\footnotesize (1.30$\times$)} & 17.6K {\footnotesize (1.30$\times$)} & 20M {\footnotesize (1487$\times$)} & 294K {\footnotesize (21.8$\times$)} \\
         & Normal & 13.5K & 34.9K {\footnotesize (2.59$\times$)} & 13.6K {\footnotesize (1.01$\times$)} & 13.9K {\footnotesize (1.03$\times$)} & 2.01M {\footnotesize (149$\times$)} & 69K {\footnotesize (5.11$\times$)} & 177K {\footnotesize (13.1$\times$)} & 14.8K {\footnotesize (1.10$\times$)} & 17.5K {\footnotesize (1.30$\times$)} & 20M {\footnotesize (1484$\times$)} & 294K {\footnotesize (21.8$\times$)} \\
         & Lognormal & 13.5K & 34.9K {\footnotesize (2.59$\times$)} & 13.6K {\footnotesize (1.01$\times$)} & 13.9K {\footnotesize (1.03$\times$)} & 2.01M {\footnotesize (149$\times$)} & 69K {\footnotesize (5.11$\times$)} & 177K {\footnotesize (13.1$\times$)} & 14.8K {\footnotesize (1.10$\times$)} & 17.6K {\footnotesize (1.30$\times$)} & 20M {\footnotesize (1483$\times$)} & 294K {\footnotesize (21.8$\times$)} \\
        \midrule
        128 & Amzn & 268K & 335K {\footnotesize (1.25$\times$)} & 268K {\footnotesize (1.00$\times$)} & 273K {\footnotesize (1.02$\times$)} & 8.27M {\footnotesize (30.9$\times$)} & 619K {\footnotesize (2.31$\times$)} & 670K {\footnotesize (2.50$\times$)} & 269K {\footnotesize (1.01$\times$)} & 315K {\footnotesize (1.17$\times$)} & 80.3M {\footnotesize (300$\times$)} & 1.23M {\footnotesize (4.60$\times$)} \\
         & Osmc & 668K & 723K {\footnotesize (1.08$\times$)} & 669K {\footnotesize (1.00$\times$)} & 674K {\footnotesize (1.01$\times$)} & 8.67M {\footnotesize (13.0$\times$)} & 1.23M {\footnotesize (1.85$\times$)} & 1M {\footnotesize (1.51$\times$)} & 671K {\footnotesize (1.01$\times$)} & 738K {\footnotesize (1.10$\times$)} & 80.7M {\footnotesize (121$\times$)} & 1.93M {\footnotesize (2.88$\times$)} \\
         & Face & 256K & 283K {\footnotesize (1.10$\times$)} & 256K {\footnotesize (1.00$\times$)} & 259K {\footnotesize (1.01$\times$)} & 2.26M {\footnotesize (8.80$\times$)} & 444K {\footnotesize (1.73$\times$)} & 374K {\footnotesize (1.46$\times$)} & 257K {\footnotesize (1.00$\times$)} & 282K {\footnotesize (1.10$\times$)} & 20.3M {\footnotesize (79.0$\times$)} & 668K {\footnotesize (2.60$\times$)} \\
         & YCSB & 663 & 10.3K {\footnotesize (15.5$\times$)} & 759 {\footnotesize (1.14$\times$)} & 786 {\footnotesize (1.19$\times$)} & 2M {\footnotesize (3018$\times$)} & 17.6K {\footnotesize (26.6$\times$)} & 79.8K {\footnotesize (120$\times$)} & 1.32K {\footnotesize (1.99$\times$)} & 1.6K {\footnotesize (2.42$\times$)} & 20M {\footnotesize (30167$\times$)} & 130K {\footnotesize (196$\times$)} \\
         & Longitudes & 13.7K & 22.3K {\footnotesize (1.63$\times$)} & 13.8K {\footnotesize (1.00$\times$)} & 13.9K {\footnotesize (1.01$\times$)} & 2.01M {\footnotesize (147$\times$)} & 41.5K {\footnotesize (3.02$\times$)} & 92.4K {\footnotesize (6.73$\times$)} & 13.9K {\footnotesize (1.01$\times$)} & 15.3K {\footnotesize (1.11$\times$)} & 20M {\footnotesize (1458$\times$)} & 154K {\footnotesize (11.2$\times$)} \\
         & Longlat & 58.1K & 67.7K {\footnotesize (1.16$\times$)} & 58.9K {\footnotesize (1.01$\times$)} & 58.7K {\footnotesize (1.01$\times$)} & 2.06M {\footnotesize (35.4$\times$)} & 122K {\footnotesize (2.11$\times$)} & 138K {\footnotesize (2.37$\times$)} & 59.1K {\footnotesize (1.02$\times$)} & 63.7K {\footnotesize (1.10$\times$)} & 20.1M {\footnotesize (345$\times$)} & 245K {\footnotesize (4.22$\times$)} \\
         & Uniform & 3.42K & 12.9K {\footnotesize (3.78$\times$)} & 3.53K {\footnotesize (1.03$\times$)} & 3.53K {\footnotesize (1.03$\times$)} & 2M {\footnotesize (586$\times$)} & 23.7K {\footnotesize (6.91$\times$)} & 83.7K {\footnotesize (24.5$\times$)} & 4.45K {\footnotesize (1.30$\times$)} & 4.44K {\footnotesize (1.30$\times$)} & 20M {\footnotesize (5847$\times$)} & 136K {\footnotesize (39.8$\times$)} \\
         & Normal & 3.5K & 13K {\footnotesize (3.70$\times$)} & 3.52K {\footnotesize (1.01$\times$)} & 3.6K {\footnotesize (1.03$\times$)} & 2M {\footnotesize (572$\times$)} & 23.7K {\footnotesize (6.78$\times$)} & 83.5K {\footnotesize (23.8$\times$)} & 3.77K {\footnotesize (1.08$\times$)} & 4.49K {\footnotesize (1.28$\times$)} & 20M {\footnotesize (5707$\times$)} & 136K {\footnotesize (38.9$\times$)} \\
         & Lognormal & 3.49K & 12.9K {\footnotesize (3.71$\times$)} & 3.52K {\footnotesize (1.01$\times$)} & 3.59K {\footnotesize (1.03$\times$)} & 2M {\footnotesize (574$\times$)} & 23.7K {\footnotesize (6.81$\times$)} & 83.4K {\footnotesize (23.9$\times$)} & 3.81K {\footnotesize (1.09$\times$)} & 4.51K {\footnotesize (1.29$\times$)} & 20M {\footnotesize (5735$\times$)} & 136K {\footnotesize (39.1$\times$)} \\
        \bottomrule
    \end{tabular}
\end{table*}

\Cref{tab:poisoning_impact_on_m_opt_all} shows the results.
Across all rows of \cref{tab:poisoning_impact_on_m_opt_all}, \textsc{PGM-attack} attains a higher post-poisoning $m_{\mathrm{opt}}$ than both \textsc{Random} and \textsc{Random-Adjacent}.
On the YCSB dataset, \textsc{PGM-attack} increases $m_{\mathrm{opt}}$ by $15.5\times$ with $1\%$ poisoning and by up to $120\times$ with $10\%$ poisoning.
This pronounced effect stems from the fact that YCSB is nearly uniform, allowing a PLA with very few segments to be constructed before poisoning.
By adding only a small number of poison keys, our attack drastically increases the required number of segments.
Random poisoning also has a relatively large impact on YCSB: at a $10\%$ poisoning rate, \textsc{Random} and \textsc{Random-Adjacent} increase $m_{\mathrm{opt}}$ by $1.99\times$ and $2.42\times$, respectively.
Nevertheless, the increase caused by \textsc{PGM-attack} is substantially larger.
The impact of our attack is also significant on datasets other than YCSB.
At a $10\%$ poisoning rate, it increases $m_{\mathrm{opt}}$ by factors ranging from $1.46\times$ to $24.5\times$, substantially exceeding the increases caused by the random baselines, which range from $1.00\times$ to $1.30\times$.

Regarding the upper bounds, Instance UB is consistently tighter than Agnostic UB and is at most $1.92\times$ the attained $m_{\mathrm{opt}}$.
This provides an instance-wise certificate that \textsc{PGM-attack} achieves at least $52\%$ of the optimum on all evaluated instances.
It also shows that incorporating instance-specific information yields a tighter bound than the closed-form Agnostic UB.
These results support the usefulness of our upper bounds for understanding the maximum achievable impact of poisoning attacks and, more broadly, for characterizing how much $m_{\mathrm{opt}}$ can increase under arbitrary insertions.

\subsubsection{Robustness to Unknown Index Parameters}
\label{app:epsilon_mismatch}

\begin{table}[t]
    \centering
    \vspace{-0.25em}
    \caption{$m_{\mathrm{opt}}$ after \textsc{PGM-attack} when the PGM-index is constructed with $\varepsilon_{\mathrm{true}}$, but poisons are generated assuming that the PGM-index parameter is $\varepsilon_{\mathrm{attack}}$. $\lambda = 0.1n$.}
    \label{tab:gray-box-results}
    \vspace{-1.0em}
    \footnotesize
    \setlength{\tabcolsep}{3.5pt}
    \renewcommand{\arraystretch}{0.88}
    \begin{tabular}{@{}l *{4}{r}@{}}
        \toprule
        Amzn & $\varepsilon_{\mathrm{true}} = 16$ & $\varepsilon_{\mathrm{true}} = 32$ & $\varepsilon_{\mathrm{true}} = 64$ & $\varepsilon_{\mathrm{true}} = 128$ \\
        \cmidrule(lr){2-5}
        $\varepsilon_{\mathrm{attack}} = 16$ & \textbf{12.1M {\footnotesize (1.52$\times$)}} & 3.19M {\footnotesize (1.30$\times$)} & 1.04M {\footnotesize (1.31$\times$)} & 386K {\footnotesize (1.44$\times$)} \\
        $\varepsilon_{\mathrm{attack}} = 32$ & 10.2M {\footnotesize (1.28$\times$)} & \textbf{4.79M {\footnotesize (1.94$\times$)}} & 1.05M {\footnotesize (1.32$\times$)} & 365K {\footnotesize (1.36$\times$)} \\
        $\varepsilon_{\mathrm{attack}} = 64$ & 9.2M {\footnotesize (1.16$\times$)} & 3.55M {\footnotesize (1.44$\times$)} & \textbf{1.73M {\footnotesize (2.17$\times$)}} & 395K {\footnotesize (1.47$\times$)} \\
        $\varepsilon_{\mathrm{attack}} = 128$ & 8.52M {\footnotesize (1.07$\times$)} & 3.03M {\footnotesize (1.23$\times$)} & 1.33M {\footnotesize (1.67$\times$)} & \textbf{670K {\footnotesize (2.50$\times$)}} \\
        \midrule
        Osmc & $\varepsilon_{\mathrm{true}} = 16$ & $\varepsilon_{\mathrm{true}} = 32$ & $\varepsilon_{\mathrm{true}} = 64$ & $\varepsilon_{\mathrm{true}} = 128$ \\
        \cmidrule(lr){2-5}
        $\varepsilon_{\mathrm{attack}} = 16$ & \textbf{9.18M {\footnotesize (1.49$\times$)}} & 3.47M {\footnotesize (1.20$\times$)} & 1.6M {\footnotesize (1.17$\times$)} & 780K {\footnotesize (1.17$\times$)} \\
        $\varepsilon_{\mathrm{attack}} = 32$ & 7.48M {\footnotesize (1.21$\times$)} & \textbf{4.3M {\footnotesize (1.49$\times$)}} & 1.66M {\footnotesize (1.21$\times$)} & 773K {\footnotesize (1.16$\times$)} \\
        $\varepsilon_{\mathrm{attack}} = 64$ & 6.84M {\footnotesize (1.11$\times$)} & 3.53M {\footnotesize (1.22$\times$)} & \textbf{2.06M {\footnotesize (1.50$\times$)}} & 809K {\footnotesize (1.21$\times$)} \\
        $\varepsilon_{\mathrm{attack}} = 128$ & 6.53M {\footnotesize (1.06$\times$)} & 3.22M {\footnotesize (1.11$\times$)} & 1.69M {\footnotesize (1.23$\times$)} & \textbf{1M {\footnotesize (1.51$\times$)}} \\
        \midrule
        Face & $\varepsilon_{\mathrm{true}} = 16$ & $\varepsilon_{\mathrm{true}} = 32$ & $\varepsilon_{\mathrm{true}} = 64$ & $\varepsilon_{\mathrm{true}} = 128$ \\
        \cmidrule(lr){2-5}
        $\varepsilon_{\mathrm{attack}} = 16$ & \textbf{3.08M {\footnotesize (1.45$\times$)}} & 1.19M {\footnotesize (1.13$\times$)} & 581K {\footnotesize (1.11$\times$)} & 286K {\footnotesize (1.11$\times$)} \\
        $\varepsilon_{\mathrm{attack}} = 32$ & 2.59M {\footnotesize (1.22$\times$)} & \textbf{1.53M {\footnotesize (1.45$\times$)}} & 594K {\footnotesize (1.14$\times$)} & 286K {\footnotesize (1.12$\times$)} \\
        $\varepsilon_{\mathrm{attack}} = 64$ & 2.4M {\footnotesize (1.13$\times$)} & 1.3M {\footnotesize (1.23$\times$)} & \textbf{760K {\footnotesize (1.45$\times$)}} & 293K {\footnotesize (1.15$\times$)} \\
        $\varepsilon_{\mathrm{attack}} = 128$ & 2.29M {\footnotesize (1.08$\times$)} & 1.2M {\footnotesize (1.13$\times$)} & 646K {\footnotesize (1.23$\times$)} & \textbf{374K {\footnotesize (1.46$\times$)}} \\
        \midrule
        YCSB & $\varepsilon_{\mathrm{true}} = 16$ & $\varepsilon_{\mathrm{true}} = 32$ & $\varepsilon_{\mathrm{true}} = 64$ & $\varepsilon_{\mathrm{true}} = 128$ \\
        \cmidrule(lr){2-5}
        $\varepsilon_{\mathrm{attack}} = 16$ & \textbf{810K {\footnotesize (11.6$\times$)}} & 144K {\footnotesize (5.74$\times$)} & 47.6K {\footnotesize (6.85$\times$)} & 15.1K {\footnotesize (22.7$\times$)} \\
        $\varepsilon_{\mathrm{attack}} = 32$ & 622K {\footnotesize (8.87$\times$)} & \textbf{361K {\footnotesize (14.4$\times$)}} & 64.7K {\footnotesize (9.30$\times$)} & 24.1K {\footnotesize (36.3$\times$)} \\
        $\varepsilon_{\mathrm{attack}} = 64$ & 328K {\footnotesize (4.67$\times$)} & 264K {\footnotesize (10.5$\times$)} & \textbf{169K {\footnotesize (24.2$\times$)}} & 35.6K {\footnotesize (53.7$\times$)} \\
        $\varepsilon_{\mathrm{attack}} = 128$ & 206K {\footnotesize (2.93$\times$)} & 162K {\footnotesize (6.44$\times$)} & 103K {\footnotesize (14.8$\times$)} & \textbf{79.8K {\footnotesize (120$\times$)}} \\
        \midrule
        Longitudes & $\varepsilon_{\mathrm{true}} = 16$ & $\varepsilon_{\mathrm{true}} = 32$ & $\varepsilon_{\mathrm{true}} = 64$ & $\varepsilon_{\mathrm{true}} = 128$ \\
        \cmidrule(lr){2-5}
        $\varepsilon_{\mathrm{attack}} = 16$ & \textbf{831K {\footnotesize (5.05$\times$)}} & 174K {\footnotesize (2.77$\times$)} & 66.7K {\footnotesize (2.40$\times$)} & 33.1K {\footnotesize (2.41$\times$)} \\
        $\varepsilon_{\mathrm{attack}} = 32$ & 490K {\footnotesize (2.98$\times$)} & \textbf{380K {\footnotesize (6.07$\times$)}} & 79.7K {\footnotesize (2.86$\times$)} & 29.5K {\footnotesize (2.15$\times$)} \\
        $\varepsilon_{\mathrm{attack}} = 64$ & 324K {\footnotesize (1.97$\times$)} & 220K {\footnotesize (3.51$\times$)} & \textbf{185K {\footnotesize (6.66$\times$)}} & 39.6K {\footnotesize (2.89$\times$)} \\
        $\varepsilon_{\mathrm{attack}} = 128$ & 245K {\footnotesize (1.49$\times$)} & 142K {\footnotesize (2.26$\times$)} & 106K {\footnotesize (3.82$\times$)} & \textbf{92.4K {\footnotesize (6.73$\times$)}} \\
        \midrule
        Longlat & $\varepsilon_{\mathrm{true}} = 16$ & $\varepsilon_{\mathrm{true}} = 32$ & $\varepsilon_{\mathrm{true}} = 64$ & $\varepsilon_{\mathrm{true}} = 128$ \\
        \cmidrule(lr){2-5}
        $\varepsilon_{\mathrm{attack}} = 16$ & \textbf{1.1M {\footnotesize (2.44$\times$)}} & 323K {\footnotesize (1.47$\times$)} & 148K {\footnotesize (1.32$\times$)} & 76.1K {\footnotesize (1.31$\times$)} \\
        $\varepsilon_{\mathrm{attack}} = 32$ & 765K {\footnotesize (1.69$\times$)} & \textbf{537K {\footnotesize (2.45$\times$)}} & 164K {\footnotesize (1.46$\times$)} & 74.1K {\footnotesize (1.27$\times$)} \\
        $\varepsilon_{\mathrm{attack}} = 64$ & 611K {\footnotesize (1.35$\times$)} & 375K {\footnotesize (1.71$\times$)} & \textbf{271K {\footnotesize (2.42$\times$)}} & 84.2K {\footnotesize (1.45$\times$)} \\
        $\varepsilon_{\mathrm{attack}} = 128$ & 534K {\footnotesize (1.18$\times$)} & 298K {\footnotesize (1.36$\times$)} & 190K {\footnotesize (1.70$\times$)} & \textbf{138K {\footnotesize (2.37$\times$)}} \\
        \midrule
        Uniform & $\varepsilon_{\mathrm{true}} = 16$ & $\varepsilon_{\mathrm{true}} = 32$ & $\varepsilon_{\mathrm{true}} = 64$ & $\varepsilon_{\mathrm{true}} = 128$ \\
        \cmidrule(lr){2-5}
        $\varepsilon_{\mathrm{attack}} = 16$ & \textbf{930K {\footnotesize (4.57$\times$)}} & 193K {\footnotesize (3.63$\times$)} & 69K {\footnotesize (5.12$\times$)} & 28.1K {\footnotesize (8.21$\times$)} \\
        $\varepsilon_{\mathrm{attack}} = 32$ & 588K {\footnotesize (2.89$\times$)} & \textbf{391K {\footnotesize (7.38$\times$)}} & 79.8K {\footnotesize (5.93$\times$)} & 27.3K {\footnotesize (7.99$\times$)} \\
        $\varepsilon_{\mathrm{attack}} = 64$ & 403K {\footnotesize (1.98$\times$)} & 243K {\footnotesize (4.59$\times$)} & \textbf{177K {\footnotesize (13.2$\times$)}} & 36.7K {\footnotesize (10.7$\times$)} \\
        $\varepsilon_{\mathrm{attack}} = 128$ & 303K {\footnotesize (1.49$\times$)} & 154K {\footnotesize (2.91$\times$)} & 111K {\footnotesize (8.23$\times$)} & \textbf{83.7K {\footnotesize (24.5$\times$)}} \\
        \midrule
        Normal & $\varepsilon_{\mathrm{true}} = 16$ & $\varepsilon_{\mathrm{true}} = 32$ & $\varepsilon_{\mathrm{true}} = 64$ & $\varepsilon_{\mathrm{true}} = 128$ \\
        \cmidrule(lr){2-5}
        $\varepsilon_{\mathrm{attack}} = 16$ & \textbf{929K {\footnotesize (4.58$\times$)}} & 192K {\footnotesize (3.64$\times$)} & 69K {\footnotesize (5.12$\times$)} & 28.3K {\footnotesize (8.06$\times$)} \\
        $\varepsilon_{\mathrm{attack}} = 32$ & 587K {\footnotesize (2.89$\times$)} & \textbf{391K {\footnotesize (7.39$\times$)}} & 79.7K {\footnotesize (5.91$\times$)} & 27.2K {\footnotesize (7.76$\times$)} \\
        $\varepsilon_{\mathrm{attack}} = 64$ & 403K {\footnotesize (1.98$\times$)} & 243K {\footnotesize (4.60$\times$)} & \textbf{177K {\footnotesize (13.1$\times$)}} & 36.7K {\footnotesize (10.5$\times$)} \\
        $\varepsilon_{\mathrm{attack}} = 128$ & 302K {\footnotesize (1.49$\times$)} & 153K {\footnotesize (2.90$\times$)} & 110K {\footnotesize (8.18$\times$)} & \textbf{83.5K {\footnotesize (23.8$\times$)}} \\
        \midrule
        Lognormal & $\varepsilon_{\mathrm{true}} = 16$ & $\varepsilon_{\mathrm{true}} = 32$ & $\varepsilon_{\mathrm{true}} = 64$ & $\varepsilon_{\mathrm{true}} = 128$ \\
        \cmidrule(lr){2-5}
        $\varepsilon_{\mathrm{attack}} = 16$ & \textbf{929K {\footnotesize (4.58$\times$)}} & 192K {\footnotesize (3.63$\times$)} & 69K {\footnotesize (5.12$\times$)} & 28.3K {\footnotesize (8.10$\times$)} \\
        $\varepsilon_{\mathrm{attack}} = 32$ & 587K {\footnotesize (2.89$\times$)} & \textbf{391K {\footnotesize (7.39$\times$)}} & 79.7K {\footnotesize (5.91$\times$)} & 27.2K {\footnotesize (7.81$\times$)} \\
        $\varepsilon_{\mathrm{attack}} = 64$ & 403K {\footnotesize (1.98$\times$)} & 243K {\footnotesize (4.60$\times$)} & \textbf{177K {\footnotesize (13.1$\times$)}} & 36.7K {\footnotesize (10.5$\times$)} \\
        $\varepsilon_{\mathrm{attack}} = 128$ & 302K {\footnotesize (1.49$\times$)} & 153K {\footnotesize (2.90$\times$)} & 110K {\footnotesize (8.16$\times$)} & \textbf{83.4K {\footnotesize (23.9$\times$)}} \\
        \bottomrule
    \end{tabular}
\end{table}

We further evaluate whether \textsc{PGM-attack} remains effective when the attacker does not know the true PGM-index parameter.
Specifically, we generate poison keys using an assumed value $\varepsilon_{\mathrm{attack}}$ and construct the index using the actual value $\varepsilon_{\mathrm{true}}$.
\Cref{tab:gray-box-results} reports the resulting $m_{\mathrm{opt}}$.

As expected, the attack is generally most effective when $\varepsilon_{\mathrm{attack}}=\varepsilon_{\mathrm{true}}$.
Nevertheless, \textsc{PGM-attack} remains effective even under parameter mismatch.
For example, on Amzn with $\varepsilon_{\mathrm{true}}=128$, the attack increases $m_{\mathrm{opt}}$ by $2.50\times$ when the attacker knows the true parameter and sets $\varepsilon_{\mathrm{attack}}=128$.
When the attacker instead assumes $\varepsilon_{\mathrm{attack}}=32$, the increase decreases to $1.36\times$.
Even so, this remains larger than the increases caused by random and random-adjacent poisoning, which are $1.01\times$ and $1.17\times$, respectively (\cref{tab:poisoning_impact_on_m_opt_all}).
Similarly, on YCSB with $\varepsilon_{\mathrm{true}}=128$, mismatched attacks still increase $m_{\mathrm{opt}}$ by at least $22.7\times$, far exceeding the $1.99\times$ and $2.42\times$ increases caused by random and random-adjacent poisoning.

These results suggest that poison keys generated to target a PLA with one value of $\varepsilon$ partially transfer to PLAs constructed with other values.
Thus, exact knowledge of $\varepsilon$ is not necessary to generate effective poison keys.
In this sense, our white-box attack already remains effective in a \textit{gray-box} setting where the attacker knows the data but not the value of $\varepsilon$.
Our method therefore provides a useful foundation for developing more sophisticated attacks under gray-box and black-box settings.

\subsubsection{Ablation Study on $\mu$}
\label{app:ablation_study_mu}

\begin{figure}[t]
    \centering
    \includegraphics[width=\columnwidth]{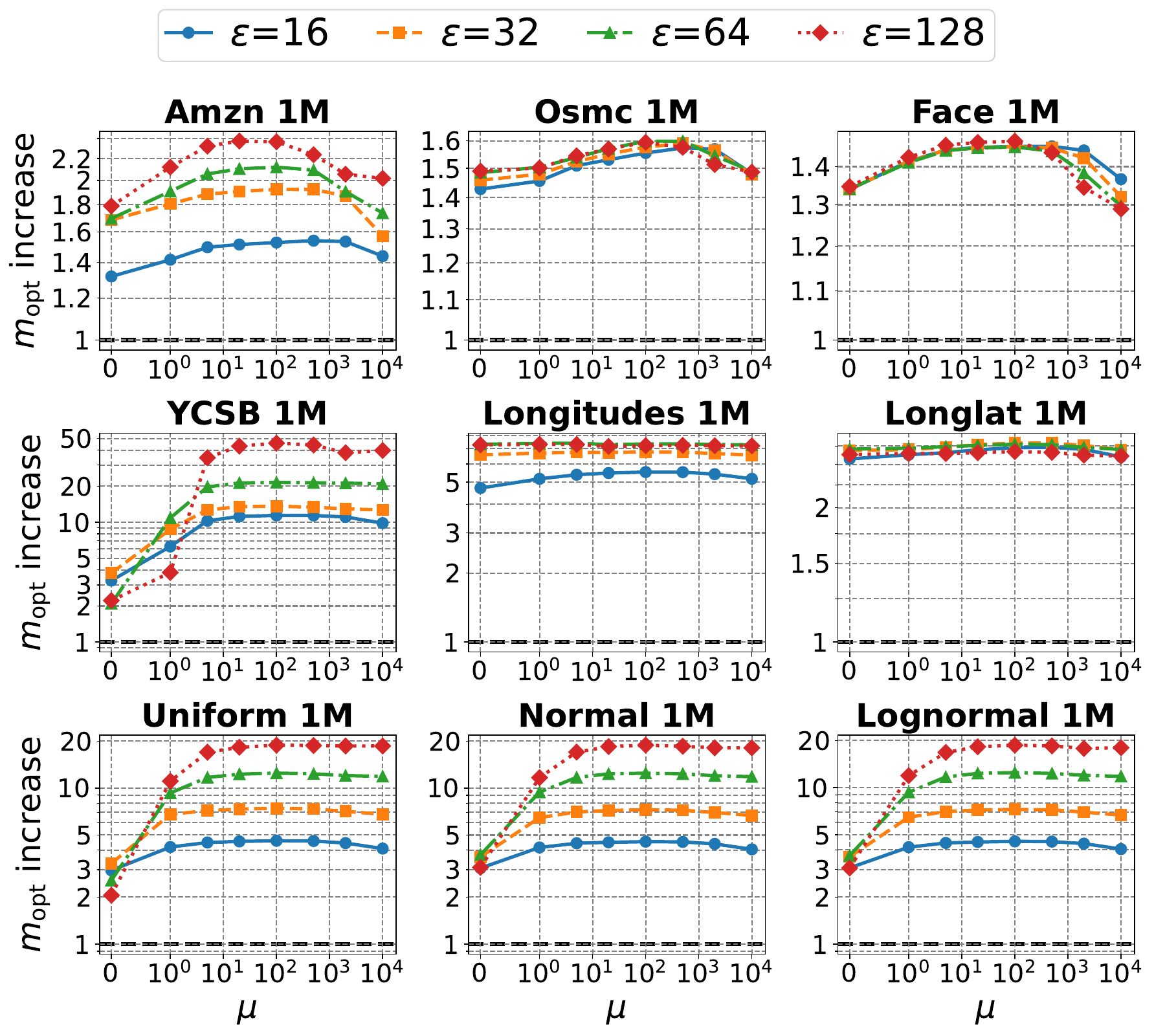}
    \caption{Effect of $\mu$ on $m_{\mathrm{opt}}$ for $\lambda = 100{,}000$.}
    \label{fig:mu_to_mopt_lambda100000_epsilon_overlay}
\end{figure}

To study the effect of $\mu$, we report in \cref{fig:mu_to_mopt_lambda100000_epsilon_overlay} the results averaged over five seeds on $n = 1$M keys. For $\varepsilon \in \{16,32,64,128\}$, we plot the factor by which $m_{\mathrm{opt}}$ increases relative to legitimate keys.
Very small or very large $\mu$ can reduce the amplification: small $\mu$ under-penalizes poison usage and exhausts budget unevenly, while large $\mu$ overreacts to budget imbalance and behaves unstably.
Empirically, $\mu \approx 100$ works robustly; dynamic $\theta$ improves over fixing $\theta = \theta_0$ ($\mu = 0$) by up to $20\times$.
These results show that dynamically updating $\theta$ yields higher measured $m_{\mathrm{opt}}$ than fixing $\theta = \theta_0$.

\subsubsection{Poison Generation Time}
\label{app:poison_generation_time}

\begin{table*}[t]
    \centering
    \caption{Poison generation time (hours).}
    \label{tab:generation-time-full}
    \begin{tabular}{@{}l r r r r r r r r@{}}
        \toprule
         & \multicolumn{4}{c}{$\lambda = 0.01\,n$} & \multicolumn{4}{c}{$\lambda = 0.1\,n$} \\
        \cmidrule(lr){2-5} \cmidrule(lr){6-9}
        Dataset & $\varepsilon = 16$ & $\varepsilon = 32$ & $\varepsilon = 64$ & $\varepsilon = 128$ & $\varepsilon = 16$ & $\varepsilon = 32$ & $\varepsilon = 64$ & $\varepsilon = 128$ \\
        \midrule
        Amzn & 3.1 & 2.9 & 4.3 & 4.7 & 4.9 & 5.8 & 8.7 & 12.3 \\
        Osmc & 2.9 & 2.5 & 2.6 & 2.3 & 4.9 & 4.8 & 4.8 & 5.0 \\
        Face & 0.9 & 0.8 & 0.8 & 0.6 & 1.3 & 1.2 & 1.0 & 1.0 \\
        YCSB & 1.5 & 1.2 & 2.0 & 3.3 & 3.5 & 4.2 & 6.0 & 14.7 \\
        Longitudes & 1.0 & 0.9 & 1.1 & 1.1 & 3.0 & 4.3 & 5.0 & 5.3 \\
        Longlat & 0.7 & 0.7 & 0.7 & 0.6 & 2.2 & 2.7 & 2.7 & 2.6 \\
        Uniform & 1.0 & 1.2 & 1.6 & 1.9 & 2.3 & 3.7 & 5.3 & 7.3 \\
        Normal & 1.0 & 1.2 & 1.6 & 1.8 & 2.3 & 3.9 & 5.5 & 7.1 \\
        Lognormal & 1.0 & 1.2 & 1.6 & 1.8 & 2.3 & 3.6 & 5.3 & 7.2 \\
        \bottomrule
    \end{tabular}
\end{table*}

\Cref{tab:generation-time-full} reports the poison generation time for all combinations of $\varepsilon$ and $\lambda$.
Across all datasets and parameter settings, poison generation takes between $0.6$ and $14.7$ hours.
Although this computational cost is non-negligible, it remains practical for an offline attacker because the poison keys can be generated in advance and subsequently inserted into the target index.

The generation time tends to increase as either $\varepsilon$ or $\lambda$ increases.
The main reason is that these settings allow the selected poison keys to shorten the current segment more substantially.
\textsc{PGM-attack} processes the key sequence from left to right.
At each step, it generates poison candidates for the segment starting at the current segment-start position, selects poison keys, and recomputes the resulting segmentation.
If the selected poisons shorten the segment only slightly, the next segment starts relatively far to the right, and the algorithm soon moves to a mostly new region of the dataset.
In contrast, when the shortening is large, the next segment starts much earlier than before poisoning.
Consequently, the algorithm must again generate and evaluate poison candidates for a segment that substantially overlaps a previously processed region.
This repeated recomputation over overlapping regions increases the total generation time.
Larger $\varepsilon$ permits greater reductions in segment length, while a larger poison budget $\lambda$ increases the cumulative shortening caused by the selected poisons; both lead to more repeated processing.

A secondary reason for the increase with $\varepsilon$ is that the segments before poisoning tend to contain more keys.
For a segment containing $c$ legitimate keys, DI-Consecutive requires $\mathcal{O}(c\log c)$ time to generate a poison candidate.
Thus, as $\varepsilon$ increases and the unpoisoned segment length $c$ becomes larger, the cost of processing each candidate segment also grows superlinearly in $c$.

\subsection{Additional Results on Index Performance}
\label{app:index_performance_results}

\begin{figure*}[p]
    \centering
    \includegraphics[width=\textwidth]{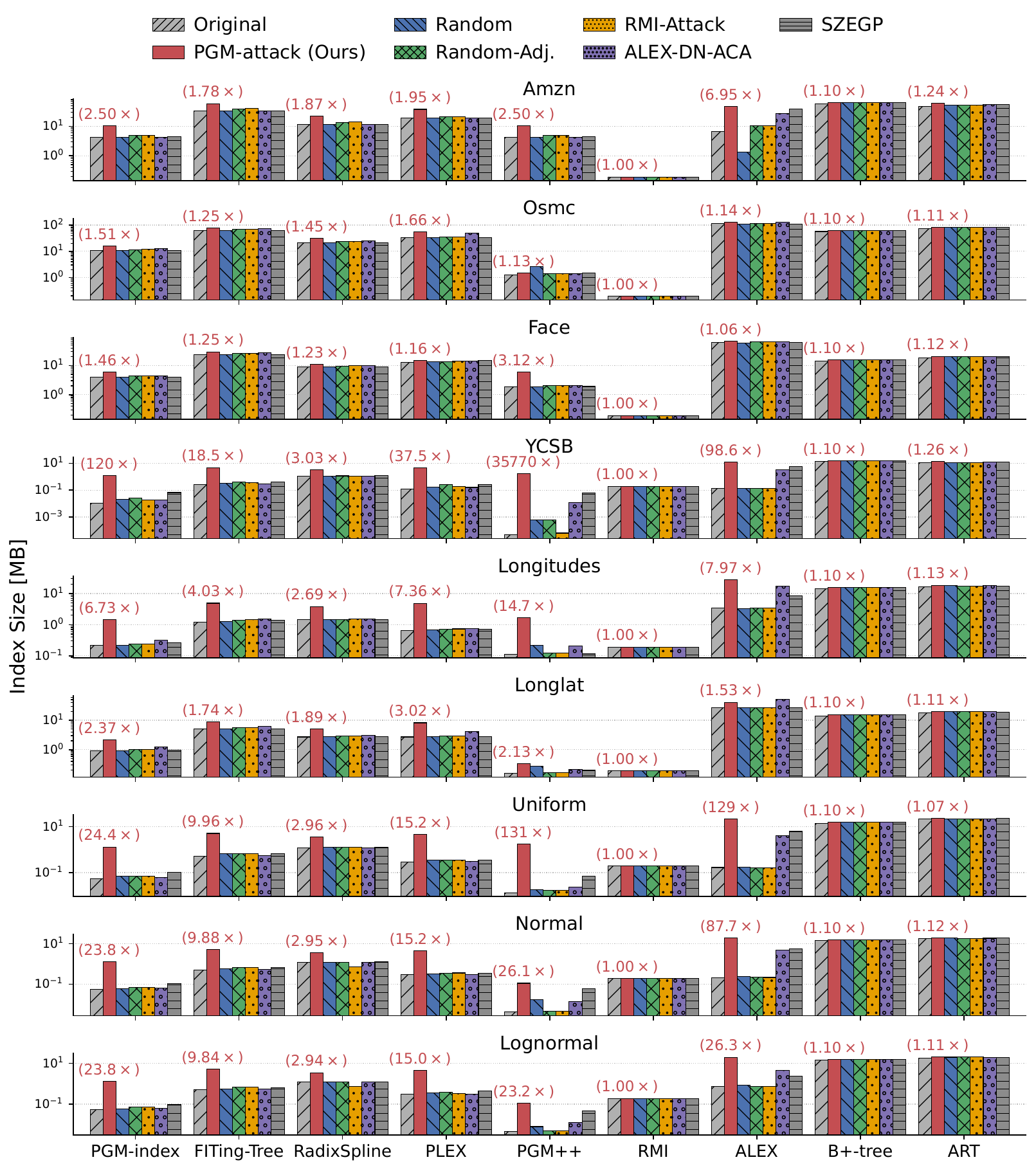}
    \caption{Index sizes before and after 10\% poisoning ($\lambda = 0.1n$) for $\varepsilon = 128$. Numbers above the \textsc{PGM-attack} bars indicate ratios relative to the corresponding original indexes.}
    \label{fig:index_size_eps128_lambda01n}
\end{figure*}

\begin{figure*}[p]
    \centering
    \includegraphics[width=\textwidth]{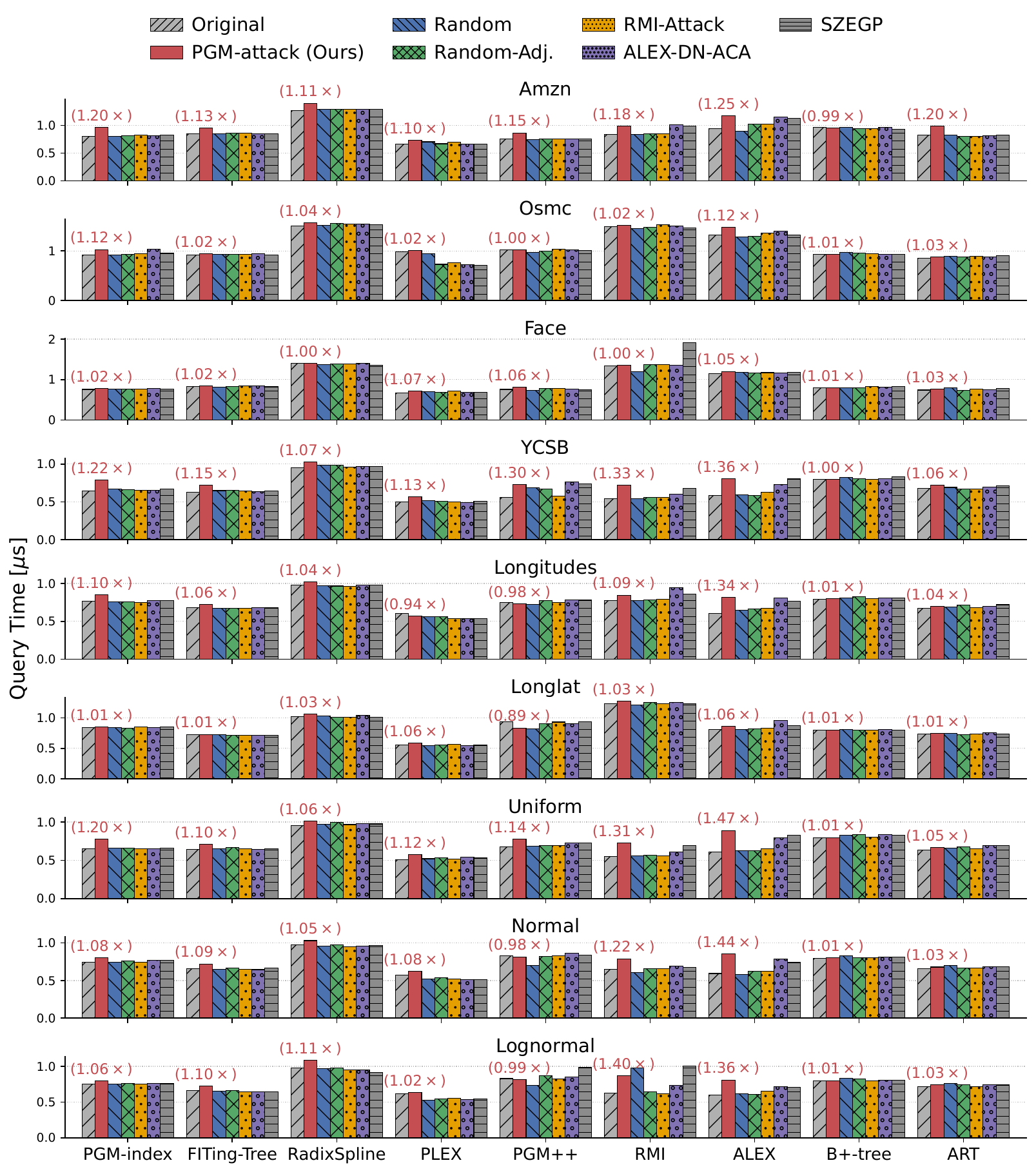}
    \caption{Query times before and after 10\% poisoning ($\lambda = 0.1n$) for $\varepsilon = 128$. Numbers above the \textsc{PGM-attack} bars indicate ratios relative to the corresponding original indexes.}
    \label{fig:query_time_eps128_lambda01n}
\end{figure*}

We further evaluate the practical impact of \textsc{PGM-attack} on index size and query time.
Because this extended evaluation covers more datasets, indexes, and poisoning methods than the evaluation in \cref{sec:experiment_poisoning_index_performance}, we present the results as figures rather than tables for readability.
\Cref{fig:index_size_eps128_lambda01n,fig:query_time_eps128_lambda01n} report the index size and query time, respectively, for $\varepsilon = 128$ and 10\% poisoning ($\lambda = 0.1n$).
For each \textsc{PGM-attack} bar, we additionally display the ratio relative to the corresponding original, unpoisoned index above the bar.

\paragraph{Poisoning Methods.}
In addition to our proposed \textsc{PGM-attack}, we evaluate two random baselines, \textsc{Random} and \textsc{Random-Adj.}
\textsc{Random} selects poison keys uniformly at random, without replacement, from the unoccupied integers within the legitimate key range.
Motivated by \cref{thm:maxerror_attack_structure}, \textsc{Random-Adj.} instead samples uniformly from unoccupied integers adjacent to legitimate keys.
We also evaluate the poisoning attack proposed for RMIs~\cite{kornaropoulos2022price}, which we refer to as \textsc{RMI-attack}.
Although \textsc{RMI-attack} was specifically designed to maximize the squared prediction error of RMI models, we apply the poison keys generated by every method to every evaluated index to examine their transferability across index architectures.

We further evaluate two methods proposed as algorithmic complexity attacks (ACAs) against ALEX.
The first is a space ACA targeting ALEX data nodes, which we denote by \textsc{ALEX-DN-ACA}~\cite{yang2023algorithmic}.
The second, denoted by \textsc{SZEGP}, is a space ACA that induces repeated structural expansion in ALEX through adversarial insertions~\cite{schuster2025learned}.
Both methods were originally designed to increase the memory consumption of an already constructed ALEX by dynamically inserting adversarially selected keys.
This differs from our poisoning setting, in which the attacker injects poison keys into the dataset before index construction.
To ensure a consistent comparison, we evaluate all methods under our poisoning setting: each index is bulk-loaded from the dataset containing the generated poison keys, rather than being attacked through online insertions after construction.

\paragraph{Effect on PLA-Based Learned Indexes.}
\textsc{PGM-attack} is not only effective against PGM-index, but also effective against all four other PLA-based learned indexes.
The largest multiplicative increases tend to occur when the original index is particularly small, that is, when the index benefits strongly from the regularity and linear approximability of the original key distribution.
In such cases, poisoning removes much of this advantage by creating many local regions that cannot be accurately represented using a small number of linear segments.

This effect is most pronounced on YCSB, although \textsc{PGM-attack} produces greater performance degradation than the other poisoning methods in nearly all cases on the remaining real-world and synthetic datasets.
On YCSB, relative to the original index, \textsc{PGM-attack} increases the index size by up to $18.5\times$ for FITing-Tree, $3.03\times$ for RadixSpline, $37.5\times$ for PLEX, and $35{,}770\times$ for PGM++.
The exceptionally large ratio for PGM++ partly reflects its extremely small original index size on YCSB.
Across all other real-world and synthetic datasets, \textsc{PGM-attack} increases the sizes of all PLA-based learned indexes by factors ranging from $1.13\times$ to $130\times$.
In nearly every case, it produces a larger increase than any competing poisoning method.

These results suggest that \textsc{PGM-attack} exploits the underlying PLA approximation problem rather than an implementation-specific property of the PGM-index.
Its poison keys therefore transfer to learned indexes that employ different PLA construction algorithms, recursive organizations, or auxiliary search structures.

The increase in query time is less pronounced than the increase in index size.
Nevertheless, in most cases, \textsc{PGM-attack} causes the largest query-time degradation among the evaluated poisoning methods, increasing query time by up to $1.30\times$.
This slowdown is likely attributable, at least in part, to the larger structures induced by the increased number of segments, which may lead to additional cache misses and traversal overhead.

\paragraph{Effect on Other Learned Indexes.}
\textsc{PGM-attack} can also substantially increase both the index size and query time of ALEX.
In particular, it increases the ALEX index size by $98.6\times$ on YCSB and by $129\times$ on Uniform.
These large increases can be explained by the fact that ALEX's bulk-loading algorithm adaptively determines its index structure based on model accuracy.
During construction, ALEX evaluates the model accuracy at each node and decides whether to retain the keys in a single data node or partition them among multiple child nodes.
Lower model accuracy can therefore lead to a more highly partitioned structure and a larger overall index.

Because ALEX uses linear models to navigate its nodes, the poison keys generated by \textsc{PGM-attack} likely also increase the prediction errors of these models.
These results demonstrate that the transferability of \textsc{PGM-attack} extends beyond the family of PLA-based learned indexes.

\paragraph{Comparison with \textsc{RMI-attack}.}
Notably, \textsc{PGM-attack} is at least as effective as \textsc{RMI-attack}, even on RMI, for which the latter was specifically designed.
On RMI, \textsc{PGM-attack} increases query time by up to $1.40\times$, whereas \textsc{RMI-attack} increases it by at most $1.04\times$.

One possible explanation is the difference between the objectives optimized by the two attacks.
\textsc{RMI-attack} selects poison keys to maximize the squared prediction error of the RMI models.
However, an RMI answers a query by performing a last-mile search, such as linear or exponential search, around the predicted position.
Squared error may therefore be imperfectly aligned with the actual query-processing cost: a small number of very large errors can dominate the objective without proportionally increasing the average last-mile search cost.
In contrast, \textsc{PGM-attack} targets local worst-case approximation error and places poison keys in regions where they substantially disrupt linear approximation.
This objective may be more closely aligned with the last-mile search cost incurred by learned indexes.

The optimization procedure used by \textsc{RMI-attack} may also limit its effectiveness.
In particular, it assumes that leaf models are responsible for equal numbers of keys and employs a hill-climbing procedure initialized from a uniform poison allocation.
Consequently, the resulting poison set may remain far from the optimum for a particular key distribution or RMI configuration.

\paragraph{Comparison with Attacks on ALEX.}
A similar pattern arises for ALEX.
Under our evaluation setting, \textsc{PGM-attack} causes greater degradation than the attacks proposed specifically for ALEX~\cite{yang2023algorithmic,schuster2025learned}.
This difference is largely attributable to the threat models for which these ACA methods were designed.
They exploit algorithmic behavior triggered by online insertions into an already constructed ALEX, such as inefficient resizing, splitting, or structural expansion, rather than constructing a poisoned dataset from which ALEX is subsequently bulk-loaded.

When the keys generated by these attacks are instead included during bulk loading, ALEX can jointly optimize its initial structure over both legitimate and adversarial keys, avoiding many of the transient inefficiencies that arise during online updates.
Bulk loading is therefore generally a more favorable setting for the index than adversarial online insertion.
The effectiveness of \textsc{PGM-attack} even in this setting suggests that it does not merely exploit an inefficient update procedure.
Rather, it reveals a broader vulnerability in the model-based construction objective itself: carefully placed keys can make the data intrinsically difficult to represent using the simple local models on which the index relies.

\paragraph{Transferability across Learned Indexes.}
Finally, \textsc{PGM-attack} exhibits substantially greater transferability than the prior index-specific attacks.
The effects of \textsc{RMI-attack}, \textsc{ALEX-DN-ACA}, and \textsc{SZEGP} are generally limited outside the architectures for which they were designed.
By contrast, \textsc{PGM-attack} consistently degrades a broad range of learned indexes.

We attribute this transferability to a design principle shared by many learned indexes: they partition the key space into regions and represent each region using a small, simple model, often a linear model.
\textsc{PGM-attack} concentrates poison keys at locations where they substantially increase local approximation difficulty and worst-case prediction error.
Because these quantities are closely related to the number of required models, the complexity of the index structure, and the cost of last-mile search, the resulting poison keys can remain effective across substantially different learned-index architectures.

\section{Duplicate-Allowed Poisoning}
\label{app:duplicate_allowed_poisoning}

In this section, we discuss poisoning in the duplicate-allowed setting.
We first describe how duplicates are handled in the actual PGM-index implementation.
We then discuss how to generate poisons under this setting.
Finally, we experimentally evaluate \textsc{PGM-attack} in this setting.

\subsection{PGM-index with Duplicates}

\begin{algorithm}[t]
    \caption{PLA construction in PGM-index with duplicate keys}
    \label{alg:pgm_duplicates_allowed}
    \begin{algorithmic}[1]
    \Require
    \Statex Nondecreasing keys array $\mathcal{X} = (x_1, x_2, \dots, x_n)$;
    \Statex Error parameter $\varepsilon$.
    \Ensure A PLA built for $\mathcal{X}$ with error at most $\varepsilon$.
    \State $\mathrm{PLA} \gets \mathrm{initializePLA}(\varepsilon)$
    \State $\mathrm{PLA}.\mathrm{addPoint}(x_1, 1)$
    \For{$i \gets 2$ \textbf{to} $n$}
        \If{$x_i = x_{i-1}$}
            \If{$i < n ~\land~ x_i + 1 < x_{i+1}$}
                \State $\mathrm{PLA}.\mathrm{addPoint}(x_i+1, i)$
            \EndIf
        \Else
            \State $\mathrm{PLA}.\mathrm{addPoint}(x_i, i)$
        \EndIf
    \EndFor
    \State \Return $\mathrm{PLA}$
    \end{algorithmic}
\end{algorithm}

\begin{figure}[t]
    \centering
    \includegraphics[width=0.9\columnwidth]{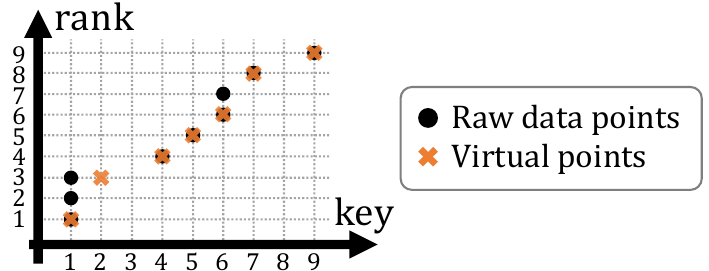}
    \caption{PLA construction with duplicate keys in the PGM-index. Instead of directly constructing the PLA from the original key-rank pairs, the PGM-index constructs the PLA over a set of virtual points introduced to handle duplicates.}
    \label{fig:dup_allowed_pgm}
\end{figure}

When given a sorted integer array containing duplicate keys, the PGM-index does not construct the PLA directly from the original key-rank pairs.
Instead, it constructs the PLA over virtual points introduced to handle duplicates.
Specifically, it first adds $(x_1,1)$.
Then, for each position $i \in \{2,3,\dots,n\}$, if $x_i$ differs from the previous key, it adds $(x_i,i)$.
If $x_i$ is a duplicate, it does not add that point directly; instead, only when $x_i$ is at the end of a duplicate run and there is a gap before the next distinct key, it adds $(x_i+1,i)$.
\Cref{alg:pgm_duplicates_allowed} presents the pseudocode, and \cref{fig:dup_allowed_pgm} illustrates the resulting virtual points.
The PLA is then constructed with error at most $\varepsilon$ on these virtual points.

This keeps the construction feasible even when a key occurs more than $2\varepsilon+1$ times;
if the raw data points are used directly, PLA construction becomes infeasible in this case (see \cref{fig:dup_allowed_pgm}).
At the same time, the resulting PLA can still answer lower-bound and related queries correctly.
In particular, $(x_i+1,i)$ captures the information immediately after the duplicate run of value $x_i$.
As a result, even with duplicate keys, the PGM-index can properly construct the approximate position and correction range used internally, allowing its final search procedure to return the correct result for lower-bound-style queries.

\subsection{Poison Generation with Duplicates}

\begin{table*}[t]
    \centering
    \caption{Results for $m_{\mathrm{opt}}$ after duplicate-allowed poisoning.
    \textsc{PGM-attack} increases $m_{\mathrm{opt}}$ relative to legitimate keys (by up to $4.68\times$ at $\lambda = 0.01\,n$ and up to $36.1\times$ at $\lambda = 0.1\,n$), whereas \textsc{Random} and \textsc{Random-Adj.} barely change $m_{\mathrm{opt}}$ (by at most $1.24\times$ and $1.30\times$, respectively).
    The amplification factor for \textsc{PGM-attack} tends to grow with $\varepsilon$.}
    \label{tab:duplicate_allowed_poisoning_results}
    \begin{tabular}{@{}c l r rrr rrr@{}}
        \toprule
         & & & \multicolumn{3}{c}{$\lambda = 0.01\,n$} & \multicolumn{3}{c}{$\lambda = 0.1\,n$} \\
        \cmidrule(lr){4-6} \cmidrule(lr){7-9}
        $\varepsilon$ & Dataset & Original & \textsc{PGM-attack} & \textsc{Random} & \textsc{Random-Adj.} & \textsc{PGM-attack} & \textsc{Random} & \textsc{Random-Adj.} \\
        \midrule
        16 & Weblogs & 291K & 340K {\footnotesize (1.17$\times$)} & 292K {\footnotesize (1.00$\times$)} & 295K {\footnotesize (1.01$\times$)} & 583K {\footnotesize (2.00$\times$)} & 297K {\footnotesize (1.02$\times$)} & 331K {\footnotesize (1.14$\times$)} \\
         & IoT & 178K & 225K {\footnotesize (1.27$\times$)} & 179K {\footnotesize (1.01$\times$)} & 181K {\footnotesize (1.02$\times$)} & 526K {\footnotesize (2.96$\times$)} & 186K {\footnotesize (1.05$\times$)} & 206K {\footnotesize (1.16$\times$)} \\
         & Wiki & 228K & 334K {\footnotesize (1.46$\times$)} & 230K {\footnotesize (1.01$\times$)} & 233K {\footnotesize (1.02$\times$)} & 937K {\footnotesize (4.11$\times$)} & 245K {\footnotesize (1.07$\times$)} & 280K {\footnotesize (1.23$\times$)} \\
         & Zipf & 83.4K & 151K {\footnotesize (1.81$\times$)} & 85.4K {\footnotesize (1.02$\times$)} & 85.5K {\footnotesize (1.03$\times$)} & 668K {\footnotesize (8.01$\times$)} & 104K {\footnotesize (1.24$\times$)} & 105K {\footnotesize (1.26$\times$)} \\
        \midrule
        32 & Weblogs & 122K & 145K {\footnotesize (1.19$\times$)} & 123K {\footnotesize (1.00$\times$)} & 124K {\footnotesize (1.01$\times$)} & 267K {\footnotesize (2.18$\times$)} & 124K {\footnotesize (1.01$\times$)} & 139K {\footnotesize (1.14$\times$)} \\
         & IoT & 78.3K & 101K {\footnotesize (1.29$\times$)} & 78.5K {\footnotesize (1.00$\times$)} & 79.3K {\footnotesize (1.01$\times$)} & 253K {\footnotesize (3.23$\times$)} & 80.6K {\footnotesize (1.03$\times$)} & 89.6K {\footnotesize (1.14$\times$)} \\
         & Wiki & 85.3K & 130K {\footnotesize (1.53$\times$)} & 85.8K {\footnotesize (1.01$\times$)} & 86.8K {\footnotesize (1.02$\times$)} & 414K {\footnotesize (4.86$\times$)} & 89.4K {\footnotesize (1.05$\times$)} & 101K {\footnotesize (1.18$\times$)} \\
         & Zipf & 24.9K & 58.1K {\footnotesize (2.33$\times$)} & 25.5K {\footnotesize (1.02$\times$)} & 25.6K {\footnotesize (1.03$\times$)} & 326K {\footnotesize (13.1$\times$)} & 30.2K {\footnotesize (1.21$\times$)} & 31.8K {\footnotesize (1.28$\times$)} \\
        \midrule
        64 & Weblogs & 49.3K & 59.8K {\footnotesize (1.21$\times$)} & 49.4K {\footnotesize (1.00$\times$)} & 50.1K {\footnotesize (1.01$\times$)} & 119K {\footnotesize (2.41$\times$)} & 49.8K {\footnotesize (1.01$\times$)} & 56.2K {\footnotesize (1.14$\times$)} \\
         & IoT & 36.3K & 47.2K {\footnotesize (1.30$\times$)} & 36.4K {\footnotesize (1.00$\times$)} & 36.8K {\footnotesize (1.01$\times$)} & 123K {\footnotesize (3.39$\times$)} & 36.9K {\footnotesize (1.02$\times$)} & 41K {\footnotesize (1.13$\times$)} \\
         & Wiki & 37.5K & 56.9K {\footnotesize (1.52$\times$)} & 37.7K {\footnotesize (1.01$\times$)} & 38K {\footnotesize (1.01$\times$)} & 194K {\footnotesize (5.16$\times$)} & 38.7K {\footnotesize (1.03$\times$)} & 42.4K {\footnotesize (1.13$\times$)} \\
         & Zipf & 7.24K & 23.7K {\footnotesize (3.27$\times$)} & 7.34K {\footnotesize (1.01$\times$)} & 7.44K {\footnotesize (1.03$\times$)} & 160K {\footnotesize (22.1$\times$)} & 8.58K {\footnotesize (1.19$\times$)} & 9.32K {\footnotesize (1.29$\times$)} \\
        \midrule
        128 & Weblogs & 20K & 25K {\footnotesize (1.25$\times$)} & 20K {\footnotesize (1.00$\times$)} & 20.3K {\footnotesize (1.01$\times$)} & 53.4K {\footnotesize (2.67$\times$)} & 20.2K {\footnotesize (1.01$\times$)} & 22.7K {\footnotesize (1.14$\times$)} \\
         & IoT & 16.3K & 22.2K {\footnotesize (1.37$\times$)} & 16.3K {\footnotesize (1.00$\times$)} & 16.5K {\footnotesize (1.02$\times$)} & 60.4K {\footnotesize (3.72$\times$)} & 16.5K {\footnotesize (1.01$\times$)} & 18.5K {\footnotesize (1.14$\times$)} \\
         & Wiki & 19.2K & 28K {\footnotesize (1.46$\times$)} & 19.3K {\footnotesize (1.01$\times$)} & 19.4K {\footnotesize (1.01$\times$)} & 94.7K {\footnotesize (4.93$\times$)} & 19.5K {\footnotesize (1.02$\times$)} & 21.1K {\footnotesize (1.10$\times$)} \\
         & Zipf & 2.18K & 10.2K {\footnotesize (4.68$\times$)} & 2.2K {\footnotesize (1.01$\times$)} & 2.25K {\footnotesize (1.03$\times$)} & 78.7K {\footnotesize (36.1$\times$)} & 2.52K {\footnotesize (1.16$\times$)} & 2.83K {\footnotesize (1.30$\times$)} \\
        \bottomrule
    \end{tabular}
\end{table*}

At a high level, our algorithm for generating poisons in the duplicate-allowed setting is the same as the algorithm for the duplicate-forbidden setting (\cref{alg:pgm_poisoning}).
The only difference is the function $\textsc{DI-Consecutive}(\mathcal{K}_\mathrm{seg}, \lambda_\mathrm{cand})$ in \cref{alg:pgm_poisoning}, namely, the function that generates $\lambda$ poisons for a given segment $\mathcal{K}_{\mathrm{seg}}$.

In the duplicate-forbidden setting, we used the discrete-intercept algorithm to efficiently evaluate candidate poison sets consisting of consecutive integers.
In the duplicate-allowed setting, we again use the discrete-intercept algorithm, but now we evaluate poison sets consisting of a single integer.
This design is motivated by \cref{thm:duplicate_allowed_structure}, which shows that, in the duplicate-allowed maximum-error maximization problem, an optimal poison set can consist of a single legitimate key.
Motivated by this result, we efficiently evaluate such singleton candidates using the discrete-intercept algorithm and select the poison set that minimizes the number of covered legitimate keys.

Note that, when selecting this single integer and evaluating the poison set, we use \emph{virtual points} rather than raw data points.
Specifically, for each key appearing among the virtual points introduced in \cref{alg:pgm_duplicates_allowed}, we form a poison set consisting of $\lambda$ copies of that key and evaluate it.
As in Appendix~\ref{app:detail_discrete_intercept_consec}, each candidate can be evaluated in $\mathcal{O}(|\mathcal{I}| \log c)$ time, where $c$ is the number of legitimate keys currently coverable by one segment.
Since the number of candidates is $\mathcal{O}(c)$, the total time complexity is $\mathcal{O}(|\mathcal{I}| c \log c)$.
By using virtual points in this way, the poison-generation algorithm partially reflects the behavior of the PGM-index PLA construction algorithm.

\subsection{Experimental Evaluation}

Here, we evaluate how much \textsc{PGM-attack} can increase the number of segments when duplicates are allowed.

\textbf{Datasets.}
To evaluate the effect of \textsc{PGM-attack} in the duplicate-allowed setting, we use three real-world datasets containing duplicates and one synthetic dataset with duplicates:
\begin{itemize}[leftmargin=1.5em]
    \item From SOSD~\cite{sosd-vldb}, we use the \textbf{Wiki} dataset. It contains 200.0M keys in total, of which 90.0M are unique.
    \item Following the PGM-index paper~\cite{ferragina2020pgm} and subsequent work, we use the \textbf{Weblogs} and \textbf{IoT} datasets. However, the MgBench dataset~\cite{mgbench_old} used in that line of work is no longer publicly available. We therefore use the mgbench dataset released by Andrew Crotty~\cite{crotty_mgbench_github}. In this benchmark, bench2 corresponds to \textbf{Weblogs} and bench3 to \textbf{IoT}. They contain 75.7M and 109.0M keys, respectively, and the numbers of unique keys are 18.7M and 108.9M.
    \item As a synthetic dataset, we generate a \textbf{Zipf} dataset. It contains 200.0M keys in total, of which 34.8M are unique.
\end{itemize}

\textbf{Methods.}
We compare \textsc{PGM-attack} with \textsc{Random} and \textsc{Random-Adj.}
Note, however, that these random baselines differ slightly from those in the duplicate-forbidden setting, although the underlying policy is analogous.

For \textsc{Random}, in the duplicate-forbidden setting we sampled uniformly from the allowed integers, i.e., integers in the key range $[k_1, k_n]$ that do not belong to the legitimate keys.
In the duplicate-allowed setting, by contrast, every integer in the key range is allowed, so we simply sample $\lambda$ integers uniformly at random from that key range.

For \textsc{Random-Adj.}, in the duplicate-forbidden setting we sampled uniformly from integers adjacent to legitimate keys.
In the duplicate-allowed setting, since integers already contained in the legitimate keys are also allowed, we instead sample $\lambda$ integers uniformly at random, with replacement, from the legitimate keys.

\textbf{Results.}
\cref{tab:duplicate_allowed_poisoning_results} reports how much \textsc{PGM-attack} increases the number of segments in the duplicate-allowed setting.

For every dataset, \textsc{PGM-attack} increases $m_{\mathrm{opt}}$ relative to the legitimate-key baseline: by up to $4.68\times$ at 1\% poisoning ($\lambda = 0.01\,n$) and up to $36.1\times$ at 10\% poisoning ($\lambda = 0.1\,n$).
In contrast, \textsc{Random} and \textsc{Random-Adj.} barely change $m_{\mathrm{opt}}$, by at most $1.24\times$ and $1.30\times$, respectively.
These results show that, even when duplicates are allowed, \textsc{PGM-attack} yields larger increases in $m_{\mathrm{opt}}$ than random insertions.

We also observe, as in the duplicate-forbidden setting, that the amplification factor of \textsc{PGM-attack} tends to grow with $\varepsilon$.
For example, when $\lambda = 0.1\,n$, the amplification factor is at most $8.01\times$ for $\varepsilon = 16$, but rises to $36.1\times$ for $\varepsilon = 128$.
As discussed in \cref{sec:experiment_poisoning_segment_number}, this trend can be understood intuitively in the same way as in the duplicate-forbidden setting.

\end{document}